\documentclass[12pt,letterpaper]{article}
\usepackage[top=1in,bottom=1in,left=1in,right=1in]{geometry}
\usepackage{amsthm}
\usepackage{amsmath}
\usepackage{amssymb}
\usepackage[all]{xy}
\usepackage{sgame}
\usepackage{graphicx}
\usepackage{caption}
\usepackage{subcaption}
\usepackage{placeins}
\usepackage{xcolor}
\usepackage{authblk}
\usepackage{bm}
\usepackage{vhistory}
\usepackage{natbib}
\bibpunct{(}{)}{,}{a}{}{;}

\graphicspath{{plots_arxiv/}}

\date{}
\makeatletter \renewenvironment{proof}[1][\proofname]
{\par\pushQED{\qed}\normalfont\topsep6\p@\@plus6\p@\relax\trivlist\item[\hskip\labelsep\bfseries#1\@addpunct{.}]\ignorespaces}{\popQED\endtrivlist\@endpefalse} \makeatother

\theoremstyle{definition}   
\newtheorem{thm}{Theorem}
\newtheorem{prop}{Proposition}
\newtheorem{lem}{Lemma}
\newtheorem{defn}{Definition}

\newtheorem{rmk}{Remark}
\newtheorem{cor}{Corollary}
\newcommand{\vs}{{\mathbf{s}}}
\newcommand{\hvs}{{\widehat{\mathbf{s}}}}
\newcommand{\liminfty}{{\underset{\bm{\sigma} \rightarrow +\infty}{\lim}}}
\newcommand{\alphaes}{{ \alpha_e^{\mathbb{S}}}}
\newcommand{\gammaesu}{{ \overline{\gamma_e}^{\mathbb{S}}}}
\newcommand{\gammaesl}{{ \underline{\gamma_e}^{\mathbb{S}}}}
\newcommand{\gammaeu}{{ \overline{\gamma_e}}}
\newcommand{\gammael}{{ \underline{\gamma_e}}}

\begin{document}

\date{This version: October, 2016}
\title{Understanding the Impacts of Dark Pools on Price Discovery}
\author{Linlin Ye \footnote{E-mail: linlinye@cuhk.edu.cn, phone:+86-755-84273420. I am very grateful to Pierre-Oliver Weill, Avanidhar Subrahmanyam, Tomasz Sadzik, Mark Garmaise, Haoxiang Zhu, Mao Ye, Joseph Ostroy, John Wiley, Obara Ichiro, Simon Board, Ivo Welch,  Antonio Bernardo, Francis Longstaff, Bernard Herskovic, William, Mann, Daniel Andrei, Shuyang Sheng, and Jernej Copic for valuable comments and discussions.  I would also like to thank the seminar participants at the University of California Los Angeles Economics Department, UCLA Anderson School, Peking University HSBC, Bologna Business School, the Chinese University of Hong Kong, the European Finance and Banking conferences. I also thank Alex Kemmsies in Rosenblatt Rosenblatt Securities and Sayena Mostowfi in TABB Group for providing the industry data.} }
\affil{The Chinese University of Hong Kong, Shenzhen}
\maketitle


\begin{abstract}\footnotesize

         This paper investigates the impact of dark pools on price discovery (the efficiency of prices on stock exchanges to aggregate information). Assets are traded in either an exchange or a dark pool, with the dark pool offering better prices but lower execution rates. Informed traders receive noisy and heterogeneous signals about an asset's fundamental. We find that informed traders use dark pools to mitigate their information risk and there is a \textit{sorting effect}: in equilibrium, traders with strong signals trade in exchanges, traders with moderate signals trade in dark pools, and traders with weak signals do not trade.  As a result, dark pools have an \textit{amplification effect} on price discovery. That is, when information precision is high (information risk is low), the majority of informed traders trade in the exchange hence adding a dark pool enhances price discovery, whereas when information precision is low (information risk is high), the majority of the informed traders trade in the dark pool hence adding a dark pool impairs price discovery.  The paper reconciles the conflicting empirical evidence and produces novel empirical predictions. The paper also provides regulatory suggestions with dark pools on current equity markets and in emerging markets.
\end{abstract}

\section{Introduction}
Over the years, the world financial system has experienced a widening of equity trading venues, among which dark pools have rapidly grown in popularity. The market share of dark pools in the US has grown from 7.51\% in 2008 to 16.57\% in 2015.\footnote{Rosenblatt Securities: \textit{Let There Be Light}, January 2016 Issue.} In contrast with a traditional stock exchange, dark pools do not publicize information about their orders and price quotations before trade. Unlike a stock exchange in which prices are formed to clear the buy and sell orders, a typical dark pool does not form such prices: it executes orders using prices derived from the stock exchanges. Those dark pools do not contribute to the process of information aggregation in the exchange, and hence they do not offer price discovery. Price discovery (i.e., the process and efficiency of prices aggregating information about assets' values) is essential to achieving the confidence of a broad community of market participants and ensuring the efficiency of capital markets. Therefore, the question of whether dark pool trading will harm price discovery has become a rising concern and matter of debate for regulators and industry practitioners.\footnote{For example, as remarked by the SEC Commissioner Kara M. Stein before the Securities Traders Association's 82nd Annual Market Structure Conference in Sep. 2015, ``As more and more trading is routed to dark venues that have restricted access and limited reporting, I am concerned that overall market price discovery may be distorted rather than enhanced.'' According to ``An objective look at high-frequency trading and dark pools,'' a report released by PricewaterhouseCoopers (2015), ``Dark pools may harm the overall price discovery process, particularly in a security in which a significant portion of that security's trade volume is in the pools.''} Academic research, for its part, has yielded conflicting results. \citet*{ye_glimpse_2011} predicts that, in theoretical studies, the addition of a dark pool strictly harms price discovery. By contrast, \citet*{zhu_dark_2013} predicts that dark pools strictly improve price discovery. Empirically, there are findings that support each of the different predictions.

This paper investigates the question whether dark pool trading will harm price discovery. In the model, there are informed speculators and uninformed liquidity traders. More specifically, informed traders have heterogeneous private signals, with the distribution of these signals determined by an \textit{information precision} level. Uninformed liquidity traders have heterogeneous demands for liquidity. Both types of traders choose among three options: a) trade in an exchange, b) trade in a dark pool, or c) do not trade (delay trade).  The exchange is modeled as market makers posting bid-ask prices and guaranteeing execution, whereas the dark pool is modeled as a crossing-mechanism that uses the average of bid and ask (mid-price) in the exchange to execute orders (if there are more buy orders than sell orders, buy orders are executed probabilistically, with some buy orders not executed, and vice versa).

We find a novel \textit{amplification effect} of dark pools on price discovery: price discovery in the exchange will be enhanced when traders' information precision is high and will be impaired when traders' information precision is low. The results help to reconcile the seemingly contradictory empirical findings about dark pool impact on the market and generate novel empirical predictions regarding the information content of dark pool trades, dark pool market share, and their relationships with exchange spread. We identify that information structure (information precision) is one key variable in determining the informational efficiency (price discovery) when markets are fragmented by dark pools. We show that the results have immediate policy implications for enhancing price discovery in equity markets and dark pool usage in emerging economies. We also provide a discussion regarding the possible measures of markets' information precision.

The intuition of the amplification effect is as follows. First, we show that, in equilibrium, there is a \textit{sorting effect}: for informed traders, those with strong signals trade in the exchange, those with modest signals trade in the dark pool, and those with weak signals do not trade. For uninformed liquidity traders, those with high liquidity demand trade in the exchange, those with modest liquidity demand trade in the dark pool, and those with low liquidity demand delay trade. The sorting effect is derived from the trade-off of trading dark pools: dark pools provide better prices than exchanges, but this is offset by a higher non-execution probability. Therefore, amongst informed traders, those with strong signals prefer an exchange because they are very confident about making profits and desire a guaranteed execution more than a better price; those with moderate signals prefer a dark pool because they are less confident about making profits and desire a better price more than execution; and finally, those with weak signals prefer not to trade because they are unconfident about making profits. A similar argument holds for liquidity traders.

Second, we show that the amplification effect holds as a result of the sorting effect. Since different information precision levels result in different distributions in the strengths of signals and hence different venue choices for the majority of the informed traders, they cause different dark pool impacts on price discovery. When information precision is high, the majority of informed traders receive strong signals and prefer an exchange. Therefore, adding a dark pool attracts only a small fraction of informed traders, compared with the liquidity traders, leaving a higher informed-to-uninformed ratio (i.e., relative ratio of informed and uninformed traders) in the exchange and hence improving price discovery in the exchange. In contrast, when information precision is low, the majority of informed traders receive modest signals and prefer a dark pool. Therefore a dark pool would attracts a higher fraction of informed traders, compared with the liquidity traders, leaving a lower informed-to-uninformed ratio in the exchange and hence impairing price discovery in the exchange.

This paper points out an important function of dark pools not yet discussed in the existing literature: dark pools help informed traders mitigate their information risk, that is, the loss that is attributable to wrong information.  When traders' information is relatively weak (meaning there is a higher probability that it is wrong), they face a high risk of losing money in trading. Dark pools provide those traders a perfect ``buffer zone'' -- a place that strictly lowers their information risk. This function of dark pools is only present, however, when traders have a noisy information structure.

To the best of our knowledge, this paper is the first to introduce a noisy information structure in a fragmented market to study dark pools and price discovery. Examining the noisiness in information is of essential importance, not only because it is much more realistic than assuming perfect information, but also because it reveals the process of price discovery by identifying the motivations of traders' choices. As a result, our predictions are more robust in the sense that the sorting and amplification effects hold in every equilibrium. In contrast, the current theoretical literature assumes that all informed traders have perfectly precise information. This obscures trading motivations and induces instability in the results.  For example, \citet*{zhu_dark_2013} studies some equilibria in which dark pools improve price discovery, but there may exist other equilibria in his model in which dark pools harm discovery. Yet, \citet*{zhu_dark_2013} does not discuss these equilibria. 

Our findings have immediate policy implications for the ongoing debate over dark pool usage. Our findings imply that, in contrast with current literature, there is no uniform impact that dark pools have on price discovery and other measures of market quality.  Dark pool activity and its impacts display significant cross-sectional variation and should be evaluated differently across various economic environments. Concrete suggestions for regulators to enhance pricing efficiency include:  (i) identifying firm characteristics and monitoring dark pool trades in firms that are likely to have a negative dark pool impact, such as high R\&D firms, young firms, small firms, and less analyzed firms, (ii) facilitating information transmission and processing, enhancing accounting and reporting disclosure systems, and improving the efficiency of the judicial systems and law enforcement against insider trading, and (iii) being cautious in emerging markets with regards to dark pool trading, given that most emerging markets are regulated by poor legal systems that lack implemental power against insider trading and have a low precision in information disclosure. A more detailed discussion is provided in Section \ref{sec_policy}.

Our study also produces testable predictions and helps to reconcile the seemingly contradictory results in the current empirical literature. One of the predictions that could motivate empirical and regulatory concerns is how much dark pool trades can forecast price movements.  We predict that the information content of  dark pool trades has an inverted U-shape relationship with the liquidity level (exchange spread), implying that assets with modest liquidity have the highest information content in their dark pool trades, whereas the most liquid and illiquid assets have the lowest information content in their dark pool trades. There are also some predictions which coincide with current theoretical literature. For example, dark pool usage also has an inverted U-shape association with exchange spread. Dark pools create additional liquidity for the market.  A more detailed discussion is in Section \ref{sec_testables}.

\textbf{Related Work:} 
There is a large collection of studies that examines information asymmetry and price discovery in financial markets, in both the theoretical and empirical fields. In theoretical studies, a large set of papers analyze non-fragmented markets, including the two pioneering works in price discovery, \citet*{glosten_bid_1985}, and  \citet*{kyle_continuous_1985}. Other studies examine fragmented lit markets, for example \citet*{viswanathan_market_2002}, \citet*{chowdhry_multimarket_1991}, and \citet*{hasbrouck_one_1995}. There are a handful of papers that study information asymmetry in a market fragmented by lit and dark venues \citep[see, e.g.,][]{hendershott_crossing_2000, degryse_dynamic_2009, buti_dark_2011}. Yet, these models assume either non-freedom of choice for traders or exogenous prices.  Our study, on the other hand, considers free venue selection for traders and endogenous prices. This paper is closely related to \citet*{zhu_dark_2013} whose trading protocols are the same as ours. But unlike \citet*{zhu_dark_2013} who considers an exact information structure, we examine a noisy information structure. Under this noisy information structure, we predict different results in price discovery and other measures from \citet*{zhu_dark_2013}. When the informational noise is absent in our model (i.e., information noise converges to zero), our prediction of price discovery coincides with \citet*{zhu_dark_2013}'s. Our paper is also related but divergent from \citet*{ye_glimpse_2011}. Whereas our model considers free selection of traders, \citet*{ye_glimpse_2011} assumes that uninformed traders are not subject to free-choice between different venues, and hence the corresponding piece of the pricing mechanism is missing. In our model, if we fix the choices of uninformed traders and only allow informed traders to choose between venues, our prediction also coincides with \citet*{ye_glimpse_2011}.

Empirical works report conflicting results regarding dark pool impact on price discovery. These results are within the predictions of our study. For example, \citet*{buti_diving_2011}, \citet*{jiang_market_2012}, and \citet*{fleming_order_2013} support an improvement for price discovery with dark trading, while \citet*{hatheway_empirical_2013}, and \citet*{weaver_trade-at_2014} discover a diminishment in price informativeness. Also, \citet*{jones_island_2005} find a negative impact for dark trading on price discovery, while \citet*{comerton-forde_dark_2015-1} find that, cross-sectionally, dark pool trading improves price discovery when the proportion of non-block dark trades are low (below 10\%, suggesting a low fraction of informational content) and harms price discovery when the proportion of non-block dark trades is high.

There are also other empirical studies that focus on dark pool operation and other measures of market quality. Some papers analyze the information content of dark pool trades. For example \citet*{peretti_is_2014} find that dark trades can predict price movement. Some study the trade-offs of dark trading. For example, \citet*{gresse_effect_2006}, \citet*{conrad_institutional_2003}, \citet*{naes_equity_2005} , and \citet*{ye_non-execution_2010} study the execution probability in dark pools.  Another category studies the association between dark trading and the exchange spread. My study predicts the same inverted U-shape as \citet*{ray_match_2010} and \citet*{preece_dark_2012}. My study also suggests a cross-sectional variance and provides insights in explaining the contradictory results reported in other papers. For example  \citet*{asic_dark_2013}, \citet*{comerton-forde_dark_2015-1}, \citet*{degryse_impact_2011}, \citet*{hatheway_empirical_2013}, and \citet*{weaver_trade-at_2014} find a positive association while \citet*{ohara_is_2011} and \citet*{ready_determinants_2014-1} find a negative association between dark pool market share and exchange spread. Others find cross-sectional differences \citep[see, e.g.,][]{nimalendran_informational_2014, buti_diving_2011}. A more detailed discussion of the relationship between our predictions and the current empirical literature is provided in Section \ref{sec_testables}.

\section{Dark Pools: An Overview}
\label{sec_dp_overview}
Over the last decade, numerous trading platforms have emerged to compete with the incumbent exchanges. Today, in the U.S. investors can trade equities in approximately 300 different venues. According to TABB (Oct. 2015),\footnote{ ``US Equity Market Structure: Q2-2015 TABB Equity Digest,'' TABB Group, Oct. 2015.} as of June 2015, there are 11 exchanges, 40 active dark pools, a handful of ECNs, and numerous broker-dealer platforms that are operating as equity trading venues in the U.S. \footnote{In Europe, according to \citet*{gomber_mifid_2010} there are around 32 dark pools operating in
 equity markets. In Australia, from \citet*{asic_dark_2013}, there are 20 dark trading venues operating. }.

Among those venues, dark pools are a type of equity trading venue that does not publicly disseminate the information about their orders, best price quotations, and identities of trading parties before and during the execution.\footnote{Although the information about orders are hidden before trade, the after executed trades are not: executed trades are recorded to the consolidate tape right after the trade. SEC requires reporting of OTC trades in equity securities within 30 seconds of execution. Also, dark pools are required to report weekly aggregate volume information on a security-by-security basis to FINRA.}\footnote{SEC Reg NMS Rule 301 (b) (3) requires all alternative trading systems (ATSs) that execute more than 5\% of the volume in a stock to publish its best-priced orders to the consolidated quote system. However, it only applies if the ATS distributes its orders to more than one participant. If it does not provide information about its orders to any participants, it is exempt from the quote rule.}\footnote{Electronic Communication Networks (ECNs) are registered as a type of ATS. But unlike dark pools, ECNs display orders in the consolidated quote stream.} The term ``dark'' is so named for this lack of transparency. Dark pools emerged as early as the 1970s as private phone-based networks between buy-side traders (See \citet*{degryse_dark_2013}). In the early days, the success of these trading venues was limited, but this has changed substantially in the last decade. Dark pools have experienced a rapid growth of trading activity in the U.S., Europe and Asia-Pacific area. Figure \ref{fig_countryvolume} shows the annual data on the market share of dark pool trading as of the consolidated volume in the U.S., Europe, and Canada, updated to 2015.\footnote{We estimated Canadian dark pool market share from ``Report of Market Share by Marketplace--Historical (2007-2014),'' IIROC, Aug 2015, ``Report of Market Share by Marketplace (historical 2015--Present),'' IIROC, May, 2016. Precisely, we estimate the market share of the following 4 dark pools operating in Canada: Liquidnet, Matchnow, Instinet, and SigmaX Canada. } According to the data, the U.S. market share of dark pools increased from about 7.51\% in 2008 to 16.57\% in 2015. The dark pool market shares in Europe and Canada are less, but they exhibit the same growth trend. In Australia, according to the Australian Securities \& Investments Commission \citep{asic_dark_2013}, as of June 2015 dark liquidity consists of 26.2\% of total value that traded in Australian equity market.\footnote{Australian Securities \& Investments Commission, ``Equidity Market Data,'' June 2015.  The number contains 12\%  block size dark liquidity and 14.1\% non-block size dark liquidity. It describes all the hidden orders in the markets including those in exchanges and dark pools.}

\begin{figure}[h]
\centering
\includegraphics[width=.85\textwidth]{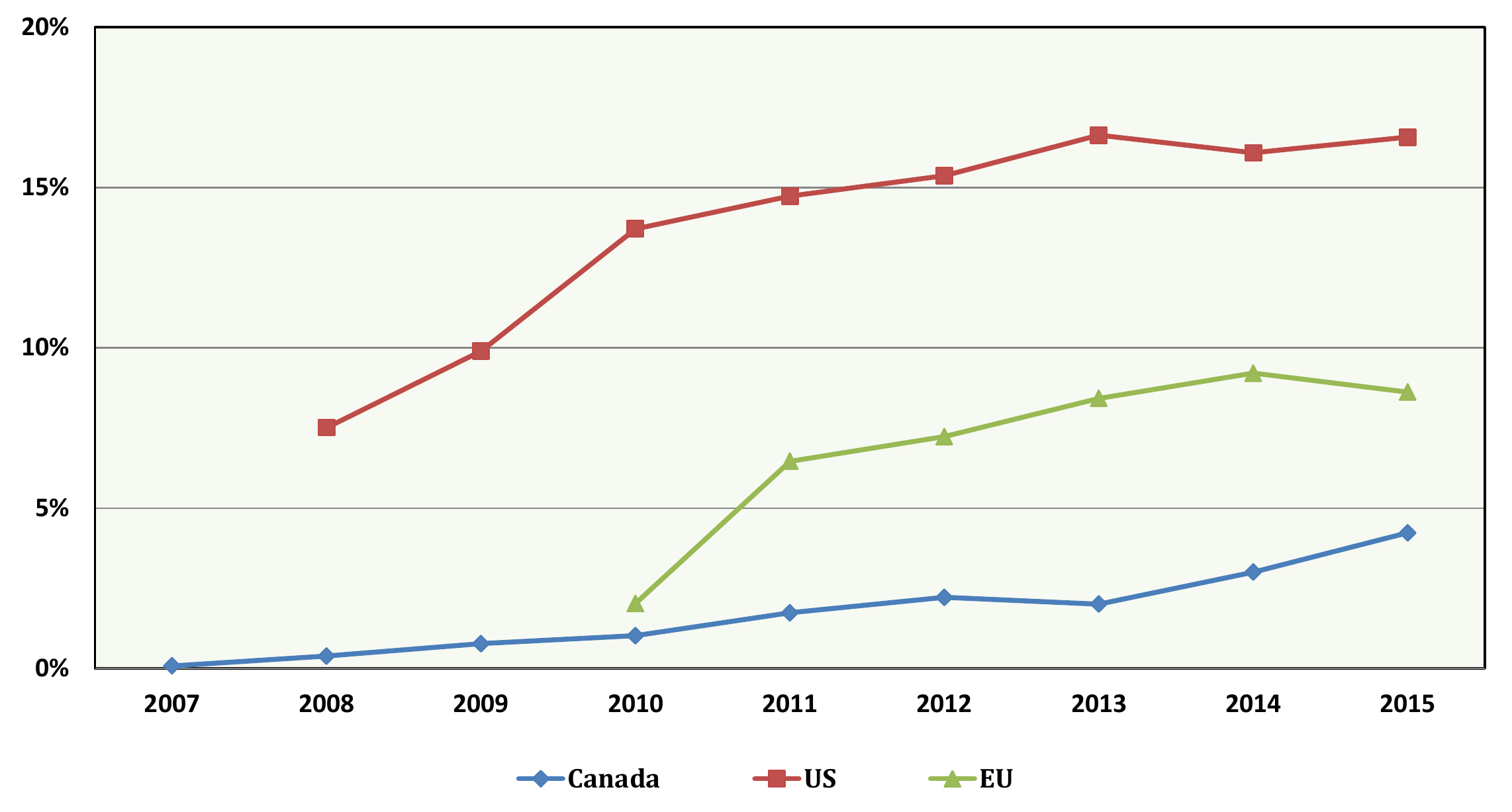}
\caption{\footnotesize{\textbf{Dark Pool Market Share}. The plot shows the annual data of dark pool volume as a percentage of the total consolidated volume in the US, Europe, and Canada. \\
\textbf{Data source}: US data (2008 - 2015) is from ``Rosenblatt Securities: \textit{Let There Be Light}, January 2016 Issue'' and Europe data (2010 - 2015) is from ``Rosenblatt Securities: \textit{Let There Be Light -- European Edition}, January 2016 Issue''. Figures in Canada (2007 - 2015) are derived from reports of IIROC.}}
\label{fig_countryvolume}
\end{figure}

One reason behind the rapid growth of dark pool trading is the technology development in electronic trading algorithms. Advances in technology have made it easier to automatically optimize routing and execution according to different sets of considerations and trading protocols. Another reason for the proliferation is the regulation changes that have been made to encourage competition between trading venues. For example, in the U.S., Regulation NMS (National Market System) was revised and reformed in 2005 to encourage the operation of various platforms, and as a consequence, a wide variety of trading centers have been established since then. Another example is the introduction of the Market in Financial Instruments Directive (MiFID) in the European Union in 2007, which spurred the creation of new trading venues, including dark pools.\footnote{In recent years, however, as the debate about dark pool usage has escalated, many countries have started to consider restrictions on dark trading. For example, Canada and Australia have required dark liquidity to provide a ``meaningful price improvement'' of at least one trading increment (i.e., one cent in most major markets), and US regulators have also been contemplating imposing such restrictions. In recent years, US regulators start to strengthen law enforcement against dark pools and urged their upgrading in operation. These cases include UBS Securities (Jan 2015), Goldman Sachs Execution \& Clearing, L.P (SIGMA X, July 2015), and Barclays (Jan 2015).}

There are two key commonalities in dark pools' operating protocols: the pricing mechanism and execution mechanism. First, dark pools generally do not provide price discovery. Instead, they typically use a price derived from an existing primary market as their transaction price. The most commonly used pricing mechanism is the mid-point mechanism: a pricing method to cross orders at the concurrent mid-point of the National Best Bid and Offer (NBBO).\footnote{\citet*{nimalendran_informational_2014} document the usage of such a pricing mechanism in their dark trading sample and find that not all trades are at the midpoint of NBBO, but about 57\% transactions are within .01\% of the price around the midpoint. In this paper, we follow the majority and adopt the mid-point pricing mechanism. } Second, unlike exchanges where orders are cleared at the exchange price, in most of the dark pools, orders don't clear.  Instead, dark pools adopt a ``rationing'' mechanism to execute orders. That is, traders anonymously place unpriced orders to the pool, and the orders are matched and executed probabilistically -- orders in the shorter side are executed for sure, whereas orders in the longer side are rationed probabilistically.

The pricing and execution mechanisms of dark pools' operation reflect the trade-off of trading in a dark pool for an individual trader. On the one hand, dark pools have lower transaction costs than exchanges (typically because orders are executed within the NBBO, with the ``trade-at rule'' further enhancing such price improvement), and lessen the price impact for big orders. On the other hand, investors suffer a lower execution rate compared with the exchange. \citet*{gresse_effect_2006} found that the execution probability in the two dark pools in his dataset was only 2-4 percent, while \citet*{ye_non-execution_2010} documents a dark pool execution probability of 4.11\% (NYSE listed) and 2.17\% (NASDAQ listed) in his dataset, in comparison with a probability of 31.47\% and 26.48\% for their exchange counterparts.\footnote{Nowadays, a rising concern of dark pools is their vulnerability to predatory trading by High Frequency Traders (HFTs) (See \citet*{mittal_are_2008}, \citet*{nimalendran_informational_2014}, \citet*{asic_dark_2013} for instance.)}.

The dark pools' participating constituent base has evolved over time. In the early years, dark pools were designed as venues where large, uninformed traders transact blocks of shares to reduce price impact. This is possible because dark pools are not subject to NMS fair access requirements and can thus prohibit or limit access to their services (see Reg ATS Rule 301(b)(5)). In recent years, however, this has changed greatly. According to an industry insider in Rosenblatt Securities Inc., ``it can be assumed that most pools are open to most investors connecting to the pool, provided the investors do not violate any codes of conduct.'' A measure of such a change is reflected in the trading sizes of dark pools.  Figure  \ref{fig_ussize} shows the average trading size in the U.S. According to the data, the US average trading size in dark pools and exchanges (NYSE and NASDAQ) have been started to converge since 2011, highlighting the fact that the participating constituents in these venues have become more and more similar. It implies that the exclusivity of a dark pool to informed traders has been weakened . As a result, more prominence has been attached to the issue of the potential impact of dark pools on price discovery, because as more informed traders obtain access to dark pools, their migration to dark pools may hurt the information aggregation process in the exchange,\footnote{This paper, as well as \citet*{zhu_dark_2013} considers full access for informed traders.}.
\begin{figure}[h]
\centering
\includegraphics[width=.85\textwidth]{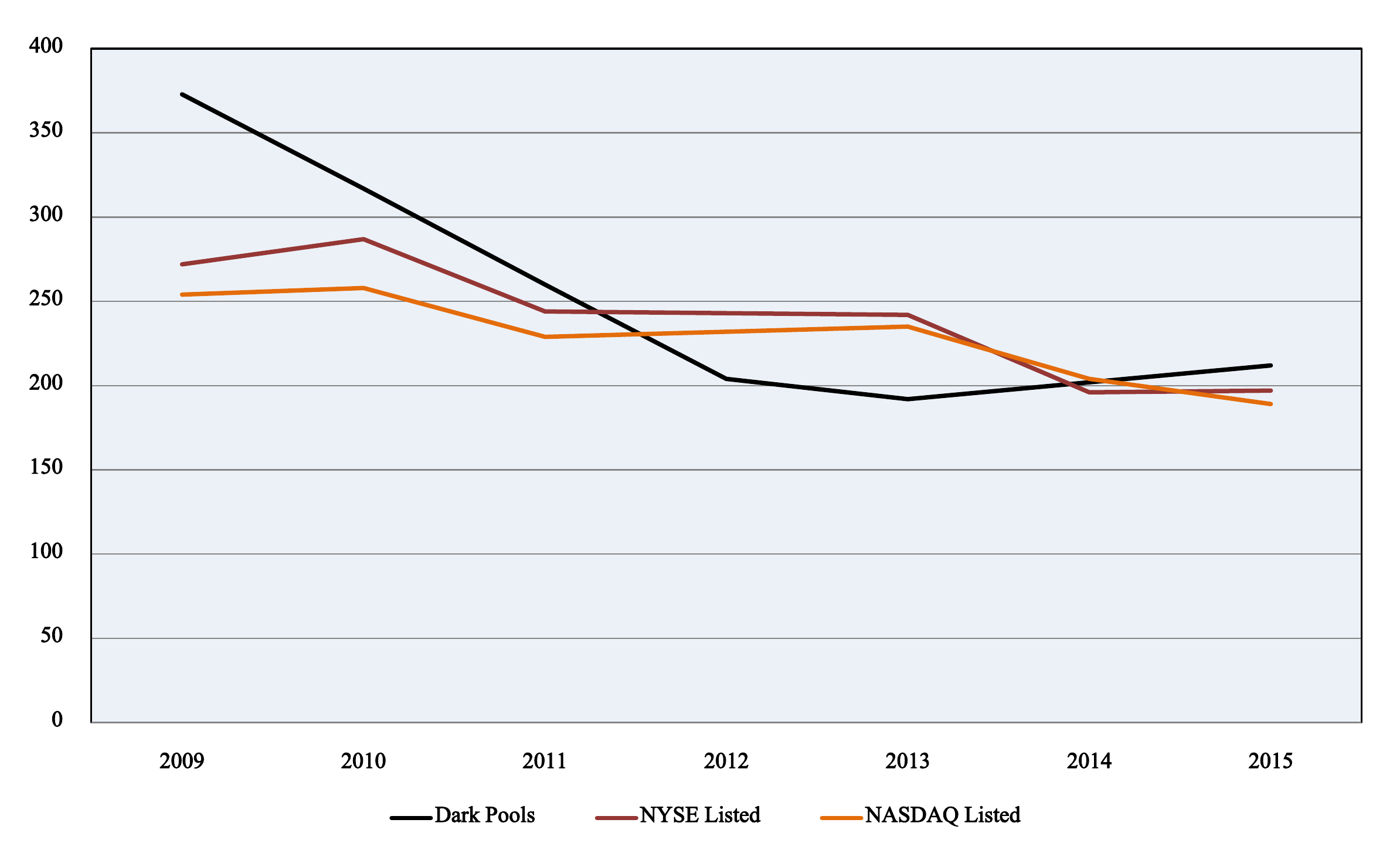}
\caption{\footnotesize{\textbf{Average Trade Size}. The plot shows the annual average trade size of US dark pools, NYSE and NASDAQ, from 2009 to 2015. \\\textbf{Data source}: Rosenblatt Securities.}}
\label{fig_ussize}
\end{figure}

Dark pools are heterogeneous. The types of dark pools can be classified according to different characteristics based on their ownership structure, pricing access, operation mechanism, constituency and other factors. All of these categories are in constant flux for the dark pools. Most of the pools also overlap in one or more categories as well, only the owner types remain constant overtime. We provide a discussion on some characteristics and their examples.

\textbf{(i) Pricing. } Dark pools use three primary pricing mechanisms. The execution will take place once two sides of a suitable trade are matched. The three pricing mechanisms are automatic pricing (usually at the midpoint of the best bid and offer), derived pricing (for example, average price during the last five minutes), and negotiated pricing (for example, Liquidnet Negotiatoin offers availability of one-to-one negotiation of price and size).

\textbf{(ii) Order Type.} There are primarily three types of order that prevails in dark pools: limit orders (to buy or to sell a security at a desired price or better), peg orders (peg to the NBBO, for example midpoint or alternate midpoint,\footnote{Traders are able to specify premiums or discounts vis-$\grave{a}$-vis the mid when placing a trade. For example, a motivated buyer may specify an order that promises to pay the mid plus a penny. This would give this trade priority over all other buy orders.}) and immediate or cancel order (IOC). A dark pool may accept a subset of these order types. Pools that accept limit orders may offer some price discovery (usually within the NBBO). These pools include, for example, Credit Suisse's CrossFinder, Goldman Sachs' Sigma X, Citi's  Citi Cross, and Morgan Stanley's MS Pool. Pools that execute peg orders do not provide price discovery. These include, for example, Instinet, Liquidnet, and ITG Posit. Pools accepting only IOC orders are single dealer platforms (SDP), where the operator works as market makers and customers interact solely with the operator's own desk (for example, Citadel Connect and Knight Link by KCG\footnote{Getco LLC once operated an SDP called GetMatched.  Following the 2013 merger of Knight Capital Group and Getco LLC, GetMatched was decommissioned.}).

\textbf{(iii) Execution Frequency and Order Information.} There are three modes of execution: scheduled crossing, continuous blind crossing, and indicated market.\footnote{See \citet*{decovny_dark_2008}.} The scheduled crossing networks include BIDS, ITG POSIT Match, and Instinet US Crossing.  In scheduled crossing networks, the two sides of a trade cross during a set period. These networks typically do not display quotes but may have an order imbalance indicator. Continuous blind crossing networks continuously cross orders for which no quotes are given. Indicated markets cross orders using participants' indications of interest (IOIs) and provide some level of transparency in order to attract liquidity. Liquidnet and Merrill Lynch offer variations on this theme.\footnote{Pipeline, a well-known dark pool using IOIs, settled allegations that it misled customers and was shut in May 2012.}

\textbf{(iv) Customer Base and Exclusivity.} There are dark pools which design their rules and monitor trading in an attempt to limit access to buy-side (natural contra-side) institutional investors. According to \citet*{boni_dark_2013}, Liquidnet ``Classic'' is one of those. A measure of the exclusivity is the average trading size of a dark pool. In May 2015, among the 40 active dark pools operating in the US, there are 5 dark pools in which over 50\% of their Average Daily Volumes are block volume (larger than 10k per trade). Those pools can be regarded as ``Institutional dark pools,'' and they include Liquidinet Negotiated, Barclays Directx, Citi Liquifi, Liquidnet H20, Instinet VWAP Cross, and BIDS Trading. Other dark pools have  percentages of block volumes less than 15\%, with most of them lower than 2\%.\footnote{``Let There Be Light , Jun 2015,''  Rosenblatt Securities, Inc. }

\textbf{(v) Ownership Structure.} According to Rosenblatt (2015), dark pools can be classified into four categories according to their ownership structure. This is the only classification that does not fluctuate over time. The four categories include the Bulge Bracket/Investment bank, Independent agency, Market maker, and Consortium-sponsored. In May 2015, The market shares of the four categories are, respectively, 55.28\%, 24.11\%, 13.79\%, and 6.82\%. Examples of the Bulge Bracket/Investment bank-owned dark pools are CS Crossfinder, UBS ATS, DB SuperX, and MS Pool. Independent agency owned pools include, for example,  ITG POSIT, Instinet CBX, ConvergEx Millennium. Market maker owned pools include Citadel Connect and Knight Link by KCG, and Consortium-sponsored pools include Level and BIDS. \footnote{``Let There Be Light , Jun 2015,''  Rosenblatt Securities, Inc. }

Finally, ``dark pools liquidity'' is not equivalent to ``dark liquidity.'' Dark liquidity, or dark volume, is a broader concept since it measures the total non-displayed market volume. Exchanges, for example, can contain ``dark'' volumes, which are applied through iceberg orders and workup processes. According to the TABB group' classification, dark volume can break down into retail-wholesaler, dark pool volume, and hidden exchange volume. As of Q2-2015, the percentages of each are 40.1\%, 39.7\%, and 20.2\% respectively. In total, the dark volume was 43.9\% of the consolidated volume.\footnote{``US Equity Market Structure: Q2-2015 TABB Equity Digest,'' TABB Group, Oct. 2015.}

\section{The Model}
\label{section_themodel}
The model considers an economy that lasts for three periods. We index the periods by 0, 1, and 2. There is one risky asset that is traded during the two periods with an uncertain fundamental value $$\tilde{v}=\left\{\begin{array}{cc}  -\sigma_v, &\mbox{ with probability } {1\over 2},   \\ \sigma_v,  &\mbox{ with probability } {1\over 2}. \end{array} \right. $$ That is to say, the risky asset has an unconditional mean zero and standard deviation $\sigma_v$.  In period 0, $\tilde{v}$ is realized, but this information is not revealed to the public.

There are two types of traders who are potentially interested in the risky asset: informed speculators and uninformed liquidity traders. We assume that they are all risk-neutral. There is a continuum of informed speculators with measure $\mu$, a continuum of uninformed liquidity buyers with measure $Z^+$, and a continuum of liquidity sellers with measure $Z^-$. We assume that $Z^+, Z^-$ are identical and continuously distributed random variables on $[0,+\infty)$, with mean $\frac{1}{2}\mu_z$. $Z^+, Z^-$ are also realized at period 0 so that liquidity buyers and liquidity sellers arrive at the market at the same time. The realizations of $Z^+, Z^-$ are not observed by any market participants.

In period 0, each informed speculator receives his or her own private signal regarding the value of the asset, $s_i=\tilde{v}+e_i$, where $i$ is the index of informed traders and $e_i$ represents the noise of the signal.\footnote{According to \citet*{gyntelberg_private_2009-1}, there are various types of private information that stock market investors may have about the fundamental determinants of a firm's value, including knowledge of the firm's products and  innovation prospect, management quality, and the strength and likely strategies of the firm's competitors. Private information may also include passively collected information about macro-variables and other fundamentals which may be dispersed among customers. Equity market order flow to a large degree reflects transactions by investors who are very active in collecting private information. A more detailed discussion is in section \ref{sec_policy}. } We assume that $e_i$ are identically independently distributed normal random variables, with mean 0 and standard deviation $\sigma_e$. Therefore, in the first period, they trade on both their private information and public information (if there is any). They can trade (either buy or sell) up to 1 unit of the asset. If there are more than one venue to trade, they can split their orders.  Without loss of generality, we assume that informed speculators only trade in period 1.\footnote{In period 2 when the informed traders' private information becomes public, they lost their information advantage. Since the informed agents are risk neutral and they only enter the market for profit, they will not actively place orders in the second period.}  The model is distinctive to \citet*{zhu_dark_2013} in the information structure. \citet*{zhu_dark_2013} assumes that all informed traders receive exact signals about the asset, whereas we consider a noisy information structure.\footnote{We do not consider information acquisition cost because it is modeled as a sunk cost in this paper. }  The introduction of a richer information structure is crucial to our analysis, not only because it is more realistic, but also because it reveals a sorting effect of market fragmentation on information. That is, in equilibrium, traders with strong signals trade in the exchange, traders with modest signals trade in the dark pool, and traders with weak signals do not trade. This sorting effect is the major economic force in the trader's venue-selection and the process of price discovery. The absence of such an effect will likely cause instability of predictions in multiple equilibra, such as discussed in \citet*{zhu_dark_2013}. A more detailed discussion is in Section \ref{sec_eqm_multi}.

A liquidity buyer (seller) comes to the market to buy (sell) 1 unit of the risky asset. Similarly, one can split their orders if there exist multiple transaction venues. The uninformed liquidity traders, however, do not have any private information. They enter the market to meet their liquidity demands. The level of their liquidity demand is measured by a delay cost, a cost that reflect how urgent one needs his or her order to be fulfilled in period 1. More precisely, if a liquidity trader, buyer or seller, cannot have his or her order executed in period 1, a delay cost is incurred. The delay cost (per unit) is represented by $\sigma_v  d_j $, where $j$ is the index for the liquidity traders. $d_j s$ are i.i.d random variables with a Cumulative Distribution Function $G(x): [0,\bar{d}]\rightarrow [0,1]$, where $G(x)\in C^2$, $1\leq\bar{d}<\infty$ and $G'(x)+xG''(x)\geq 0, \forall x \in [0,1]$. \footnote{This additional assumption is for the uniqueness of the equilibrium. It is satisfied by many commonly used distributions. For example, a uniform distribution. } Again, $d_j s$ are realized at period 0.

There are two venues for traders to trade: an exchange (the Lit market) and a dark pool.  We will then consider a benchmark model where there is only one trading venue for the agents -- the exchange only.  By comparing our model with the benchmark model, we are able to study the impact of a dark pool to the public exchange, and the interaction between the two venues. We now specify the transaction rules in the two venues and the problems of each type of traders.

Finally, the distributions of $\tilde{v},Z^+, Z^-, \{e_i\}, \{d_j\}$ are all publicly known information.

\subsection{Transaction rules in the exchange (Lit market)}
A lit market is an exchange for the asset. The exchange is modeled in the spirit of \citet*{glosten_bid_1985}. Precisely, in the lit market, there is a risk neutral market maker who facilitate transactions. The objective of the market maker is to balance his or her budget. The market maker has no private information. Therefore, at period 0, the market maker announces a bid and an ask price for the risky asset, based only on public information. The announced bid and ask price will be the prices for any order submitted to the exchange in period 1, and will be committed by the market maker. Because of symmetry of $\tilde{v}$ and the fact that the unconditional mean of $\tilde{v}$ is zero, the midpoint of the market maker's bid and ask is zero. Therefore, the ask price in the lit market is some $A>0$, and bid price in the lit market is $-A$.  That is, the half-spread is represented by $A$. We normalize $A$ by the standard deviation of $\tilde{v}$, $\frac{A}{\sigma_v}$, and get the normalized half-spread. For simplicity, we refer to $A$ as the ``spread,'' and $\frac{A}{\sigma_v}$ as the ``normalized spread.'' The spread represents a transaction cost in the lit market, because all traders, buyers or sellers, lose $A$ dollars (per unit) to the market maker whenever they trade on the exchange. Thus, alternatively, we also refer to $A$ as the (per unit) ``exchange transaction cost'' and $\frac{A}{\sigma_v}$ as the (per unit) ``normalized exchange transaction cost.''

In period 1, since informed speculators hold some information advantage about the asset, the market maker may lose money to the informed traders ex post. For example, if the realized value of the asset is $\sigma_v$, then the market maker loses money if he is trading against a ``Buy'' order. Precisely, let $\overline{\gamma_e}, \underline{\gamma_e}$ be the respective fraction of informed speculators who place ``Buy'' and who place ``Sell'' orders on exchange, and let $\alpha_e$ be fraction of uninformed liquidity traders who trade in the exchange. For now we assume that they do not split orders among venues, then WLOG if the realized value of $\tilde{v}$ is $\sigma_v$, the ex post payoff of the market maker is
$$ \mbox{MM payoff}=\sigma_v[(\underline{\gamma_e}\mu -\overline{\gamma_e}\mu) +( \alpha_e Z^- -\alpha_e Z^+)]+ A[\overline{\gamma_e}\mu+\underline{\gamma_e}\mu+\alpha_e Z^+ +\alpha_e Z^-],$$
where the first term is the market maker's profit on the asset. It is composed of the net gain from the informed traders, $\underline{\gamma_e}\mu -\overline{\gamma_e}\mu$, and the net gain from the uninformed traders, $ \alpha_e Z^- -\alpha_e Z^+$. The second term is the gains obtained from the transaction fee (spread) per every exchange order. If the realized value of the asset is $-\sigma_v$, by symmetry, the market maker's payoff shall be the same as above. In this way, we also refer to $\overline{\gamma_e}$ as the fraction of informed who ``make money'' (trade in the ``right direction''), and $\underline{\gamma_e}$ the fraction of informed who ``lose money'' (trade in the ``wrong direction'').

A market maker's objective is to break even on average.\footnote{One can think of this as as result of the competition among market makers. For simplicity, we assume that there is one market maker operating.} That is,
$$ 0 = \mathbb{E} \left\{\sigma_v[(\underline{\gamma_e}\mu -\overline{\gamma_e}\mu) +( \alpha_e Z^- - \alpha_e Z^+)]+ A[\overline{\gamma_e}\mu+\underline{\gamma_e}\mu+\alpha_e Z^+ +\alpha_e Z^-]\right\}.$$

Since  $\mathbb{E}Z^+ = \mathbb{E} Z^- = \frac{1}{2}\mu_z$, the market maker's objective becomes
$$0= \sigma_v(-\overline{\gamma_e}+\underline{\gamma_e})\mu+ A[(\overline{\gamma_e}+\underline{\gamma_e})\mu+\alpha_e \mu_z].$$
It implies that,
\begin{align}
A& = \frac{\overline{\gamma_e}-\underline{\gamma_e}}{\overline{\gamma_e}+\underline{\gamma_e}+\alpha_e \frac{\mu_z}{\mu}} \sigma_v.
\label{spread}
\end{align}

If $\overline{\gamma_e}\geq \underline{\gamma_e}(\geq 0)and \alpha_e\geq 0$, then the normalized spread $0\leq \frac{A}{\sigma_v} \leq 1$. In the next sections, we will show that in equilibrium, $\overline{\gamma_e}>\underline{\gamma_e}$. In other words, informed traders are more likely to ``make money'' (trade in the ``right direction''). Intuitively this is true because of their information advantage. Therefore, on average, the market maker loses money to the informed.

At the end of period 1, the market maker observe the exchange volumes $V_b, V_s$ for ``Buy'' volume and ``Sell'' volume respectively. Based on such information, the market maker then announces a closing price $P_1 = \mathbb{E}[\tilde{v}|V_b, V_s]$, which we consider as a proxy for the fundamental value of the asset $\tilde{v}$.  This is because $\mathbb{E}[\tilde{v}|P_1, V_b, V_s]=\mathbb{E}[\mathbb{E}[\tilde{v}|P_1, V_b, V_s]|P_1]=P_1$. We are interested in how much the price $P_1$ can aggregate information in the market (price discovery), that is, how close $P_1$ is to the true value of the asset.

In period 2, since the realization of $\tilde{v}$ has already been revealed, all trades  will be made at the price that is equal to that realization. Thus, the payoff of the market maker in period 2 is automatically zero.

The reason we model the exchange as a market maker instead of other trading protocols such as limit order books is for the same reason as \citet*{zhu_dark_2013}. It is a simple but tractable way to capture the basic trade-off of dark pools. These trade-offs include lower transaction costs (lower spread) and higher execution risks, which is common to most trading protocols.
\subsection{Transaction rules in the dark pool}
We consider the operational costs of the dark pool as a sunk cost, and hence not considered in the model. Also, we normalize the entry fee of a dark pool as zero. The trading protocols in the dark pool we consider, include the pricing mechanism, which refers to on what price the dark pool execute orders, and the execution mechanism, which refers to how to match the buying and selling orders.\footnote{As in section \ref{sec_dp_overview}, we point out that not all dark pools are equal. There might be other features that investors concern. But for simplicity we focus on the two major aspects of a dark pool.}

We restrict our attention to dark pools of a particular pricing mechanism: the midpoint pricing. That is, the orders in the dark pool are crossed at the midpoint of the bid-ask in the exchange.  Since the midpoint of the exchange price is 0, the transaction price in the dark pool is 0. The midpoint pricing mechanism is a reflection of an advantage trading in the dark pool: price improvement. As we point out previously, a trader has to pay a transaction cost (the spread) $A$ on the exchange, no matter at which direction he or she is trading. But in the dark pool, such cost is reduced to 0.

The execution mechanism we consider in this paper is a rationing mechanism. That is, orders in the shorter side are executed with probability one, whereas orders in the longer side are executed probabilistically to balance the market. For example, suppose the realization of $\tilde{v}$ is $\sigma_v$ (the case when $\tilde{v}=-\sigma_v$ is symmetric). Let $\overline{\gamma_d}$, $\underline{\gamma_d}$ be the fractions of informed speculators who trade in the ``right direction'' and ``wrong direction'' respectively, $\alpha_d$ be the fraction of uninformed liquidity traders who trade in the dark pool in period 1, then the respective expected execution rates (taken with respect to $Z^-, Z^+$) for trading in the ``wrong direction'' and in the ``right direction'' are: ,
\begin{align}
\bar{R}&=\mathbb{E}\left[\min\left\{1, \frac{\overline{\gamma_d}\mu + \alpha_d Z^+}{\underline{\gamma_d}\mu + \alpha_d Z^-} \right\} \right],
\label{R_upper}\\
\underline{R}&= \mathbb{E}\left[\min\left\{1, \frac{\underline{\gamma_d}\mu + \alpha_d Z^-}{\overline{\gamma_d}\mu + \alpha_d Z^+} \right\}\right].
\label{R_lower}
\end{align}
Therefore, $\underline{R},\bar{R} \in [0,1]$. The execution mechanism in the dark pool reflects a disadvantage of trading in the dark pool: execution risk. On average, one cannot expect that his or her orders be executed with probability 1 in a dark pool. In contrast, the market maker in the exchange is able to provide such certainty.

Moreover, as we will show in the next section, $\overline{\gamma_d}>\underline{\gamma_d}$. This means that the information asymmetry exists in the dark pool and informed traders are more likely to trade in the ``right direction.'' Therefore, $ \underline{R}\leq \bar{R}$. That is to say, orders that are in the ``right direction'' are less likely to be executed than orders that are in the ``wrong direction.'' In this way,  we obtain a measure of dark pool adverse selection cost in the dark pool by
\begin{align*}
(\bar{R}-\underline{R})\sigma_v
\end{align*}
We therefore refer to $\bar{R}-\underline{R}$ as the ``Normalized dark pool adverse selection cost.''

Without loss of generality, we assume that the dark pool operates only in period 1. In period 2, since the realization of $\tilde{v}$ is revealed, orders in the exchange are executed at that realized value. The dark pool loses its advantage and becomes redundant as nobody is willing to trade there. Therefore, unless cancelled, orders that failed to execute in period 1 will be routed to the exchange and executed there in period 2.

\subsection{The informed speculators' problem}
As we point out, the informed traders only participate in period 1, when they can use their private information to their advantage. Upon the reception of a signal, the informed speculators update their beliefs about the asset fundamental value using Bayes' rule. Let $B(s)$ be the probability that the realization is high ($\sigma_v$), conditional on signal $s$, then by Bayes' rule,
\begin{align}
B(s)=\Pr(\tilde{v}=\sigma_v|s)=\frac{\phi(\frac{s-\sigma_v}{\sigma_e})}{\phi(\frac{s-\sigma_v}{\sigma_e})+\phi(\frac{s+\sigma_v}{\sigma_e})},
\label{Belief}
\end{align}
where $\phi(x)$ is the pdf of a standard normal distribution function. $B(s)\in (0,1)$ and $B(s)$ is strictly increasing in $s$.

Consider an informed trader with signal $s$,  given the exchange spread, $A$,  and the dark pool execution probabilities, $\bar{R}, \underline{R}$, the expected (per unit) ``Buy'' and ``Sell'' profit in each venue, or do not trade, are respectively,
\begin{alignat*}{3}
&\mbox{Exchange(Lit):}&& \mbox{ ``Buy'': }B(s)\sigma_v-(1-B(s))\sigma_v -A ,\\
&                    &&   \mbox{ ``Sell'': }  -[B(s)\sigma_v-(1-B(s))\sigma_v]-A .\\
&\mbox{Dark pool:}&& \mbox{ ``Buy'': }B(s)\underline{R}\sigma_v-(1-B(s))\bar{R}\sigma_v,\\
&                 &&  \mbox{  ``Sell'': } -[B(s)\underline{R}\sigma_v-(1-B(s))\bar{R}\sigma_v].\\
&\mbox{Not trade:}&& 0.
\end{alignat*}

An informed speculator's problem is then, given his or her signal $s$, to choose a trading direction in $\{``Buy", ``Sell"\}$, the quantity in each venue \{Exchange(Lit),  Dark pool, Do not trade\} to maximize his or her total expected payoff, such that total quantity does not exceed 1 unit.\footnote{The case that the informed speculator simultaneously place ``Buy'' and ``Sell'' orders in each venue is not considered,  because the agents have no individual impact to the market. By the linearity of the per unit profit in each venue, it is never optimal to do so.}

We argue that, in equilibrium, whenever he or she decides to trade, an informed trader will place a ``Buy'' order if his or her signal is positive, and a ``Sell'' order if his or her signal is negative. Moreover, almost surely it is optimal for him to send the entire order to one of the two venues, or not trade at all. The argument is summarized in Lemma \ref{non-split informed}.

\begin{lem}\textbf{(Trading direction and non-split orders, informed)}\footnote{A non-slit order is strictly preferred in this model. This is a stronger result than \citet*{zhu_dark_2013}, in which it is only weakly optimal to not split orders for the informed because they are all indifferent between the two venues. }
If an informed trader decide to trade, it is strictly optimal to ``Buy'' if his or her signal $s>0$ and to ``Sell'' if $s<0$. Moreover, with probability one, an informed trader strictly prefers to send the entire order to one of the two venues, or do not trade at all.
\label{non-split informed}
\end{lem}
The trading direction is rather straightforward since a positive signal indicates that the asset's fundamental value is more likely to be high (i.e., $\sigma_v$), and hence more profitable in a ``Buy'' direction, whereas a negative signals indicates a low value (i.e., $-\sigma_v$) and hence more profitable in a ``Sell'' direction. And, since each trader's signal is drawn from the same continuous distribution, and there is a continuum of informed traders, by law of large numbers, the realization of signals among them are continuously distributed. Therefore, the beliefs are distributed continuously. Since no individual has impact on the market, and the expected profit in each venue is linear in the agents' beliefs, it is with probability 1 that, for any informed trader with signal $s$, one venue (or not trade) is strictly better than others.

By Lemma \ref{non-split informed}, the potential trading direction is determined once an agent receives his or her signal. Moreover, the magnitude of $B(|s|)$ reflects the probability that this trading direction is ``right.'' Thus $|s|$ can be regarded as the strength of one's signal, and $B(|s|)$, can be regarded as the agent's confidence level in their information. A strong signal (i.e., a high $|s|$) represents a strong belief that the trading direction is ``right,'' whereas weak signals (i.e., low $|s|$)  represents a weak belief in the trading direction. We will show in the next section, how much credit an informed trader gives to his or her private information is crucial in determining his or her strategies of venue selection.

Based on an informed traders' signal strength, $B(|s|)$, the payoffs of trading in each venue and no trade are, respectively,
\begin{alignat}{3}
&\mbox{Exchange(Lit)}&&: B(|s|)\sigma_v-(1-B(|s|))\sigma_v -A ,
\label{payoff_informed_lit}\\
&\mbox{Dark pool}&&: B(|s|)\underline{R}\sigma_v-(1-B(|s|))\bar{R}\sigma_v,
\label{payoff_informed_dp}\\
&\mbox{Not trade}&&: 0.
\label{payoff_informed_no}
\end{alignat}

An informed agent's problem is then reduced to choosing one of the two venues and sending the entire 1 unit to it, with a trading direction specified in Lemma \ref{non-split informed}, or not trade at all, to yield the maximum payoff, based on his or her confidence level $B(|s|)$.

Finally, we define the strategy of an informed speculator who receives a signal $s$ by a mapping
$$h_I(s): (\infty, \infty) \rightarrow \{\mbox{``Buy''}, \mbox{ ``Sell''}\} \times \{\mbox{Exchange(Lit)},  \mbox{ Dark pool}, \mbox{ Not trade}\}.$$

\subsection{The uninformed liquidity traders' problem}
Liquidity buyer or seller types are specified by the level of their liquidity demand -- the (per unit) delay cost $d$. If the agent fails to have his or her order executed in period 1, he or she will bear a (per unit) cost of $\sigma_v d$.  Therefore a higher delay cost implies a higher demand for liquidity, and a higher devaluation on execution risk for the traders.

More precisely, a type $d$ uninformed liquidity buyer's (seller's) per unit payoffs of trading in the exchange, in the dark pool, or delaying trade are, respectively,
\begin{alignat}{3}
&\mbox{Exchange(Lit)}&&: -A,
\label{payoff_uninformed_lit}\\
&\mbox{Dark pool}&&: - {(\bar{R}-\underline{R})\over 2}\sigma_v - (1- {\bar{R}+\underline{R})\over 2})\sigma_v d ,
\label{payoff_uninformed_dp}\\
&\mbox{Delay trade}&&: -\sigma_v d.
\label{payoff_uninformed_no}
\end{alignat}

Similarly, we argue that in period 1, it is strictly optimal for any liquidity trader to send the entire order to one of the two venues, or delay the trade, almost surly. The argument is summarized in Lemma \ref{non-split uninformed}.
\begin{lem}
\textbf{(No split orders, uninformed)}A liquidity trader (buyer or seller) strictly prefers to send the entire order to one of the venues, or delay trade.
\label{non-split uninformed}
\end{lem}
The intuition of Lemma \ref{non-split uninformed} is similar. Since all individuals are infinitesimal, no single trader has an impact on the market. For any liquidity trader, he or she either strictly prefers one venue over the other or is indifferent between two venues (or do not trade). Since the distribution of the delay cost $d$ is continuous, it is with probability one that one venue (or delay) is strictly better than the other.

By Lemma \ref{non-split uninformed}, a type $d$ liquidity buyer's (or seller's) problem is to maximize his or her payoff (i.e., minimize the costs), by choosing one of the venues in which trade the entire order in period 1, or to delay trade to period 2.

Moreover, we define the strategy of a type $d$ uninformed liquidity trader by a mapping:
$$h_{U,\iota}(d): [0,\bar{d}] \rightarrow \{\mbox{Exchange(Lit)},  \mbox{ Dark pool}, \mbox{ Delay trade}\},$$
where $\iota \in \{ \mbox{Buyer}, \mbox{ Seller}\}$

Finally, the trading timeline of the model is summarized in Figure \ref{timeline}. At period 0, the asset fundamental value $\tilde{v}$, the measure of liquidity buyers $Z^+$ and liquidity sellers $Z^-$, the signal for each informed trader $s_i$, the per unit delay cost for each uninformed trader $d_j$ are realized. But none of this information is public. Also, at period 0, the market maker announces the bid-ask prices with the spread $A$.  After that, traders select venues in which place orders,  which are executed according to the transaction rules in each venue. At the end of period 1, before the revelation or the value of the asset, the market maker announces a closing price of period 1, based on the volumes he observes in the exchange during that period. Then after the revelation of $\tilde{v}$, orders that failed to execute in period 1 are routed to the exchange (unless cancelled) and execute at the revealed value of $\tilde{v}$. The market is then closed.
\begin{figure}[h]
\centering
\includegraphics[width=\textwidth]{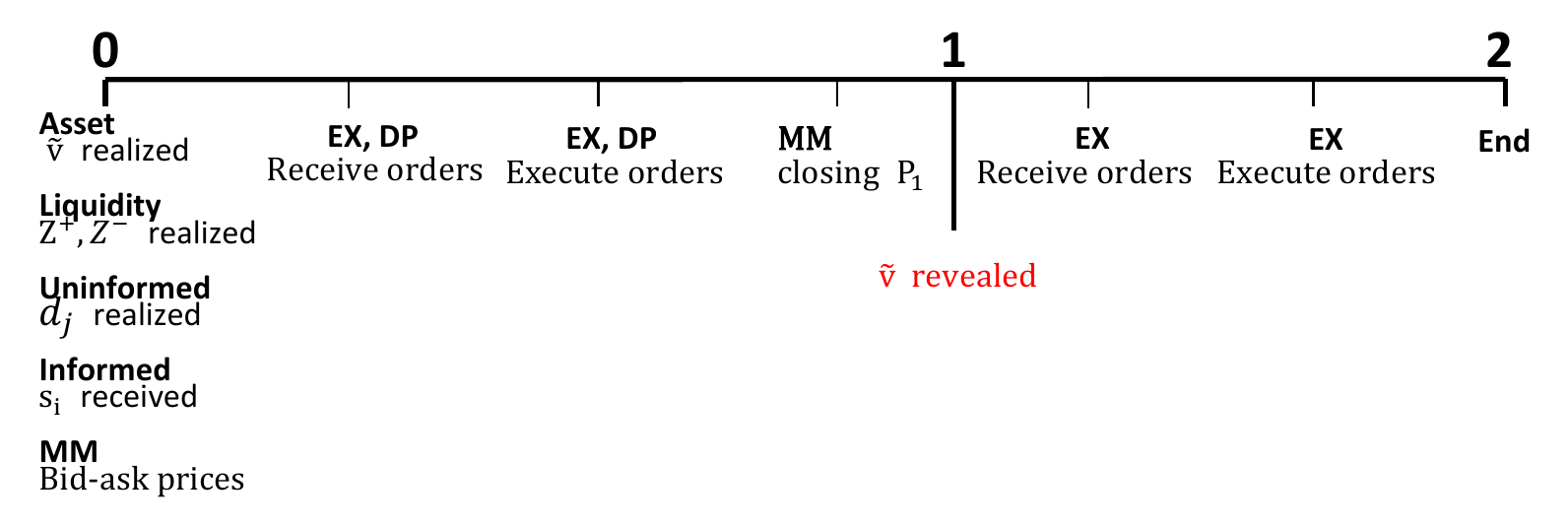}
\caption{Trading Timeline}
\label{timeline}
\end{figure}
\section{The Equilibrium}
The model we describe in Section \ref{section_themodel} assumes that both the exchange (Lit), and the dark pool are available to traders. We refer to it as the ``Multi-venue'' Model. We now introduce a benchmark in which there is only one venue that is operating: the exchange (Lit market). We refer to it as the ``Single-venue'' Model. The comparison between the two model in Section \ref{sec_compstatics} gives us insights into the impacts of dark pools to market behaviors.

\subsection{Benchmark model: without a dark pool}
\label{sec_benchmark_model}
In the benchmark model,  all else are the same except that the exchange (the lit market) is the only trading venue available for traders. Lemma \ref{non-split informed} and Lemma \ref{non-split uninformed} also hold in this model, i.e., traders do not split their orders. We use the superscription ``$\mathbb{S}$'' to denote the ``single venue'' model. The equilibrium is defined as follows:
\begin{defn}
\textbf{(Benchmark: without a dark pool)} An equilibrium of the ``Single-venue'' model is a strategy for the informed speculators, $h^{\mathbb{S}}_I(s)$, a strategy for the uninformed liquidity traders, $h^{\mathbb{S}}_{U,\iota}(d)$, $\iota \in \{\mbox{Buyer}, \mbox{ Seller}\}$, an exchange spread $A^{\mathbb{S}}$, a set of participation fractions $\overline{\gamma_e}^{\mathbb{S}}, \underline{\gamma_e}^{\mathbb{S}}, \alpha_e^{\mathbb{S}}$, such that
\begin{itemize}
\item (i) given $A^{\mathbb{S}}$, $h^{\mathbb{S}}_I(s)$ and $h^{\mathbb{S}}_{U,\iota}(d)$ are optimal, respectively, for an informed speculator with signal $s$ and for an uninformed liquidity trader with per unit delay cost $d$;
\item (ii)  given $\overline{\gamma_e}^{\mathbb{S}}, \underline{\gamma_e}^{\mathbb{S}}$, and $\alpha_e^{\mathbb{S}}$, the exchange spread $A^{\mathbb{S}}$ makes a market maker in the exchange break-even on average;
\item (iii) $\overline{\gamma_e}^{\mathbb{S}}, \underline{\gamma_e}^{\mathbb{S}}$ measure the respective fractions of informed traders who trade in the ``right'' and ``wrong'' direction in the exchange, and $\alpha_e^{\mathbb{S}}$ measures the period 1 exchange fraction of uninformed traders.
\end{itemize}
\label{eqn def3}
\end{defn}

Given $\overline{\gamma_e}^{\mathbb{S}}, \underline{\gamma_e}^{\mathbb{S}}$, and $\alpha_e^{\mathbb{S}}$, an exchange spread $A^{\mathbb{S}}$ that makes the market maker break even on average satisfies (\ref{spread}). That is,
\begin{align}
A^{\mathbb{S}}& = \frac{\overline{\gamma_e}^{\mathbb{S}}-\underline{\gamma_e}^{\mathbb{S}}}
{\overline{\gamma_e}^{\mathbb{S}}+\underline{\gamma_e}^{\mathbb{S}}+\alpha_e^{\mathbb{S}} \frac{\mu_z}{\mu}} \sigma_v.
\label{spread_benchmark}
\end{align}

Equation (\ref{spread_benchmark}) implies that if $\overline{\gamma_e}^{\mathbb{S}}\geq \underline{\gamma_e}^{\mathbb{S}}\geq 0$, and $\alpha_e^{\mathbb{S}}> 0$, then $\sigma_v\geq A^{\mathbb{S}}\geq 0$. Considering an informed trader with signal ``$s$,'' by Lemma \ref{non-split informed}, the optimal trading direction is to ``Buy'' if $s\geq 0$ and to ``Sell'' if $s<0$. Then given $A^{\mathbb{S}}$, The expected payoffs of trading in the exchange and do not trade are, respectively:
\begin{alignat*}{3}
&\mbox{Exchange(Lit)}&&: B(|s|)\sigma_v-(1-B(|s|))\sigma_v -A^{\mathbb{S}} ,\\
&\mbox{Not trade}&&: 0.
\end{alignat*}

\begin{figure}[h]
\centering
\begin{subfigure}[b]{0.48\textwidth}
\centering
\includegraphics[width=\textwidth]{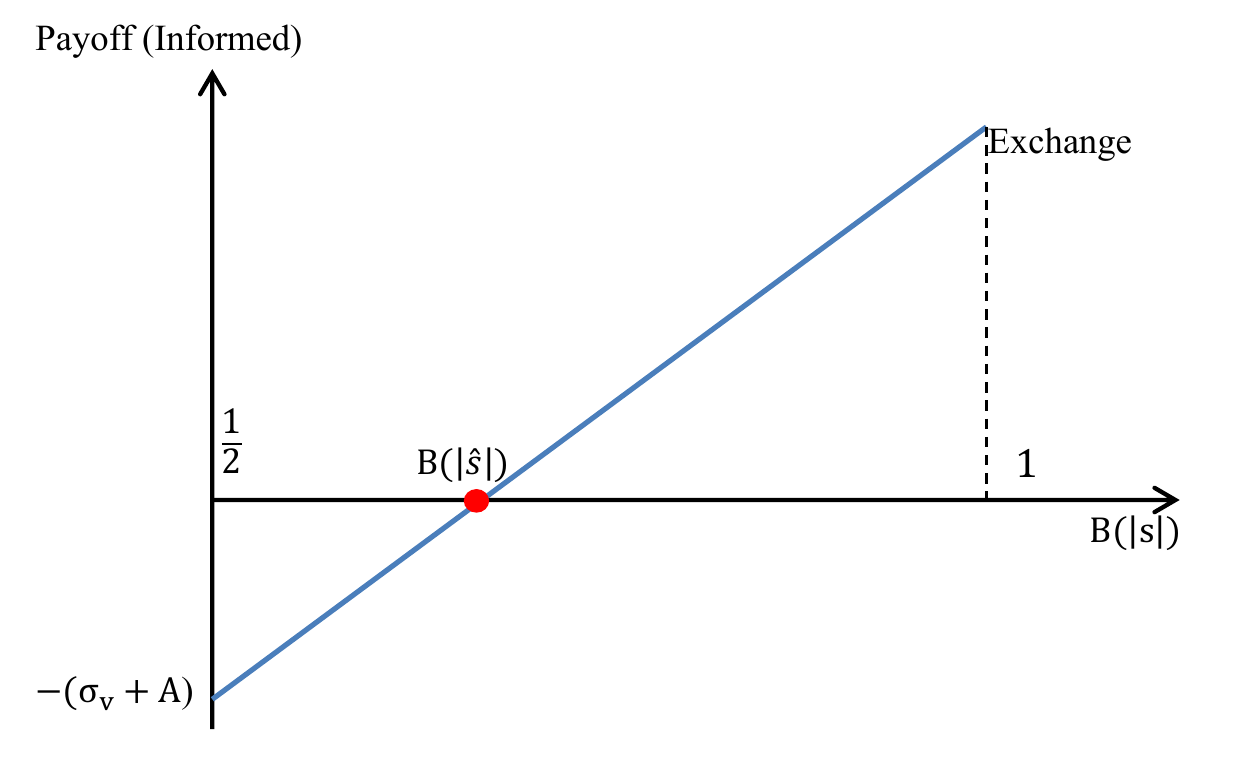}
\caption{Informed}
\label{fig_payoffs_informed_benchmark}
\end{subfigure}
\begin{subfigure}[b]{0.48\textwidth}
\includegraphics[width=\textwidth]{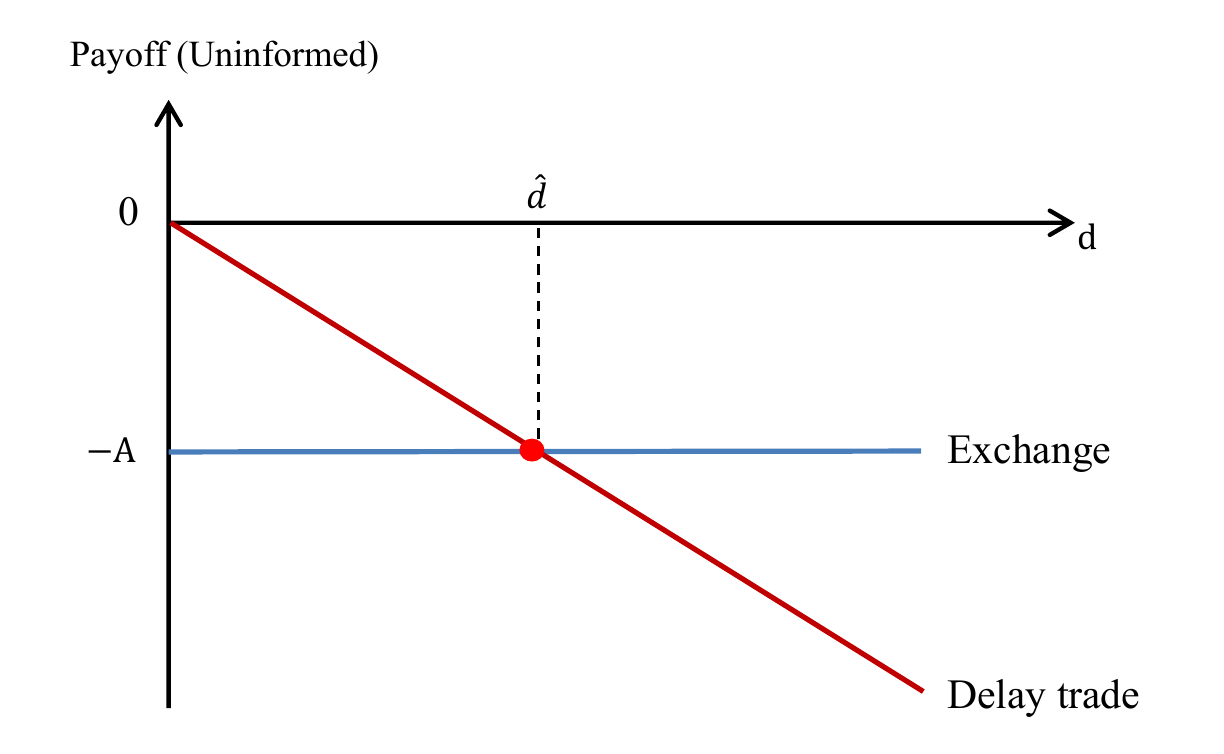}
\caption{Uninformed}
\label{fig_payoffs_uninformed_benchmark}
\end{subfigure}
\caption{Payoffs For Traders, Single-venue}
\end{figure}

Suppose $\sigma_v\geq A^{\mathbb{S}}\geq 0$, then if the signal is extremely weak, i.e., $B(|s|)=\frac{1}{2}$, or, $s=0$, the expected payoff of trading in the exchange is strictly negative, and it is strictly optimal not to trade. In contrast, if the signal is extremely strong, i.e., $B(s)=1$, or, $s=\pm \infty$, the expected payoff of trading in the exchange is strictly positive, and it is strictly optimal to trade in the exchange. This is illustrated in Figure \ref{fig_payoffs_informed_benchmark}. Therefore there must exist some cut-off point
$\widehat{s}>0$ such that the $\widehat{s}$ type informed traders are indifferent between trading in the exchange and do not trade. That is,
\begin{align}
B(\widehat{s})\sigma_v-(1-B(\widehat{s}))\sigma_v -A^{\mathbb{S}}=0,
\label{shat_benchmark}
\end{align}
and the optimal choice for an informed trader with signal $s$ is then
\begin{align}
h^{\mathbb{S}}_I(s)&=\left\{\begin{array}{ll}
                                                                         (\mbox{``Buy''},\mbox{ Exchange(Lit)}) &\mbox{ if } s \geq \widehat{s},\\
                                                                          (\mbox{``Sell''},\mbox{ Exchange(Lit)}) &\mbox{ if } s < -\widehat{s}, \\                                                                                                                                                  \mbox{Do not trade}&\mbox{ others. }                                                                    \end{array}\right.
\label{strategy_informed_single}
\end{align}
Without loss of generality, we assume that the realization of $\tilde{v}$ is $\sigma_v$. If all informed speculators follow the same optimal strategy, then the fraction of informed traders who will trade in the ``right'' and ``wrong'' directions across the population are, respectively,
\begin{align}
\overline{\gamma_e}^{\mathbb{S}}&= Pr (s\geq \widehat{s}|\tilde{v}=\sigma_v)=Pr (s\leq -\widehat{s}|\tilde{v}=-\sigma_v)=1-\Phi(\frac{\widehat{s}-\sigma_v}{\sigma_e}),
\label{gammabare_benchmark}\\
\underline{\gamma_e}^{\mathbb{S}}&= Pr (s< -\widehat{s}|\tilde{v}=\sigma_v)=Pr (s> \widehat{s}|\tilde{v}=-\sigma_v)=1-\Phi(\frac{\widehat{s}+\sigma_v}{\sigma_e}).
\label{gammalowe_benchmark}
\end{align}

(\ref{gammabare_benchmark}),(\ref{gammalowe_benchmark}) imply that $\overline{\gamma_e}^{\mathbb{S}}\geq \underline{\gamma_e}^{\mathbb{S}}> 0$.

Now, we consider an uninformed liquidity trader with a (per unit) delay cost ``$d$.'' Similarly, his or he payoffs of trading in the exchange and delaying trade are, respectively:
\begin{alignat*}{3}
&\mbox{Exchange(Lit)}&&: -A^{\mathbb{S}},\\
&\mbox{Delay trade}&&: -\sigma_v d.
\end{alignat*}
Since $d\in[0, \bar{d}]$ with $\bar{d}\geq 1$, and $\sigma_v\geq A^{\mathbb{S}}\geq 0$, if the liquidity trader is extremely patient, i.e., $d=0$, it is strictly optimal to delay trade to period 2. In contrast, if the liquidity trader is extremely impatient, i.e., $d=\bar{d}>1$, it is strictly optimal to trade in the exchange. This is shown in Figure \ref{fig_payoffs_uninformed_benchmark}. Therefore, there also exists a cut-off $\widehat{d}$ such that the type ``$\widehat{d}$ '' liquidity trader is indifferent between trading in the exchange and delaying trade to the next period. That is,
\begin{align}
-A^{\mathbb{S}}&=-\sigma_v \widehat{d}.
\label{dhat_benchmark}
\end{align}
To combine (\ref{shat_benchmark}) with (\ref{dhat_benchmark}), we derive that
$$\widehat{d} = 2B(\widehat{s})-1.$$
The optimal strategy for uninformed liquidity traders is then,
\begin{align}
h^{\mathbb{S}}_{U,\iota}(d)&=\left\{\begin{array}{ll}
            (\mbox{``Buy'' if $\iota$=Buyer, or ``Sell'' if $\iota$=Seller},\mbox{ Exchange(Lit)}) &\mbox{ if } d \geq 2B(\widehat{s})-1,\\
            \mbox{Delay trade} &\mbox{ others. }
            \end{array}\right.
\label{strategy_uninformed_single}
\end{align}
The period 1 exchange participation rate for the uninformed traders is then
\begin{align}
\alpha_e^{\mathbb{S}} &= Pr(d\geq \widehat{d}) = 1-G(2B(\widehat{s})-1),
\label{alphae_benchmark}
\end{align}
and $0\leq \alpha_e^{\mathbb{S}}\leq 1$.

We then find a cut-off equilibrium. Theorem \ref{thm existence0} summarizes the existence and uniqueness.

\begin{thm}
\textbf{(Existence and Uniquness, benchmark)} For any $\sigma_e, \sigma_v \geq 0$, there exists an equilibrium in which traders follow cut-off strategies. That is, the respective optimal strategies for informed speculators and uninformed liquidity traders, $h^{\mathbb{S}}_I(s)$ and $h^{\mathbb{S}}_{U,\iota}(d)$, are defined as (\ref{strategy_informed_single}) and (\ref{strategy_uninformed_single}), with the cut-off $\widehat{s}$ determined by (\ref{shat_benchmark}). The exchange spread $A^{\mathbb{S}}$ satisfies (\ref{spread_benchmark}), and the participation fractions $\overline{\gamma_e}^{\mathbb{S}},\underline{\gamma_e}^{\mathbb{S}},\alpha_e^{\mathbb{S}}$ are determined respectively by (\ref{gammabare_benchmark}), (\ref{gammalowe_benchmark}), (\ref{alphae_benchmark}), (\ref{spread_benchmark}).

Moreover, every equilibrium is a cut-off equilibrium, and the equilibrium is unique if $\sigma_e, \sigma_v>0$, $G'(x)+xG''(x)\geq 0, \forall x \in [0,1]$.
\label{thm existence0}
\end{thm}


The benchmark clearly gives us some insight regarding the \textit{sorting effect} on types of traders. In equilibrium, it is strictly optimal for informed traders with relatively strong signals to trade in the exchange and for those with weak signals not to trade (avoid trading). Similarly, it is strictly optimal for uninformed liquidity traders who are relatively patient to trade in the exchange and for those who are relatively impatient to delay trade. The exchange provides functions to separate certain types of traders from others. As we will point out later, such a sorting effect is even strengthened in the presence of a dark pool.


\subsection{Multi-venue model: with a dark pool}
\label{sec_eqm_multi}
Two trading venues are available in the multi-venue model: an exchange (Lit) and a dark pool. To differentiate from the single-venue model, we do not use the superscription $\mathbb{S}$ in the multi-venue model. The equilibrium of the multi-venue is defined as follows:
\begin{defn}\textbf{(Multi-venue, with a dark pool)}
An equilibrium is a strategy for the informed speculators, $h_I(s)$, a strategy for and for the uninformed liquidity traders,  $h_{U,\iota}(d)$, an exchange spread, $A$, two expected execution rate in the dark pool $\bar{R}, \underline{R}$ , and a set of participation fractions $\overline{\gamma_e},\underline{\gamma_e},\overline{\gamma_d},\underline{\gamma_d},\alpha_e,\alpha_d$, s.t.
\begin{itemize}
\item (i) $h_I(s)$ is optimal for informed speculators with signal $s$, whereas $h_{U,\iota}(d)$ is optimal for uninformed liquidity traders with (per unit) delay cost $d$, given $A, \bar{R}$, and $\underline{R}$.
\item (ii) the exchange spread $A$ makes a market maker in the exchange break-even on average, given $\overline{\gamma_e},\underline{\gamma_e},\overline{\gamma_d},\underline{\gamma_d},\alpha_e$, and $\alpha_d$;
\item (iii) the dark pool operates using a mid-pricing and a rationing execution mechanism. $\underline{R}$ and $\bar{R}$ are the respective expected execution probability for orders that are in the ``right'' and in the ``wrong'' directions;
\item (iv) $\overline{\gamma_e}$ and $\underline{\gamma_e}$ measure the respective fractions of informed traders in the exchange who trade in the ``right'' and ``wrong'' directions. $\overline{\gamma_d}$ and $\underline{\gamma_d}$ measure the respective fractions of informed traders in the dark pool who trade in the ``right'' and ``wrong'' directions. $\alpha_e$ and  $\alpha_d$ measure the respective fraction of uninformed traders who trade in the exchange and in the dark pool in period 1.
\end{itemize}
\label{eqn def}
\end{defn}

Consider an informed speculator with signal ``$s$.'' Based on the strength of his or her signal $B(|s|)$, the payoffs of trading in the exchange, the dark pool and do not trade are summarized in (\ref{payoff_informed_lit}), (\ref{payoff_informed_dp}), (\ref{payoff_informed_no}). These payoffs are shown in Figure \ref{fig_payoffs_informed}.
\begin{figure}[h]
\centering
\begin{subfigure}[b]{0.48\textwidth}
\centering
\includegraphics[width=\textwidth]{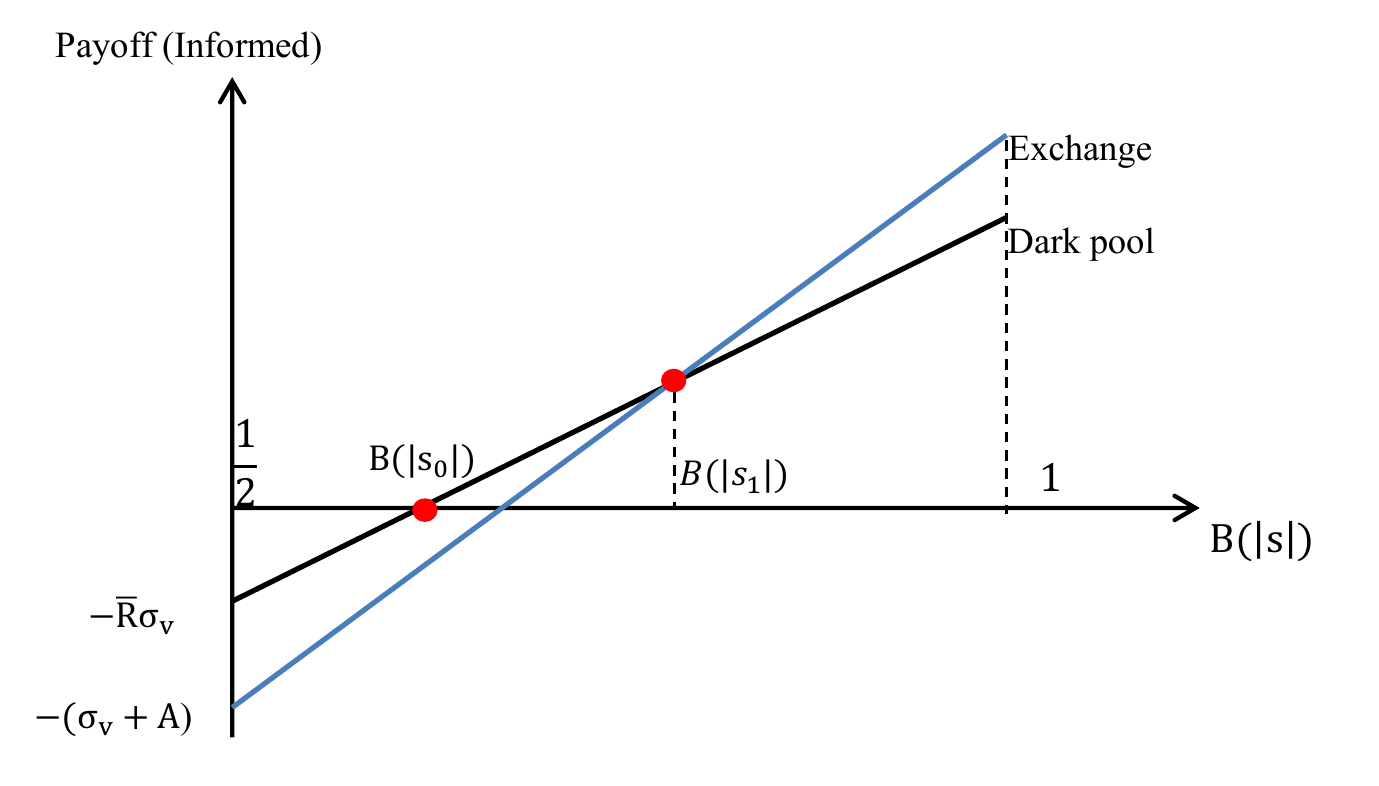}
\caption{Informed}
\label{fig_payoffs_informed}
\end{subfigure}
\begin{subfigure}[b]{0.48\textwidth}
\includegraphics[width=\textwidth]{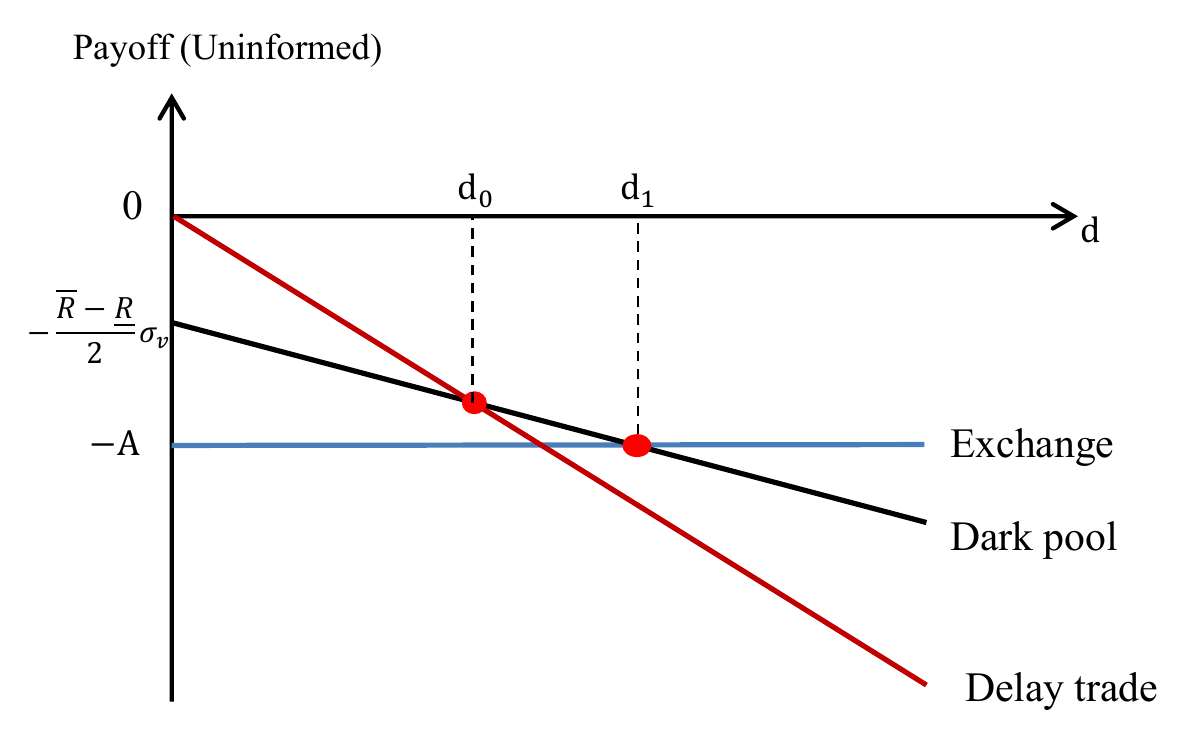}
\caption{Uninformed}
\label{fig_payoffs_uninformed}
\end{subfigure}
\caption{Payoffs For Traders, Multi-venue}
\end{figure}

Suppose  $1\geq \bar{R} \geq \underline{R}> 0$ and $\sigma_v\geq A\geq 0$. As is shown in Figure \ref{fig_payoffs_informed}, if a trader receives extremely weak signals ($s=0$ for example), it is never profitable to trade, since trading is costly. However, whenever an informed trader decides to trade, he faces a trade-off between execution certainty in the exchange and price improvement in the dark pool. When $|s|$ is low, the need for price improvement overwhelms the need for execution, in which case, trading in a dark pool is better. But as the signals becomes stronger, the need for execution grows faster than the need for price improvement.  This can be observed from the fact that the exchange payoff has a higher slope with respect to $B(|s|)$ than the dark pool payoff. Therefore, when $s$ is extremely high, it is possible that the two intersect. Suppose an informed trader with signal $s_0>0$ is indifferent between trading in a dark pool and not trade, an informed with signal $s_1>0$ is indifferent between trading in a dark pool and in the exchange, then by (\ref{payoff_informed_lit}), (\ref{payoff_informed_dp}), and (\ref{payoff_informed_no}), $s_0, s_1$ satisfies:
\begin{align}
B(s_0)(\bar{R}+\underline{R})&=\bar{R}
\label{equation0}\\
B(s_1)\left[(1-\bar{R})+(1-\underline{R})\right]\sigma_v&=A+(1-\bar{R})\sigma_v.
\label{equation1}
\end{align}

At this point, the existence and relationship of $s_0$ and $s_1$ is not established yet. For now, we suppose that $(s_0, s_1)$ exists and  $s_0<s_1<+\infty$ (we will prove that this is true in every equilibrium), the optimal strategy for an informed trader with signal $s$ is then
\begin{align}
h^{\mathbb{S}}_I(s)&=\left\{\begin{array}{ll}
                                                                         (\mbox{``Buy''},\mbox{ Exchange(Lit)}) &\mbox{ if } s \geq s_1,\\
                                                                         (\mbox{``Buy''},\mbox{ Dark pool}) &\mbox{ if } s_0 \leq s < s_1,\\
                                                                         (\mbox{``Sell''},\mbox{ Dark pool}) &\mbox{ if } -s_1 \leq s < -s_0,\\
                                                                          (\mbox{``Sell''},\mbox{ Exchange(Lit)}) &\mbox{ if } s < -s_1, \\                                                                                                                                                  \mbox{Do not trade}&\mbox{ others. }                                                                    \end{array}\right.
\label{strategy_informed}
\end{align}
This is illustrated in Figure \ref{fig_cutoff_informed}. That is, it is strictly optimal that informed traders with strong signals to trade in the exchange, informed traders with modest signals to trade in the dark pool, and informed traders with weak signals to not trade.
\begin{figure}[h]
\centering
\includegraphics[width=.8\textwidth]{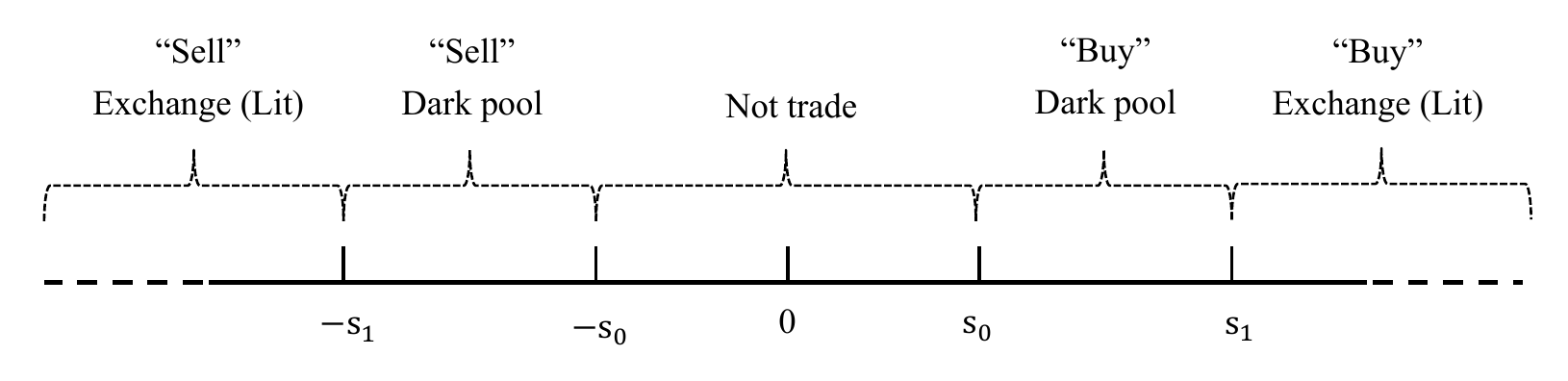}
\caption{Strategy of Informed Traders}
\label{fig_cutoff_informed}
\end{figure}

If all informed traders follow such strategy, the exchange fraction of informed who trade in the ``right'' and ``wrong'' directions are, respectively,
 \begin{align}
\overline{\gamma_e}&=\Pr(s\geq s_1|\tilde{v}=\sigma_v)=\Pr(s\leq -s_1|\tilde{v}=-\sigma_v)=1-\Phi(\frac{s_1-\sigma_v}{\sigma_e}),
\label{gamma_upper}\\
\underline{\gamma_e}&=\Pr(s<-s_1|\tilde{v}=\sigma_v)=\Pr(s>s_1|\tilde{v}=-\sigma_v)=1-\Phi(\frac{s_1+\sigma_v}{\sigma_e}).
\label{gamma_lower}
\end{align}

And the dark pool fraction of informed who trade in the ``right'' and ``wrong'' directions are, respectively,
  \begin{align}
\overline{\gamma_d}&=\Pr(s_0\leq s< s_1|\tilde{v}=\sigma_v)=\Pr(-s_1\leq s< -s_0|\tilde{v}=-\sigma_v)=\Phi(\frac{s_1-\sigma_v}{\sigma_e})-\Phi(\frac{s_0-\sigma_v}{\sigma_e}),
\label{beta_upper}\\
\underline{\gamma_d}&=\Pr(-s_1\leq s< -s_0|\tilde{v}=\sigma_v)=\Pr(s_0\leq s< s_1|\tilde{v}=-\sigma_v)=\Phi(\frac{s_1+\sigma_v}{\sigma_e})-\Phi(\frac{s_0+\sigma_v}{\sigma_e}).
\label{beta_lower}
\end{align}

Similarly, for the uninformed, the payoffs of trading in the exchange, in the dark pool, and delaying trade are respectively given in (\ref{payoff_uninformed_lit}), (\ref{payoff_uninformed_dp}), and (\ref{payoff_uninformed_no}), as illustrated in Figure \ref{fig_payoffs_uninformed}. Again, a liquidity trader with extremely low liquidity demands would find it optimal to delay trade. However, if he decides to trade in period 1, only those with extremely high liquidity demands (i.e., extremely impatient) are willing to trade, for the similar reason as the informed traders. Let $d_0$ and $d_1$ respectively represent the type of liquidity traders who are indifferent between delaying trade and trading in a dark pool, and the type who are indifferent between delaying trading in a dark pool and in the exchange, then by (\ref{payoff_uninformed_lit}), (\ref{payoff_uninformed_dp}), and (\ref{payoff_uninformed_no}) we have
\begin{align*}
 - {(\bar{R}-\underline{R})\over 2}\sigma_v - (1- {\bar{R}+\underline{R})\over 2})\sigma_v d_0 & =-\sigma_v d_0, \\
- {(\bar{R}-\underline{R})\over 2}\sigma_v - (1- {\bar{R}+\underline{R})\over 2})\sigma_v d_1 & =-A.
\end{align*}
Combine this with (\ref{equation0}) and (\ref{equation1}), we derive that
\begin{align*}
d_0 &= 2B(s_0)-1,\\
d_1 &= 2B(s_1)-1.
\end{align*}

By a similar argument, the optimal strategy for an uninformed trader is also a cut-off strategy:
\begin{align}
h^{\mathbb{S}}_{U,\iota}(d)&=\left\{\begin{array}{ll}
            (\mbox{``Buy'' if $\iota$=Buyer, or ``Sell'' if $\iota$=Seller},\mbox{ Exchange(Lit)}) &\mbox{ if } d \geq 2B(s_1)-1,\\
            (\mbox{``Buy'' if $\iota$=Buyer, or ``Sell'' if $\iota$=Seller},\mbox{ Dark pool}) &\mbox{ if } 2B(s_0)-1 \\
						&\leq d <2B(s_1)-1,\\
            \mbox{Delay trade} &\mbox{ otherwise. }
            \end{array}\right.
\label{strategy_uninformed}
\end{align}
This is described in Figure \ref{fig_cutoff_uninformed}. The exchange fraction, $\alpha_e$, and dark pool fraction, $\alpha_d$, of uninformed liquidity traders, are, respectively,
\begin{align}
\alpha_e&=1-G(2B(s_1)-1),
\label{alpha_e}\\
\alpha_d&=G(2B(s_1)-1)-G(2B(s_0)-1).
\label{alpha_d}
\end{align}
\begin{figure}[h]
\centering
\includegraphics[width=.8\textwidth]{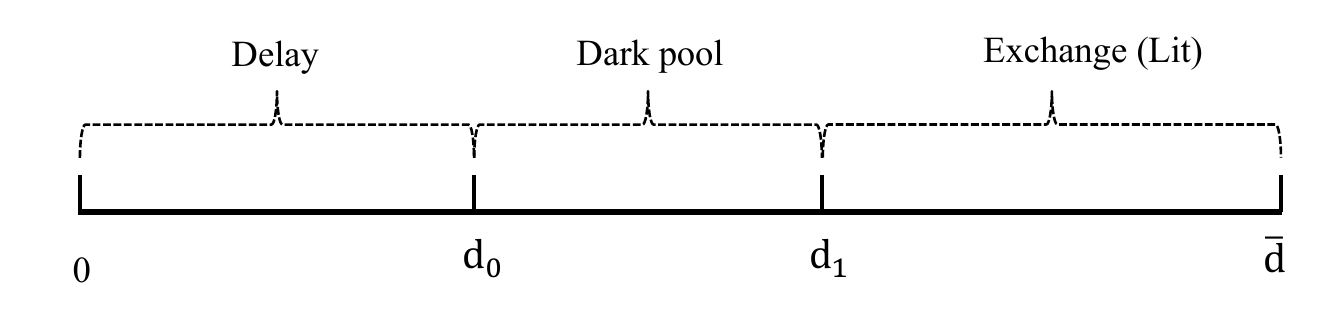}
\caption{Strategy of Uninformed Traders}
\label{fig_cutoff_uninformed}
\end{figure}
The fact that the traders the cut-off of uninformed traders' are functions of the cut-off of informed traders' reveals that, in equilibrium, uninformed and informed traders always move together. It cannot happen that uninformed traders move collectively from one venue to another, forming a new equilibrium without influencing the behavior of the informed traders. This is in contrast with \citet*{zhu_dark_2013}.

Given $\overline{\gamma_e},\underline{\gamma_e},\alpha_e$, the exchange spread $A$ captured in (\ref{spread}) makes the market maker break even. Also, given $\overline{\gamma_d},\underline{\gamma_d},\alpha_d$, and given the distribution of $Z^+$ and $Z^-$, the expected execution rates in the dark pool, $\bar{R}$ and $\underline{R}$, are respectively determined by (\ref{R_upper}) and (\ref{R_lower}).


If such $s_0, s_1$ exists, we find a cut-off equilibrium. But the existence is not obvious. The difficulty arises from two aspects. First, we cannot simply apply a fixed point theorem because it cannot distinguish the trivial equilibrium from others: a trivial equilibrium is one in which all trades happen in one venue, for example, the exchange. Second, the equilibrium involves a very complicated equation system and these equations are non-linear and are not likely to exhibit monotonicity. Nevertheless,  we are able to show in Theorem \ref{thm existence} that the equilibrium exists. Moreover, all equilibria are cut-off equilibra, and all equilibra are non-trivial.
\begin{thm}
\textbf{(Equilibrium with DP)} For any $\sigma_v, \sigma_e >0$, an equilibrium exists in which traders follow cut-off strategies. That is, the respective optimal strategies for informed and uninformed traders, $h_I(s)$ and $h^U_{\iota}(d)$, are defined as in (\ref{strategy_informed}) and (\ref{strategy_uninformed}), with cut-offs $(\mathit{s}_0 ,\mathit{s}_1 )$ solving (\ref{equation0}) and (\ref{equation1}), $0< s_0 < s_1 $. Moreover, every equilibrium is a cut-off equilibrium, and every equilibrium is non trivial (meaning positive participation for both informed and uninformed traders in both venues).

The exchange spread, $A$, the expected execution rates, $\bar{R},\underline{R}$, are determined, respectively, by (\ref{spread}), (\ref{R_upper}), and (\ref{R_lower}). The set of participation fractions,  $\{\overline{\gamma_e},\underline{\gamma_e},\overline{\gamma_d},\underline{\gamma_d},\alpha_e, \alpha_d\}$ are determined by (\ref{gamma_upper}), (\ref{gamma_lower}), (\ref{beta_upper}), (\ref{beta_lower}), (\ref{alpha_e}), and (\ref{alpha_d}).
\label{thm existence}
\end{thm}

\begin{cor}
\textbf{(Liquidity begets liquidity)} $\alpha_d > 0$ if and only if $\overline{\gamma_d}-\underline{\gamma_d}> 0$.
\label{cor_liquidity begets liquidity}
\end{cor}

The equilibrium characterized in Theorem \ref{thm existence} is distinctive to \citet*{zhu_dark_2013} in the following aspects. First, in contrast with \citet*{zhu_dark_2013}, in equilibrium in our model, there is a \textit{sorting effect} of market fragmentation, and both uninformed and informed traders always move together. It is respectively optimal for informed traders with strong signals, modest signals, and weak signals to trade in the exchange, in the dark pool, and do not trade, whereas it is respectively optimal for uninformed traders with high, modest, and low degrees of impatience to trade in the exchange, in the dark pool, and delay trade. In \citet*{zhu_dark_2013}, however, such a sorting effect is absent for informed traders. In his model, informed traders are homogeneous and indifferently between venues. This may cause the instability of its prediction. For example, uninformed traders can collectively move from the dark pool to the exchange. This movement may increase the adverse selection cost in the dark pool so much so that they will stay in the exchange, and price discovery is strictly decreased. These equilibra are not discussed in \citet*{zhu_dark_2013}.  Our prediction is more robust in the sense that traders always move together and this sorting effect exists in every equilibrium. The same predictions on price discovery hold in every equilibrium.

Second,  unlike \citet*{zhu_dark_2013}, in which there exists some cases where informed traders do not participate in the dark pool, we predict that all equilibrium is non-trivial. That is, informed and uninformed participate in both venues in all equilibra, as captured in Corollary \ref{cor_liquidity begets liquidity}. This casts light on the dynamics of liquidity creation in a dark pool: informed and uninformed traders tend to arrive the dark pool in a clustered fashion, which in turn attract more liquidity to the dark pool, as documented in the literature.\footnote{\citet*{sarkar_liquidity_2009} provide a more detailed description of such process.} One explanation why \citet*{zhu_dark_2013} predicts a different result is that he assumes exact signals for traders. As we have pointed out, traders with strong signals tend to prefer an exchange. It is possible that, in some cases, they all crowd in the exchange and are absent in the dark pool. But again, this might be subject to an unstable status. In our model, this will not happen because with a noisy information structure, the dark pool will always be attractive to some informed traders. This is related to the following aspect.

The equilibrium described in Theorem \ref{thm existence} also disclose one important function of dark pools: a function that cannot be captured without a noisy information structure. That is, dark pools help to mitigate traders' information risk, i.e., the loss atributable to bad information. Dark pools take a role as a ``buffer zone'' for informed traders -- a gambling place for those who are less well-informed to trade. This adds value to the trade-off of dark pools, and shall clearly not be neglected. When information becomes noisier, more informed traders will find dark pools more valuable places to trade. Also, if traders become risk-averse, the importance of this function for dark pools will increase to a great extent.

\begin{cor}
Given any $\sigma_e, \sigma_v >0$, $s_1> \widehat{s}$, and in correspondent, $d_1> \widehat{d}$.
\label{cor_exchange_amplify}
\end{cor}

\begin{cor}
\textbf{(Adverse selection)}
$\forall \sigma_e, \sigma_v >0$,   $0<\underline{\gamma_e}<\overline{\gamma_e}$, $0<\underline{\gamma_d}<\overline{\gamma_d}$, and $\bar{R}-\underline{R}>0$.
\label{cor_adverse}
\end{cor}
\begin{proof}
If $\bm{\sigma}\in (0, +\infty)$, by Theorem \ref{thm existence}, $0<s_0<s_1$. Therefore by definition of (\ref{Belief}), and (\ref{spread}), (\ref{R_upper}), (\ref{R_lower}), (\ref{gamma_upper}), (\ref{gamma_lower}), (\ref{beta_upper}), (\ref{beta_lower}), (\ref{alpha_d}), it must be that $\frac{A}{\sigma_v}, \alpha_d, \alpha_e \in (0,1)$ and $0<\underline{\gamma_e}<\overline{\gamma_e}<1, 0<\underline{\gamma_d}<\overline{\gamma_d}<1$.  Therefore $0<\underline{R}<\bar{R}<1$.
\end{proof}

Corollary \ref{cor_exchange_amplify} states that dark pools strictly decrease traders' participation in the exchange. Corollary \ref{cor_adverse} states that there exists adverse selection in both the exchange and the dark pool. Market makers lose money to informed traders on average.

\section{Dark Pool Trading and Information Structure}
\label{sec_compstatics}
In this section, we restrict our attention to the following questions. These questions will be discussed in Section \ref{sec_cross section},  \ref{sec_cross model}, and \ref{sec_price discovery}, respectively.
\begin{itemize}
\item (i) How do each venue's market participation and information asymmetry level vary with the information structure, i.e., ``$\sigma_e$''?
\item (ii) How does adding a dark pool impact market participation and information asymmetry?
\item (iii) How does adding a dark pool impact price discovery, and what are the determinants?
\end{itemize}

\subsection{Information Precision and Market Characteristics}
\label{sec_cross section}
To recall, dark pools are of important value for informed traders who are less well-informed because they mitigate their informational risks. When information becomes more precise, such need decreases, and a migration of traders from one venue to another shall be observed. In this section, we study how the traders' participation and information asymmetry level in each venue vary with the informational structure. The results are shown in Proposition \ref{prop_monotonicity_lowsigma} and Proposition \ref{prop_monotonicity_lowsigmaparticipation}. The numerical example is in Figure  \ref{fig_participation}. We use $\sigma_e$ to capture the information precision for informed traders. A lower $\sigma_e$ corresponds with lower noises, hence a higher precision in their signals.

\begin{prop} \textbf{(Exchange spread, Dark pool adverse selection costs)}
If $\sigma_e$ is large, then both the exchange spreads and the dark pool adverse selection costs increase in information precision. That is, as $\sigma_e$ decreases,
    \begin{itemize}
      \item(Without DP): $\frac{A^{\mathbb{S}}}{\sigma_v}$ strictly increases;
      \item (With DP): Similarly, $\frac{A}{\sigma_v}$ increases, $\widehat{\bar{R}}-\widehat{\underline{R}}$ increases,
    \end{itemize}
\label{prop_monotonicity_lowsigma}
\end{prop}

\begin{prop} \textbf{(Participation rates)}
 Suppose $\sigma_e$ is large. Then for informed traders, as information precision increases, both the exchange and the dark pool participation increase. In contrast, for uninformed traders, as information precision increases, the exchange participation decreases while the dark pool participation increase. And total uninformed participation decreases. That is, as $\sigma_e$ decreases,
    \begin{itemize}
      \item(Without DP): $\overline{\gamma_e}^{\mathbb{S}} - \underline{\gamma_e}^{\mathbb{S}}$ strictly increases, and $\alpha_e^{\mathbb{S}}$ strictly decreases;
      \item (With DP): Similarly, $\overline{\gamma_e} - \underline{\gamma_e}$, $\overline{\gamma_d} - \underline{\gamma_d}$  increases, $\alpha_e$ decreases, $\alpha_d$ increases, and $\alpha_e+\alpha_d$ decreases.\footnote{$\overline{\gamma_e} - \underline{\gamma_e}$ and $\overline{\gamma_d} - \underline{\gamma_d}$ capture the ``meaningful'' participation of informed trades, in the sense that they are the fractions of informed trades that trade in the ``right'' direction net the fractions that trade in the ``wrong'' direction.}
    \end{itemize}

\label{prop_monotonicity_lowsigmaparticipation}
\end{prop}

\begin{rmk}
when $\sigma_e$ is large, as in Proposition \ref{prop_monotonicity_lowsigma} and Proposition \ref{prop_monotonicity_lowsigmaparticipation} , dark pool participation for informed traders and dark pool adverse selection cost INCREASES with information precision. When $\sigma_e$ is small, however, they may DECREASE with information precision. We have not been able to obtain comparative statics when $\sigma_e$ is small, but we show this inverted U-shape in the numerical example in Figure  \ref{fig_participation}.\footnote{In all our plots, we use a set of parameters in which $\mu_z = 60, \mu = 30$, $Z^+, Z^-$ has Gamma distributions with mean $30$ and variance $30$ and  $G(d)=\frac{d}{3}$ for $\bar{d} \in [0,3]$.} While we provide an explanation in the context, the explicit proof is of future work.
\label{remark1}
\end{rmk}

\begin{figure}[h]
\centering
\includegraphics[width=1\textwidth]{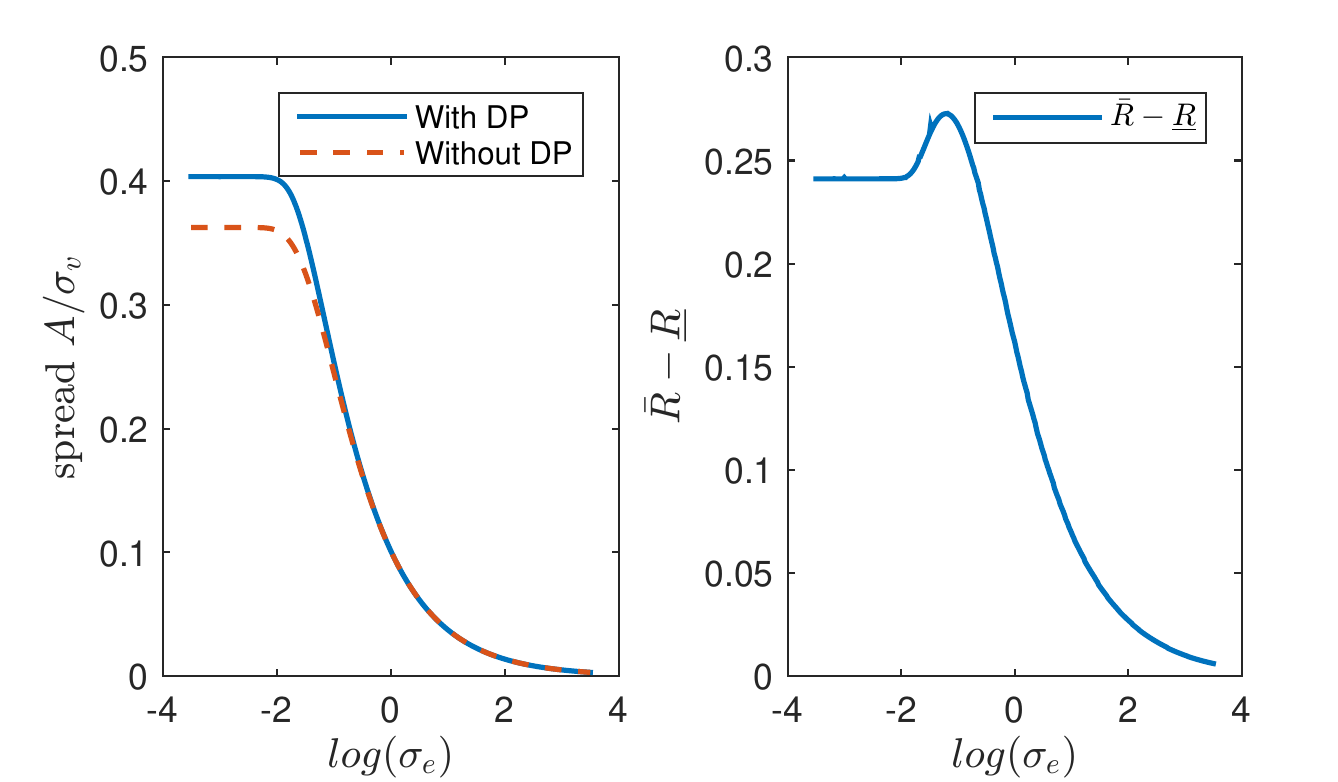}
\caption{\footnotesize \textbf{Transaction Costs.} The left-hand figure shows the normalized spreads on the exchange and how they vary with $\log(\sigma_e)$; the right-hand figure shows the adverse selection cost in the dark pool and how it vary with $\log(\sigma_e)$. In both figures, $\log(\sigma_v) =0$.}
\label{fig_spread}
\end{figure}

\begin{figure}[h]
\centering
\includegraphics[width=1\textwidth]{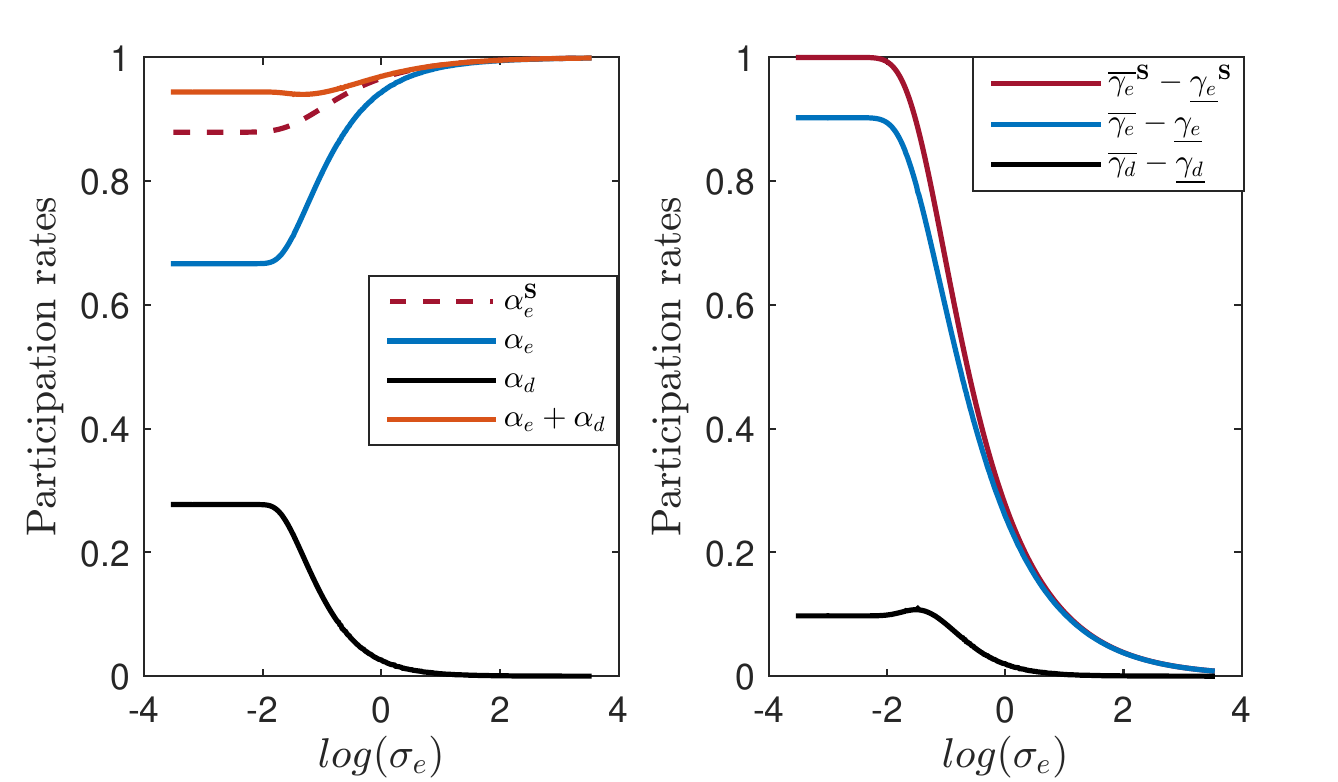}
\caption{\footnotesize \textbf{Participation Rates.} The left figure plots the expected participation rates of the uninformed and how they vary with $\log(\sigma_e)$. The right one shows the participation rates for the informed traders how they vary with $\log(\sigma_e)$. In both plots, $\log(\sigma_v) =0$, $\mu_z = 60, \mu = 30$.}
\label{fig_participation}
\end{figure}

In the exchange, when signals become more precise, both the informed exchange participation, $\overline{\gamma_e} - \underline{\gamma_e}$, and exchange spread, $A$, increase, whereas the uninformed exchange participation, $\alpha_e$, decreases. The intuition is as follows. In equilibrium the informed traders are sorted by the strengths of their signals. when there is an increment in their information precision, the overall strengths of their signals are increased. Therefore, some informed traders migrate from ``do not trade'' to ``trade in the dark pool'' and from ``trade in the dark pool'' to ``trade in the exchange.'' This will cause a strict increase of information asymmetry level in the exchange, and hence an increase of the exchange spread. Consequently, some liquidity traders migrate from ``trade in the exchange'' to ``trade in the dark pool,'' which decreases the uninformed participation in the exchange.

In the dark pool, the dark pool informed participation, $\overline{\gamma_d} - \underline{\gamma_d}$,  and  the dark pool adverse selection, $\widehat{\bar{R}}-\widehat{\underline{R}}$, exhibit an inverted U-shape with information precision. The intuition for the inverted U-shape is as follows. A change in the information precision changes the distribution of the signals' strengths. When the information precision level is low (i.e., $\sigma_e$ is high), as the precision grows, signals become more concentrated in the relative ``modest'' group, and more informed traders migrate from ``do not trade'' to the dark pool. Overall, this induces a greater proportion of informed participation in the dark pool, and the dark pool adverse selection increases. In contrast, when the information precision level is high (i.e., $\sigma_e$ is low), as precision grows, signals become more concentrated in the relative ``strong'' group. Thus, more informed traders migrate from the dark pool to the exchange, leaving a lower proportion of informed trades in the dark pool, and the dark pool adverse selection decreases.

An interesting comparison with \citet*{zhu_dark_2013} is that, although \citet*{zhu_dark_2013} does not consider the information structure, he discusses the comparative statics of market behaviors as a function of $\sigma_v$. $\sigma_v$ and $\sigma_e$ are comparable in the sense that, all else equal,  informed traders' information advantage increases in both information precision (i.e., as $\sigma_e$ decreases), and the asset value uncertainty (i.e., as $\sigma_v$ increases, see a more detailed discussion in Section \ref{sec_price discovery}).

We highlight two major differences between our predictions and those of \citet*{zhu_dark_2013}. First, our model predicts that traders' participation exhibits a smooth variation cross-sectionally (i.e., when $\sigma_v$ grows), whereas there is a discontinuity in that of \citet*{zhu_dark_2013}. In \citet*{zhu_dark_2013}, in equilibrium informed traders don't trade in dark pools for some assets unless the asset's value uncertainty is high (i.e., $\sigma_v$ is high). In contrast, we predict that both informed and liquidity traders trade in dark pools in a clustering fashion, regardless of $\sigma_v$.  This is a more realistic prediction. If there are some assets for which dark pools only attract liquidity traders, one would expect a persistent gap between the average size of dark pools and the average size of lit markets. Yet, this is not true as we observe in Figure \ref{fig_ussize}.  This, again, emphasizes that dark pools function as informational risk mitigators and that they are always lucrative for traders, informed or uninformed.

Second, \citet*{zhu_dark_2013} predicts that informed traders' participation in dark pool always squeezes out liquidity traders (i.e., $\alpha_d$ decreases as informed trades grow in the dark pool), whereas we predict that the two can grow simultaneously, especially when informed traders' information is relatively imprecise. The explanation is that the informed trading intensity in the dark pool is always high in \citet*{zhu_dark_2013} because traders have exact information. But in our model, the intensity is neutralized to some extent because some speculators trade in the ``wrong'' direction.

\subsection{Dark Pool Impacts on Market Characteristics}
\label{sec_cross model}
In this section, we study how the market responds when a dark pool is added alongside an exchange. Precisely, we compare the equilibrium traders' participation and exchange spread between the two models:  the ``Single-venue'' model and the ``Multi-venue'' model. In the comparison, we fixed the information structure (i.e., $\sigma_e$). The result is shown in Proposition \ref{prop_compare}. This result coincides    with \citet*{zhu_dark_2013}, except that the effect on the exchange spread $A$ is uncertain when information is imprecise (i.e., $\sigma_e$ high).

\begin{prop}
 Given any $\sigma_v, \sigma_e>0$, then adding a dark pool alongside an exchange a) \textbf{\textit{(Participation)}}: decreases the participation in the exchange for both informed and uninformed traders, but increases the total market participation, and b) \textbf{\textit{(Exchange spread)}}: widens the spread on the exchange, if information precision is high ($\sigma_e$ is small).

 That is,  suppose $\frac{\mu_z}{\mu}\geq \frac{R}{1-R}\frac{1}{1-G(\widehat{k})}$ where  $R = \mathbb{E} \left[\min\left\{1, \frac{R^+}{R^-} \right\} \right]$, and $\widehat{k}$ is uniquely determined by $\widehat{k}=\frac{1}{1+[1-G(\widehat{k}]\frac{\mu_z}{\mu}}$
 then
 \begin{itemize}
 \item  (i) $(\overline{\gamma_e}^{\mathbb{S}}-\underline{\gamma_e}^{\mathbb{S}})\geq (\overline{\gamma_e}-\underline{\gamma_e})$, $\alpha_e^{\mathbb{S}} \geq \alpha_e$, and if $\sigma_e$ is sufficiently small or large, $\alpha_e^{\mathbb{S}} \leq \alpha_e +\alpha_d$. And,
 \item (ii) $\frac{A^{\mathbb{S}}}{\sigma_v}\leq \frac{A}{\sigma_v}$ if $\sigma_e$ is small.
 \end{itemize}
\label{prop_compare}
\end{prop}

\begin{rmk}
When information precision is high ($\sigma_e$ is low), as in Proposition \ref{prop_compare}, we proved that $\frac{A^{\mathbb{S}}}{\sigma_v}\leq \frac{A}{\sigma_v}$ (i.e., adding a dark pool WIDENS the exchange spread). When information precision is low ($\sigma_e$ is high),  however, it is possible that $\frac{A^{\mathbb{S}}}{\sigma_v}> \frac{A}{\sigma_v}$ (i.e., adding a dark pool NARROWS the exchange spread ).\footnote{When $\sigma_e$ is large, it is either $\frac{A^{\mathbb{S}}}{\sigma_v}<\frac{A}{\sigma_v}$ when $\sigma_e$ is large, or undetermined (in which,  as $\sigma_e\rightarrow +\infty$, $\frac{A^{\mathbb{S}}}{\sigma_v}$ equals $\frac{A}{\sigma_v}$, and their first order derivatives with respective to $\sigma_e$ are equal.  )} This could be caused by the fact that, in these cases, the informed traders have moved to dark pools so much that the information asymmetry level in the exchange has deceased dramatically. While we discuss this briefly in Appendix \ref{proof_prop_compare}, the explicit analysis is of future work.
\label{remark2}
\end{rmk}

Proposition \ref{prop_compare} states that adding a dark pool will decrease informed and uninformed traders' exchange participation but increase the total participation. Thus, dark pools create additional liquidity. This, again, is explained by the migration of traders. Because adding a dark pool enlarges the opportunity sets for both informed and uninformed traders, there will be migrations of both types of traders from both ``Not trade'' and ``trade in the exchange'' to ``trade in the dark pool.''  Therefore, the dark pool attracts not only additional liquidity but also part of the liquidity from the exchange. As a consequence, the exchange participation decreases, but the total participation of traders increases. This is captured in figure \ref{fig_participation} in which $\alpha_e\leq \alpha^{\mathbb{S}}_e \leq \alpha_e+\alpha_d$.

The impact of a dark pool to the exchange spread, however, is not straightforward. The spread depends on the level of information asymmetry in the exchange, which in turn depends on the intensity of informed and uninformed trades. As we have pointed out, the addition of a dark pool induces an outflow of both informed and uninformed traders. The resulting proportion of the two in the exchange depends on which overwhelms the other. When the informed traders have high information precision (i.e., low $\sigma_e$), a large fraction of them strictly prefers to stay in the exchange, and only a small fraction will migrate to the dark pool, compared with the migration of uninformed traders. As a result, the exchange information asymmetry strictly increases and exchange spread, ``$\frac{A}{\sigma_v}$,'' is enlarged. When the informed traders have low precision in their information (i.e., $\sigma_e$ is high), however, there is a large fraction of the informed who prefer to migrate to the ``buffer zone,'' the dark pool, and the relative proportion of informed traders in the exchange decreases. As a result, the exchange spread may or may not decrease, depending on how intense the migration is.\footnote{Note that $\frac{A}{\sigma_v}$ depends on both $\frac{\overline{\gamma_e}-\underline{\gamma_e}}{\overline{\gamma_e}+\underline{\gamma_e}}$ and $\frac{\overline{\gamma_e}-\underline{\gamma_e}}{\alpha_e}$, when $\sigma_e$ is large, $\frac{\overline{\gamma_e}-\underline{\gamma_e}}{\alpha_e}$ decreases when adding a dark pool but not necessarily $\frac{\overline{\gamma_e}-\underline{\gamma_e}}{\overline{\gamma_e}+\underline{\gamma_e}}$. The overall effect on $\frac{A}{\sigma_v}$ is uncertain. }
\subsection{Dark Pool Impacts on Price Discovery}
\label{sec_price discovery}
Price discovery is measured by the informativeness of $P_1$. At the end of period 1, the market maker observes the period 1 exchange order flows $V_b, V_s$, which respectively represents the ``buy'' volume and the ``sell'' volume and announces a  closing price $P_1 = \mathbb{E}[\tilde{v}|V_b, V_s]$.  $P_1$ is perceived as a proxy for the fundamental value of the asset. This is so because $\mathbb{E}[\tilde{v}|P_1, V_b, V_s]=\mathbb{E}[\mathbb{E}[\tilde{v}|P_1, V_b, V_s]|P_1]=P_1$. We are interested in how informative $P_1$ is, that is, how close $P_1$ is to the true value of the asset.

We consider similar measures as suggested by \citet*{zhu_dark_2013}. Without loss of generality, we assume that the true value $\tilde{v} = +\sigma_v$. Let the likelihood ratio
$$r = \log \frac{\Pr (\tilde{v}=+\sigma_v|V_b, V_s)}{\Pr (\tilde{v}=-\sigma_v|V_b, V_s)}=
\log \frac{\phi_z(Z^+ = \frac{1}{\alpha_e} [V_b - \overline{\gamma_e} \mu] )\cdot \phi_z(Z^- = \frac{1}{\alpha_e} [V_s - \underline{\gamma_e}\mu] )}{\phi_z(Z^- = \frac{1}{\alpha_e} [V_b - \overline{\gamma_e} \mu] )\cdot \phi_z(Z^+ = \frac{1}{\alpha_e} [V_s - \underline{\gamma_e}\mu] )}.$$
And
\begin{align*}
P_1 &= \sigma_v \Pr(\tilde{v}=+\sigma_v|V_b, V_s)+ (-\sigma_v)\Pr(\tilde{v}=-\sigma_v|V_b, V_s)\\
&=\frac{\Pr (\tilde{v}=+\sigma_v|V_b, V_s) -\Pr (\tilde{v}=-\sigma_v|V_b, V_s)}{\Pr (\tilde{v}=+\sigma_v|V_b, V_s) + \Pr (\tilde{v}=-\sigma_v|V_b, V_s)} \sigma_v
\end{align*}
Therefore
\begin{align*}
P_1 =\frac{e^r -1}{e^r +1} \sigma_v.
\end{align*}
 Clearly, if $r$ is higher, $P_1$ is closer to the true value $\sigma_v$. If $r=+\infty$, then $P_1 = \sigma_v$, in which case $P_1$ is completely informative.  Therefore, $r$ can be considered as a measure of the informativeness.

Another measure of informativeness that we consider is the scaled root-mean-squared error (RMSE), in which
$$\mbox{RMSE}=\frac{\left[\mathbb{E}[(\tilde{v}-P_1)^2|\tilde{v}=\sigma_v]\right]^{.5}}{\sigma_v} = \mathbb{E}\left[\frac{4}{(e^r +1)^2}|\tilde{v}=\sigma_v\right].$$
It is scaled by $\sigma_v$. Since $r\in (0,1)$, the scaled pricing error (RMSE) is between 0 and 1. If RMSE is higher, there are more pricing errors, and there is less price discovery.

Since $V_b, V_s$ are random variables, $r$ is also a random variable. When $\mu_z, \sigma_z^2$ are large enough, we can approximate the density of $\phi_z(\cdot)$ by a normal distribution $\mathcal{N}(.5\mu_z, .5\sigma_z^2)$.\footnote{We use the same approximation as in \citet*{zhu_dark_2013}, in which it shows that when $\mu_z$ and $\sigma_z^2$ are large enough, $Z^+$ is approximately normal. } Substituting the density functions, we get an approximate $r$ by
$$r^{Approx} = \frac{2(\overline{\gamma_e}-\underline{\gamma_e})\mu}{\alpha_e^2 \sigma_z^2}(V_b-V_s).$$

Given that $\tilde{v} = \sigma_v$, Since $V_b-V_s$ has a distribution of $\mathcal{N}\left((\overline{\gamma_e}-\underline{\gamma_e})\mu , \alpha_e^2 \sigma_z^2 \right)$,  so $r^{Approx}$ has a distribution  of
$$\mathcal{N}\left(2 \mathcal{I}(\overline{\gamma_e}, \underline{\gamma_e},\alpha_e)^2,  4 \mathcal{I}(\overline{\gamma_e}, \underline{\gamma_e}, \alpha_e)^2 \right) ,$$
where
$$\mathcal{I}(\overline{\gamma_e}, \underline{\gamma_e},\alpha_e) = \frac{(\overline{\gamma_e}-\underline{\gamma_e})\mu}{\alpha_e \sigma_z}.$$

Thus, the magnitude of $\mathcal{I}(\overline{\gamma_e}, \underline{\gamma_e},\alpha_e)$ can be taken as a measure of the price discovery in the exchange.  To be consistent with definitions of \citet*{zhu_dark_2013}, we also refer to it as ``signal-to-noise'' ratio. We consider two measures of price discovery: the signal-to-noise ratio $\mathcal{I}(\overline{\gamma_e}, \underline{\gamma_e},\alpha_e)$ and the scaled RMSE under the normal approximation.

By the same argument as \citet*{zhu_dark_2013}, under the normal approximation, a higher signal-to-noise ratio $\mathcal{I}(\overline{\gamma_e}, \underline{\gamma_e},\alpha_e)$ always corresponds to a lower scaled RMSE. That is, they are in nature the same measure.  Therefore, we only plot the ``signal-to-noise'' in our numerical example in Figure \ref{fig_price discovery-sigmav}.

We introduce a measure for the informed traders: that is, the measure of their ``information advantage'':
$$\bm{\sigma} = \frac{\sigma_v}{\sigma_e}.$$
An informed speculator's ``information advantage'' is defined as the asset's fundamental uncertainty $\sigma_v$ times the precision of the signals $\frac{1}{\sigma_e}$. Clearly, a higher $\sigma_v$ reflects a high level of undisclosed information, therefore, a higher profitability of the informed speculators. Also, a lower $\sigma_e$ means a higher precision of the private information, and hence a higher informational profit. Proposition \ref{prop_informativeness} summarizes the price discovery as a function of $\bm{\sigma}$ and the impact of a dark pool to price discovery.

\begin{prop} Price discovery (i.e. the informativeness of $P_1$) in the exchange is an increasing function of informed traders' ``information advantage'' ($\bm{\sigma}$). And, there exists a threshold, $\bm{\bar{\sigma}}>0$, such that, a) when $\bm{\sigma}<\bm{\bar{\sigma}}$, adding a dark pool impairs price discovery, and b) when $\bm{\sigma}$ is large, adding a dark pool enhances price discovery.

That is, suppose $\widehat{k}\leq \frac{\mu_z}{\mu}<+\infty$, where $\widehat{k}$ is uniquely determined by $\widehat{k}=\frac{1}{1+[1-G(\widehat{k}]\frac{\mu_z}{\mu}}$, then $\mathcal{I}(\overline{\gamma_e}, \underline{\gamma_e},\alpha_e)$, $\mathcal{I}(\overline{\gamma_e}^{\mathbb{S}}, \underline{\gamma_e}^{\mathbb{S}},\alpha_e^{\mathbb{S}})$ increase in $\bm{\sigma}$ and RMSE, $\mbox{RMSE}^{\mathbb{S}}$ decrease in $\bm{\sigma}$, when $\sigma_e>0$ is large enough, and $\exists \bm{\bar{\sigma}} >0$ such that
\begin{itemize}
\item (i)   if $ \bm{\sigma}\in (0,\bm{\bar{\sigma}})$, adding a dark pool will strictly decrease the informativeness of the price in exchange, that is, $\mathcal{I}(\overline{\gamma_e}, \underline{\gamma_e},\alpha_e) < \mathcal{I}(\overline{\gamma_e}^{\mathbb{S}}, \underline{\gamma_e}^{\mathbb{S}},\alpha_e^{\mathbb{S}})$  , and $ \mbox{RMSE}>\mbox{RMSE}^{\mathbb{S}} $
\item (ii) if $\bm{\sigma}$ is sufficiently large,  adding a dark pool will increase the informativeness of the price in exchange, that is, $\mathcal{I}(\overline{\gamma_e}, \underline{\gamma_e},\alpha_e) \geq \mathcal{I}(\overline{\gamma_e}^{\mathbb{S}}, \underline{\gamma_e}^{\mathbb{S}},\alpha_e^{\mathbb{S}})$, and $ \mbox{RMSE} \leq\mbox{RMSE}^{\mathbb{S}}$

\end{itemize}
\label{prop_informativeness}
\end{prop}

When a dark pool is added alongside an exchange, the impact on price discovery is depending on the resulting ratio of informed traders and uninformed traders in the exchange. As we have discussed in Section \ref{sec_cross model}, when a dark pool is introduced to the market, it induces migrations of both informed traders and liquidity traders from the exchange to the dark pool. When $\bm{\sigma}$ is high, on average, informed traders have high profitability, a high proportion of the informed would rather stay in the exchange, and only a small proportion migrate from the exchange to the dark pool, compared with the liquidity traders. Therefore adding a dark pool increases the ``signal-to-noise'' ratio and improves the informativeness of $P_1$ in the exchange. When $\bm{\sigma}$ is low, however, on average the informed have low profitability so that a higher proportion would rather migrate from the exchange to trade in the ``buffer zone,'' the dark pool, compared with the liquidity traders. This leaves a lower proportion of informed traders in the exchange. The ``signal-to-noise'' ratio decreases and price discovery declines.

In Figure \ref{fig_price discovery-sigmav}, the right plots ``signal-to-noise'' ratio as a function of $\bm{\sigma}=\frac{\sigma_v}{\sigma_e}$. It increases with $\bm{\sigma}$, indicating that informed traders' trading intensity grows with higher ``informational advantage,'' and hence price discovery increases. Introducing a dark pool alongside an exchange decreases price discovery when $\bm{\sigma}$ is low (i.e., $\sigma_v$ is low or $\sigma_e$ is high), and increases when $\bm{\sigma}$ is high (i.e., $\sigma_v$ is high or $\sigma_e$ is low). The left further illustrates the dark pool impact on price discovery in a 2-dimensional context (i.e., $\sigma_v$ and $\sigma_e$). .
\begin{figure}[h!]
\centering
\begin{subfigure}{.45\textwidth}
\centering
\includegraphics[width=1\textwidth]{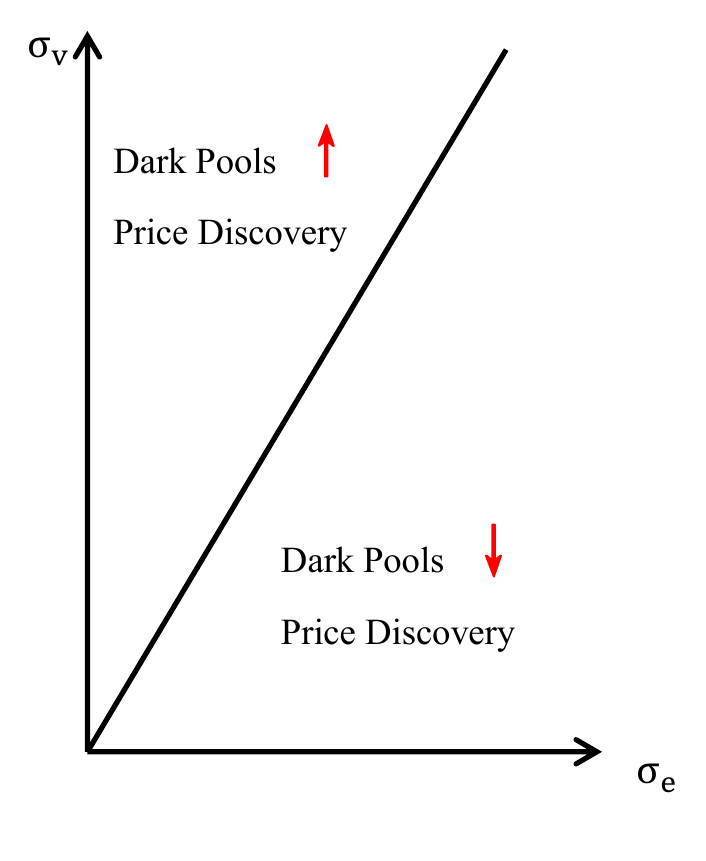}
\end{subfigure}
\begin{subfigure}{.45\textwidth}
\centering
\includegraphics[width=1\textwidth]{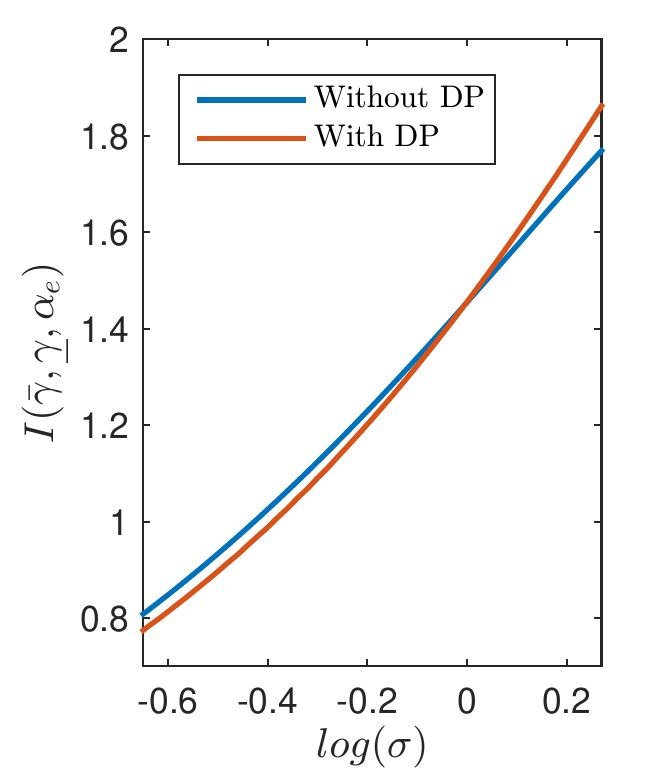}
\end{subfigure}
\caption{\footnotesize \textbf{Price Discovery.}  The left plots the dark pool impact on price discovery with 2-dimension: $\sigma_v$ and $\sigma_e$. The right plots the ``Signal-to-noise'' ratio $I(\overline{\gamma_e},\protect\underline{\gamma_e},\alpha_e)$ as a
function of $\bm{\sigma}=\frac{\sigma_v}{\sigma_e}$.
}
\label{fig_price discovery-sigmav}
\end{figure}

The results highlight an important effect dark pools have on price discovery -- an \textit{amplification Effect} That is, dark pools enhances price discovery when it is high, whereas dark pools impairs price discovery when it is low. An economy needs to be prudent in introducing dark pools to its equity market, especially when the economy has a poor information environment (low quality in information disclosure, poor legal systems and enforcement, etc.) We provide a more detailed discussion in  Section \ref{sec_policy}.

This result is in contrast with \citet*{zhu_dark_2013}, in which adding a dark pool strictly increases the price discovery. According to our analysis, the important reason \citet*{zhu_dark_2013} predicts a strict increase is due to the fact that it assumes an extreme case where signals for informed traders are perfect (i.e., $\sigma_e \rightarrow 0$ in our model). As we have pointed out, when information is in high precision (i.e., $\sigma_e$ is low), the majority of the informed traders prefer the exchange, where dark pools will attract relatively less fraction informed traders from the exchange, compared with the liquidity traders, and leave a higher ratio of informed-to-uninformed traders in the exchange, hence improve price discovery. Thus, \citet*{zhu_dark_2013} is consistent with our prediction. In reality, however, \citet*{zhu_dark_2013}'s prediction may not hold because the information structure is much richer and exhibits significant cross-sectional difference (we will discuss this in Section \ref{sec_policy}). Policies and measures should be tailored to this issue in a different information environment.

\citet*{zhu_dark_2013} also depicts a scenario when uninformed liquidity trader types are discrete. It shows that in this case, to a large degree, price discovery will be harmed by the introduction of dark pools because uninformed traders of discrete types are more likely to get ``stuck'' in their original venues while some informed traders flow from the exchange to dark pools and decrease price discovery. Our prediction corresponds to this scenario. In our prediction, the discrete type and ``stickiness''of uninformed traders will further increase the chance that price discovery be harmed.

\textbf{Determinants of the impact.} From the perspective of a regulator, when introducing dark pools, an important issue is what fraction of the assets will be harmed in their price discovery. In order to answer that question, one should examine the determinants and the overall impact dark pools have on price discovery.

We consider a proxy which we refer to as the ``likelihood that dark pools harm price discovery.''
$$\bar{\sigma}_v=\underset{x>0}{\sup} \left\{x |\forall \sigma_v\in(0,x) ,  \mathcal{I}(\overline{\gamma_e}^{\mathbb{S}}, \underline{\gamma_e}^{\mathbb{S}},\alpha_e^{\mathbb{S}})>\mathcal{I}(\overline{\gamma_e}, \underline{\gamma_e},\alpha_e)\right\}.$$
By Proposition \ref{prop_informativeness}, such $\bar{\sigma}_v$ must exist.  A higher $\bar{\sigma}_v$ reflects a higher fraction of assets whose price discovery will be harmed by adding a dark pool.

We consider two determinants. The first is the precision of traders' private information, the inverse of $\sigma_e$. Proposition \ref{prop_informativeness} indicates that the likelihood dark pools harm price decreases with precision level. Another determinant we consider is the relative measure of informed traders, $\frac{\mu}{\mu_z}$. The effects of the two on $\bar{\sigma}_v$ is summarized in Proposition \ref{prop_sigma_tendency}.

\begin{prop}
the likelihood that price discovery will be harmed by dark pool trading ($\bar{\sigma}_v$) decreases in information precision, $\sigma_e$, and increases in the relative measure of informed traders, $\frac{\mu}{\mu_z}$.

That is, Suppose $\widehat{k}\leq \frac{\mu_z}{\mu}<+\infty$, where $\widehat{k}$ is uniquely determined by $\widehat{k}=\frac{1}{1+[1-G(\widehat{k}]\frac{\mu_z}{\mu}}$, then
 \begin{itemize}
 \item (i) $\bar{\sigma}_v$ increases in $\sigma_e$.  As $\sigma_e \rightarrow 0^+$, $\bar{\sigma}_v \rightarrow 0$, and as $\sigma_e \rightarrow +\infty$, $\bar{\sigma}_v \rightarrow +\infty$. And,
 \item (ii)  for any sequence of $\{(\frac{\mu}{\mu_z})\}$, there exists a subsequence $\{(\frac{\mu}{\mu_z})_n\}$ such that as $(\frac{\mu}{\mu_z})_n$ increases, $\bar{\sigma}_v$ increases, also,  as $(\frac{\mu}{\mu_z})_n \rightarrow 0^+$, $\bar{\sigma}_v \rightarrow 0$.\footnote{We cannot directly show that $\bar{\sigma}_v$ increases in $\frac{\mu}{\mu_z}$, but we are able to show a upper bound of $\bar{\sigma}_v$ that is increasing in $\frac{\mu}{\mu_z}$.}
 \end{itemize}
\label{prop_sigma_tendency}
\end{prop}

The numerical example is given in Figure \ref{fig_barsigma}. Proposition \ref{prop_sigma_tendency} states that dark pools are beneficial for price discovery in an economy with a good information environment (i.e., high information precision and low size of informed traders), whereas they are bad for price discovery in an economy with a poor information environment (i.e., low information precision and high size of informed traders). Proposition \ref{prop_sigma_tendency} gives regulators insights into how to improve the economy and informativeness of prices. Policies and measures can be taken to enhance the market performance. Also, it points out important considerations for countries that are going to allow dark pools and provides them a benchmark to measure market quality. More details are in Section \ref{sec_policy}.

\begin{figure}[h]
\centering
\includegraphics[width=1\textwidth]{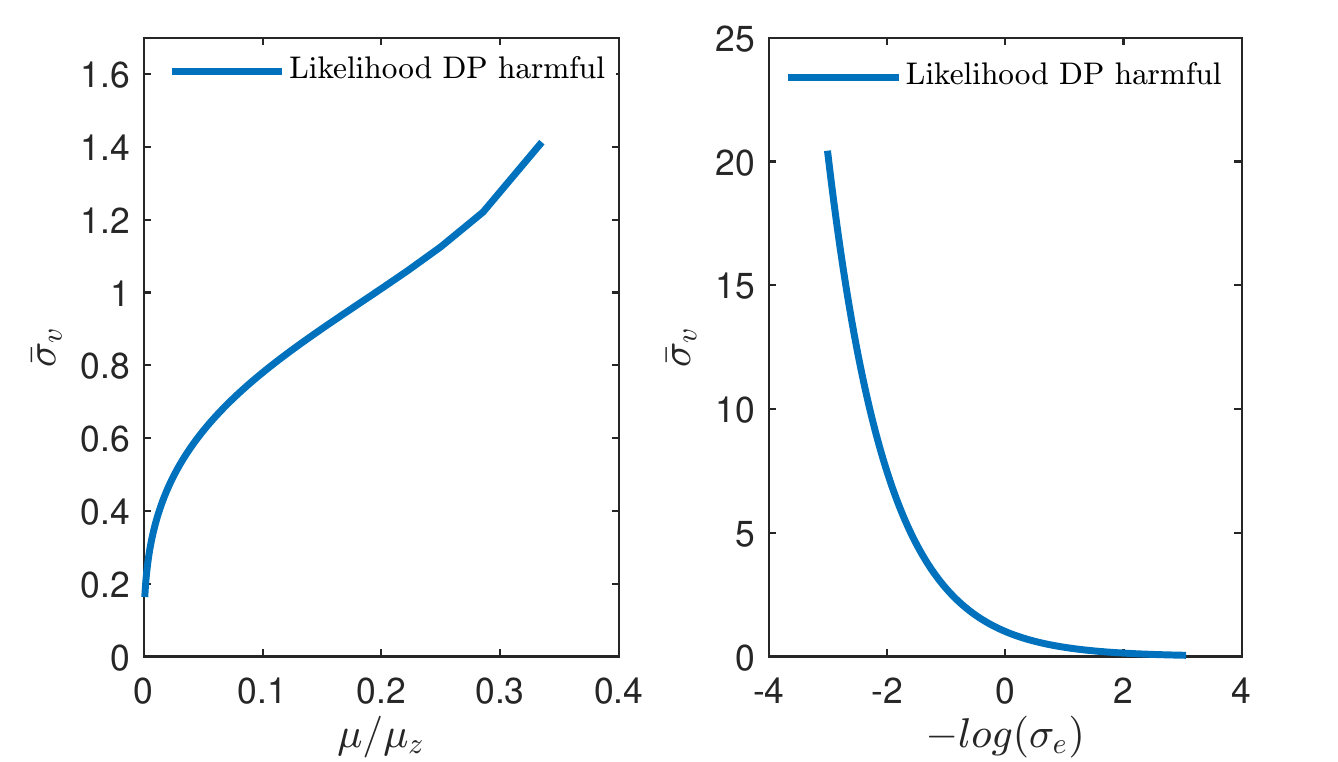}
\caption{\footnotesize \textbf{The Likelihood DPs Harm Price Discovery $\bar{\sigma}_v$}. The left-hand figure plots the threshold $\bar{\sigma}_v$ as a function of the relative size of informed traders, $\frac{\mu}{\mu_z}$.  The right-hand shows the threshold $\bar{\sigma}_v$ as a function of the information precision, $-\log(\sigma_e)$. On the left, $\log(\sigma_e)=0$. On the right, $\frac{\mu}{\mu_z}=.2$. }
\label{fig_barsigma}
\end{figure}

\section{Discussion: Empirical and Regulatory}
\label{sec_implication}

In this section, we provide a discussion about empirical implications and policy suggestions. The discussion is intended to provide insight into seemingly contradictory results in the empirical literature, as well as give exploration of channels for future research and regulatory concern. In these analyses, the economic force we consider is the variation of the  information structure, more precisely, the informed traders' ``information advantage,'' $\bm{\sigma}=\frac{\sigma_v}{\sigma_e}$, or the information imprecision,  $\sigma_e$, if $\sigma_v$ is fixed. We refer to ``good information environment'' by more precise information and less informed traders. Although we attempt to attribute the difference of the findings to the different information structures, we preserve a conservative interpretation in these predictions. In general, our model suggests that dark pool activity and its impacts display significant cross-sectional variation and thus should be evaluated differently in various economic environments.

\subsection{Measurement for information precision}
This paper shows that the level of information precision is essential to determining the impact of dark pools on price discovery. In this section, we provide a brief discussion about the measurement for information precision and information environment.

In this paper, there are two factors to consider for an economy's information environment: the precision of (private) information and the number of informed traders. The notion of a better information environment includes a higher precision in traders' (private) information and fewer informed traders.\footnote{Private information is not necessarily insider information. A big fraction of it is information that is publicly available but hard to collect, transmit, and process by the majority of the public.} We discuss the respective measurement for information precision of individual stocks and for the information environment of the whole economy. The former helps us to conduct the cross-sectional analysis for individual stocks while the latter gives regulators guidance on regulating dark pool trades as a whole.

\begin{itemize}
\item (1) For individual stocks, the measurement for the level of information precision include the following aspects.

\textbf{Firm characteristics}. Researchers have found that firms with greater growth volatility (such as high R\&D firms, young firms), smaller size, or fewer analyst followers have lower informational precision in traders' informational predictions \citep{li_information_2012, maffett_financial_2012, lang_cross-sectional_1993, baginski_determinants_1997}. Therefore, a firm's age, size, number of analyst followers, R\&D ratios, and other measures regarding its growth volatility can be used as proxies for the firm's (equity's) level of information precision. In addition to that, researchers use analysts' forecast variation and errors to proxy the level of information precision among traders  \citep{botosan_role_2004, gleason_analyst_2003}.

\textbf{Information acquisition and processing}. The activity of information acquisition and processing greatly affect the level of information precision. For example, \citet*{frankel_determinants_2006} has found that the informativeness of analysts' reports greatly depends on how lucrative the trades' are and how costly is the information acquisition.  Generally, a more lucrative and matured financial market, with more competition, more innovation in trading technologies, and years of trade has a higher level of trader ability for information acquisition and processing \citep{louis_conservatism_2014, clement_analyst_1999, chen_investor_2005}. Therefore, measures about the firms' profitability,  maturity, industry competition, level of innovation, and number (and years) of traders can be used as proxies for individual stocks' information precision.

\item (2) For the measurement of information environment in the macro-setting, measurements include measures on the strength of legal institutions and law enforcement against insider trading, the functionality of the public disclosure system, and the availability and efficiency of media transmission. Generally, public disclosure and media channels can enhance the precision of informed traders' forecasts\footnote{Although there is a debate regarding the association between public and private information, researchers generally find that public disclosures may be processed into private information by informed investors, and there is a positive correlation between the precisions of public and private information. See \citet*{botosan_role_2004} and \citet*{kim_market_1994}.} and stronger legal systems can significantly reduce the number of insiders.
\end{itemize}

\subsection{A Summary of Testable Empirical Predictions}
\label{sec_testables}
\textbf{1. Dark pool execution probability.} We predict that dark pool non-execution probability increases with information precision (i.e. $1-\frac{\bar{R}+\underline{R}}{2}$ increases as $\sigma_e$ decreases). Also, an asset's exchange spread increases with its dark pool non-execution probability (i.e. $\frac{A}{\sigma_v}$ increases in $1-\frac{\bar{R}+\underline{R}}{2}$).

This prediction suggest that the trade-off of dark pools is higher in an economy with a good information environment. The trade-off is documented in many empirical papers. For example, \citet*{gresse_effect_2006}, \citet*{conrad_institutional_2003}, \citet*{naes_equity_2005} , and \citet*{ye_non-execution_2010} study crossing networks in the US and conclude that dark pools, in comparison with exchanges, have lower trading costs (within spread price) but higher non-execution probability. \citet*{he_determinants_2006} studied Australia's Centre  Point  dark  pool and found that the dark pool execution probability increases with dark pool activity. In contrast, \citet*{kwan_trading_2014} find that the dark pool execution probability increases in the trading friction in exchanges: the minimal price improvement.

The change of execution probability can be explained as follows: the execution depends on two factors: traders' total participation and dark pool information asymmetry level. The former irons the difference between the two sides in the pool and increases the execution rate, whereas the latter does the opposite. In the numerical example in Figure \ref{fig_dp_tradeoff}, we show that, without pricing frictions in the exchange, the expected dark pool execution rate decreases as the information becomes more precise ($\sigma_e$ decreases).
\begin{figure}[h]
\centering
\includegraphics[width=.75\textwidth]{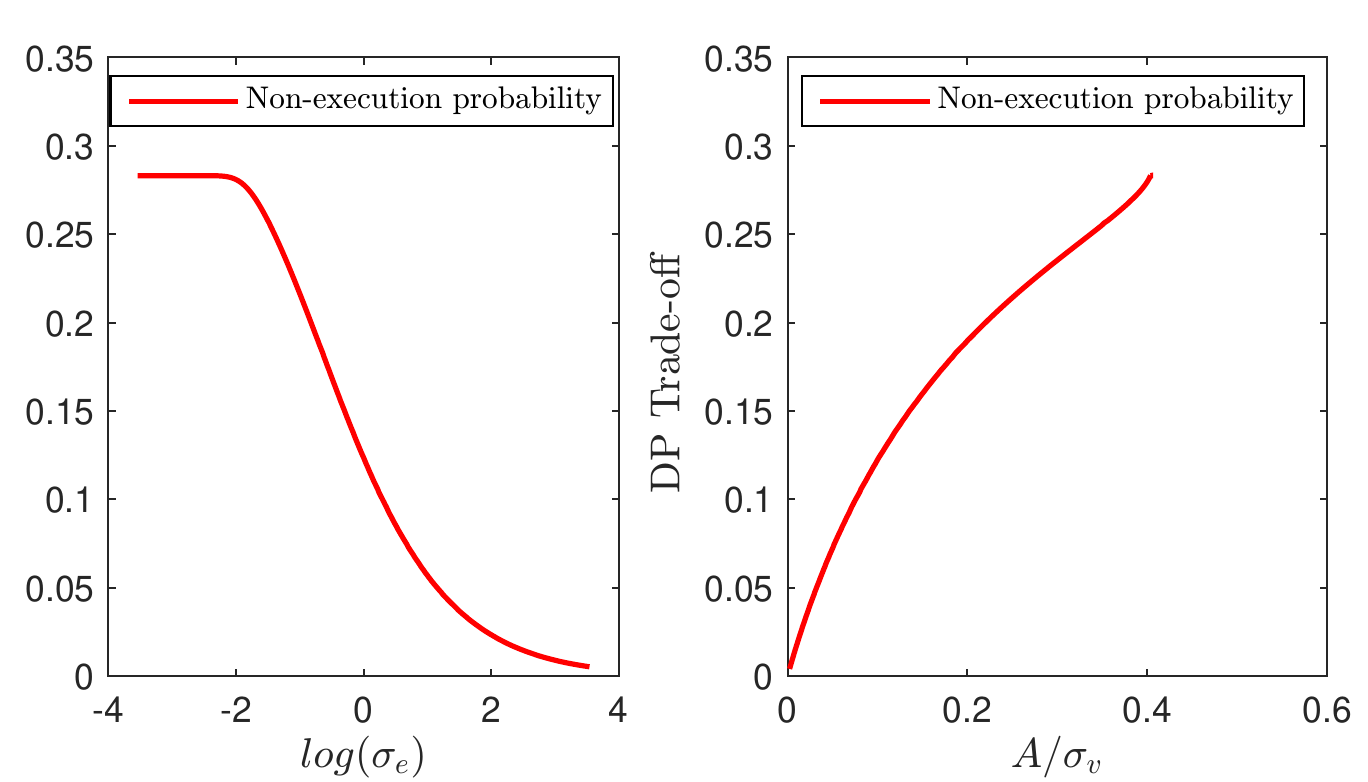}
\caption{\footnotesize \textbf{Execution Probability and Trade-off of A Dark Pool}. The left figure plots the non-execution probability as a function of  $\log(\sigma_e)$.  The right-hand figure plots the non-execution probability  as a function of the exchange spread $A/\sigma_v$. In both plots, $\log(\sigma_v)=0$.}
\label{fig_dp_tradeoff}
\end{figure}

\textbf{2. Dark pool usage and market characteristics.} All else equal, in an economy/industry/asset that has a high information precision, dark pool market share decreases with information precision and with exchange spread, whereas in an economy/industry/asset that has low information precision, dark pool market share increases with information precision and with exchange spread. More precisely,
\begin{itemize}
\item (1) dark pool market share has an inverted U-shape relationship with the information precision,
\item (2) dark pool market share has an inverted U-shape relationship with the exchange spread.
\end{itemize}

The prediction follows from Proposition \ref{prop_monotonicity_lowsigma}, Proposition \ref{prop_monotonicity_lowsigmaparticipation}, and Remark \ref{remark1}. To measure dark pool usage, we analyze the volumes in each venue. Since informed traders have no profit to trade in period 2 due to the disclosure of information, they cancel their unexecuted orders and leave the market in period 2. The remaining orders continue to execute in the exchange. The expected trading volume in the dark pool, in the exchange, and total consolidated volume are, respectively:
\begin{align}
V_d &= (\bar{R} \underline{\gamma_d} + \underline{R} \overline{\gamma_d})\mu + \frac{\bar{R}+\underline{R}}{2}\alpha_d \mu_z,\\
V_e &= ( \underline{\gamma_e} +  \overline{\gamma_e})\mu + \alpha_e \mu_z+ (1-\alpha_e-\alpha_d) \mu_z + \left( 1-\frac{\bar{R}+\underline{R}}{2}\right)\alpha_d \mu_z,\\
V &= V_d + V_e.
\end{align}
We distinguish the components of dark volumes by  ``Dark uninformed volumes'' and ``Dark informed volumes'' respectively as:
\begin{align}
V^{U}_d &= \left( 1-\frac{\bar{R}+\underline{R}}{2}\right)\alpha_d \mu_z,\\
V^{I}_d &= V_d - V^{U}_d.
\end{align}
Figure \ref{fig_volumes} illustrates equilibrium behavior of dark pool market share and dark pool ``informed volume'' share. 
\begin{figure}[h]
\centering
\includegraphics[width=.75\textwidth]{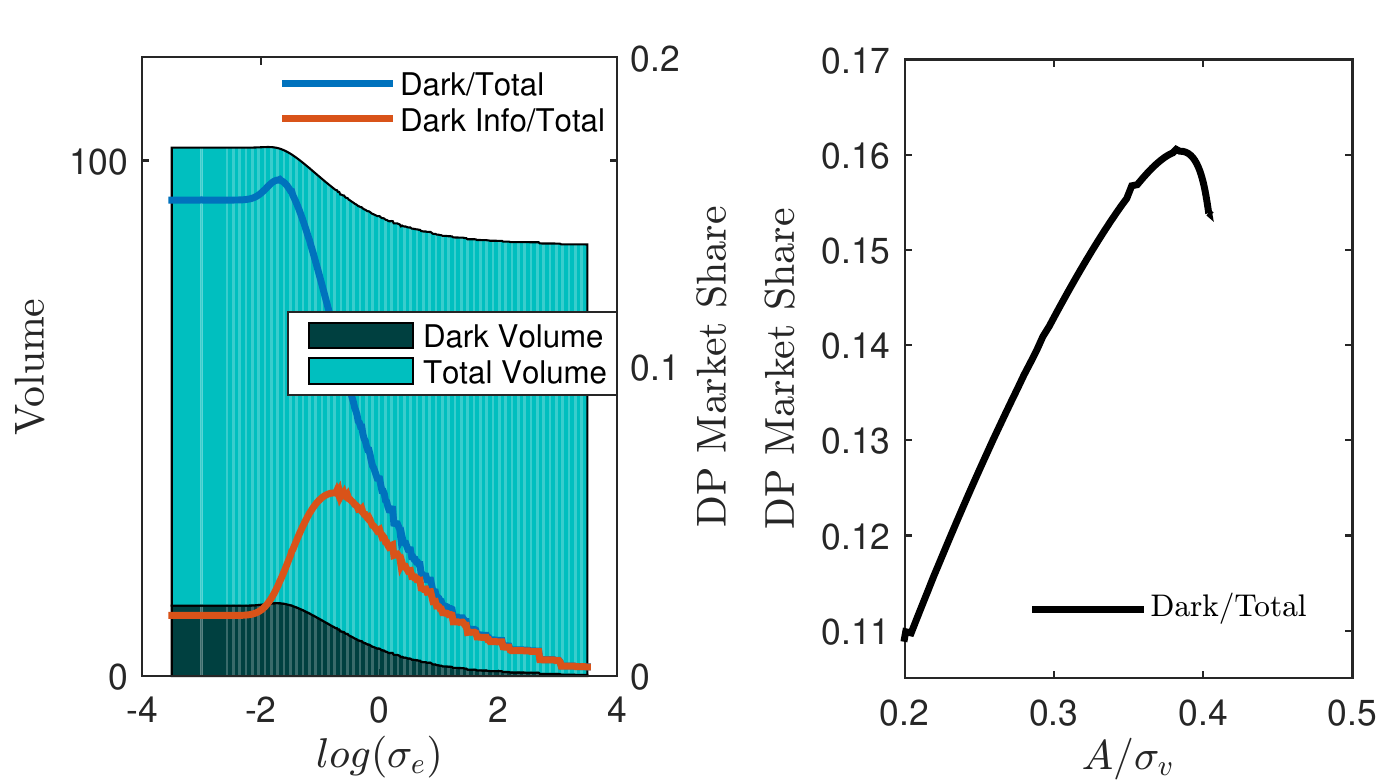}
\caption{\footnotesize \textbf{Dark Volumes and Market Share}. The left figure plots the dark pool volume, total volume and dark pool market share as a function of $\log(\sigma_e)$.  The right-hand figure shows the dark pool market share  as a function of the exchange spread $A/\sigma_v$. In both plots $\log(\sigma_v)=0$.}
\label{fig_volumes}
\end{figure}
Though this prediction coincides with \citet*{zhu_dark_2013}, our model emphasizes the role of the trader's information structure.

This prediction is consistent with \citet*{ray_match_2010} and \citet*{preece_dark_2012}, which report a similar inverted U-shape between dark pool usage and exchange spread. Other empirical studies have reported contradictory results using different datasets. For studies using different US datasets, \citet*{hatheway_empirical_2013} and \citet*{weaver_trade-at_2014} find a positive association while \citet*{ohara_is_2011} and \citet*{ready_determinants_2014-1}  find  a negative association between dark trading and exchange spread. \citet*{asic_dark_2013} and \citet*{comerton-forde_dark_2015-1} study Australian dark trading and find a positive relationship. \citet*{degryse_impact_2011} find a positive relationship for European dark fragmentation. Our model suggests that such a relationship varies cross-sectionally, depending on the specific information structure. The cross-sectional difference is reflected in \citet*{nimalendran_informational_2014},  \citet*{buti_diving_2011}, and \citet*{ohara_is_2011}. More cross-sectional studies that specify the characteristics of firms and countries are needed.

\textbf{3. Information content of dark pool trades.}
In an economy with high information precision, the information content of dark pool trades decreases with information precision and with exchange spread. By contrast, in an economy low information precision, the information content of dark pool trades increases with information precision and with exchange spread. More precisely,
\begin{itemize}
\item (1)  the information content of dark pool trades has an inverted U-shape relationship with the information precision,
\item (2) the information content of dark pools trades has an inverted U-shape relationship with the exchange spread.
\end{itemize}
The prediction follows from Proposition \ref{prop_monotonicity_lowsigma} and Proposition \ref{prop_monotonicity_lowsigmaparticipation}. We use two measures for the dark pool information content.  The first measure is the DP \textit{Predictive Fraction} -- the fraction of dark pool volumes that are traded in the ``right direction'' (i.e., fraction of volumes that predict the movement of prices). The higher the fraction is, the higher is the information content of a dark pool. In this model, the Predictive Fraction is defined as
$$\mbox{DP Predictive Fraction}=\frac{\underline{R} (\overline{\gamma_d}\mu + .5\alpha_d \mu_z)}{V_d}.$$
Another measure we consider  is the normalized adverse selection costs, $\bar{R}-\underline{R}$. The inverted U-shape of the two measures with the exchange spread is depicted in Figure \ref{fig_predicability}.
\begin{figure}[h]
\centering
\includegraphics[width=.75\textwidth]{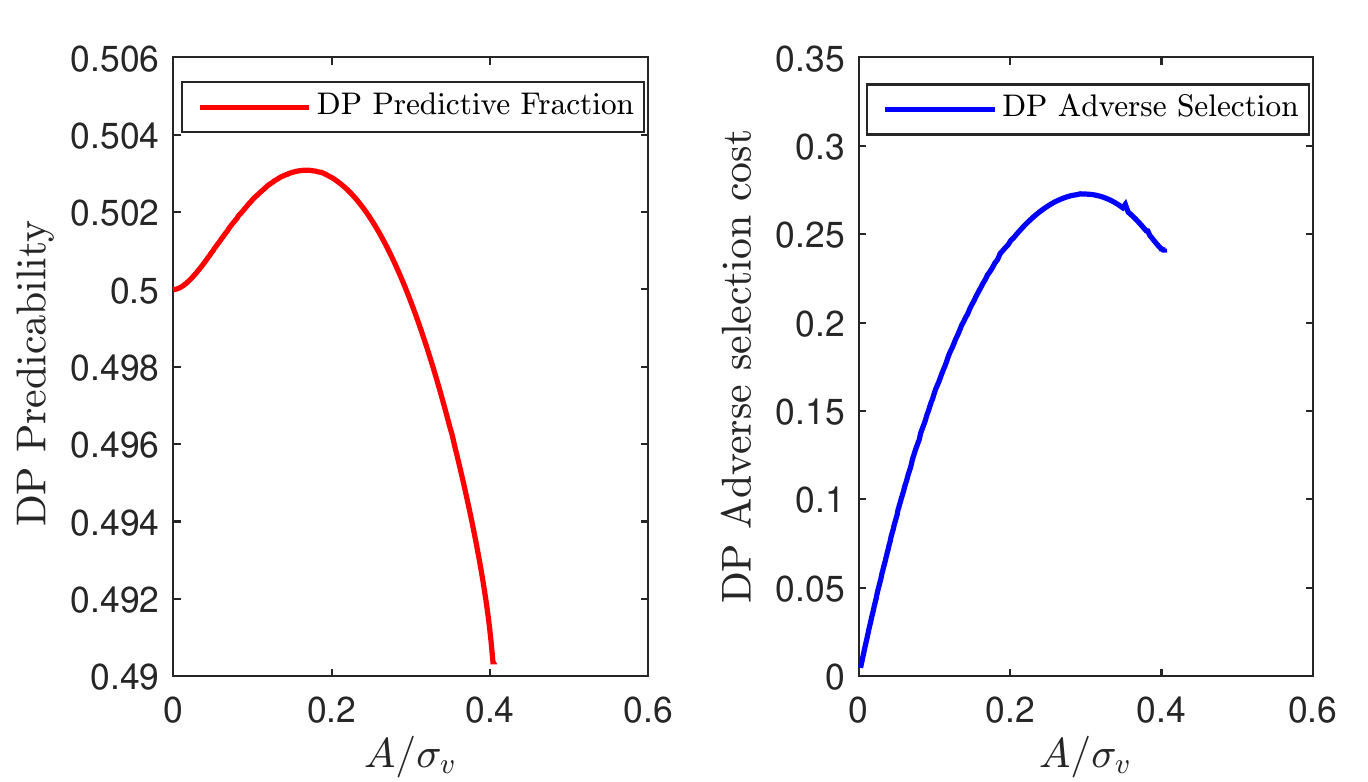}
\caption{\footnotesize \textbf{Predicability of Dark Pool Trades}. The left figure plots the dark pool ``Predictive Fraction'' as a function of spread $A/\sigma_v$.  The right-hand plots the dark pool adverse selection costs as a function of spread $A/\sigma_v$. }
\label{fig_predicability}
\end{figure}
There are relatively few studies that look at this issue. \citet*{peretti_is_2014} conclude that dark pool trades can significantly forecast price movements. \citet*{nimalendran_informational_2014} study trades in a large crossing network and find that the information content in a dark pool is positively associated with the exchange spread.

But as we point out, under different information environments, the dark pool informational content may differ cross-sectionally. Further study in this area is needed.

\textbf{4. Impacts of adding a dark pool alongside an exchange.} We predict that
\begin{itemize}
\item (i) \textbf{Liquidity externality.} Adding a dark pool alongside an exchange decreases the exchange volume but increases the overall volume.
\item (ii) \textbf{Price discovery and exchange spread.} Dark pools have an \textit{amplification effect} on price discovery. That is, the introduction of dark pools  enhances price discovery when price discovery is high, and impairs price discovery when price discovery is low.  Moreover, the improvement of price discovery is associated with a wider exchange spread, whereas the deterioration of price discovery can be associated with a wider or narrower spread.
\end{itemize}
Prediction 4 follows directly from Proposition \ref{prop_compare}, Proposition \ref{prop_informativeness} and Remark \ref{remark2}.\footnote{Remark \ref{remark2} points out, when private information is imprecise, it is possible that price discovery is decreased while spread increases. If this is the case a dark pool can be strictly detrimental to the exchange.} Few studies focus on the direct impact of introducing dark pool trading. For example, \citet*{hendershott_crossing_2000} and \citet*{jones_island_2005} found that there was a reduction in price efficiency after Island ECN stopped displaying its limit order book. \citet*{chlistalla_competition_2010} finds that the entrance of Chi-X, a dark pool in the US, decreased spread.

Other research  studies the relationship between price discovery and dark pool trading intensity within the fragmented framework. O'Hara and Ye (2011) and Jiang et al (2012) find a positive association between price discovery and dark pool trading,  whereas \citet*{hatheway_empirical_2013} and  \citet*{weaver_trade-at_2014} find the opposite. \citet*{comerton-forde_dark_2015-1} conduct a more comprehensive cross-sectional study and show that, when the fraction of non-block trades in dark pools is high (above 10\%, suggesting that dark pools contain a high fraction of informational orders), then dark trading harms price discovery, whereas if dark pools contain less informational orders, dark trading improves price discovery. \citet*{comerton-forde_dark_2015-1}'s prediction is consistent with ours in the sense that we predict an inverted U-shape for the relation of dark pool information content and the information precision.  More research is still needed on the important question of the effect of dark pool activity on price efficiency for different types of stocks in the cross-section.
\subsection{Regulatory Considerations}
\label{sec_policy}
Price discovery is the essential economic function of an exchange. As Alan and Schwartz (2013) point out, price discovery, as a public good, gives investors confidence and promotes the interests of listed entities and the broader community through an efficient secondary market for capital. More precisely, an exchange-produced price benefits a broad spectrum of market participants who use it for marking to market, derivatives valuation, mutual-fund cash flow estimation,  estates, and dark pool pricing.  Thus, the efficiency of how prices are discovered becomes a serious matter in measuring market quality. In the periods of time when markets are deeply fragmented by dark pool trading, it is of extreme importance for regulators to be wary of the impacts dark pools have on price discovery.

\textbf{1. What should regulators do?} Regulators should be cautious in controlling dark pool trading in order to not harm price discovery. To do that, regulators should examine the following aspects. First, dark pool trading should be regulated to a level that distinguishes firm characteristics. As we have pointed out, traders generally possess low precision for high R\&D firms, young firms, small firms, and less-analyzed firms. Introducing dark pools to these firms might cause a decrease in price discovery. Second, a monitoring system measuring the public's ability to process information should be built, and dark pool trading should be under dynamic revision. Third, countries should continue to improve their judicial system to prevent insider trading, and, at the same time, take  measures to improve the efficiency of public disclosure, including accounting information enhancement and financial reporting regulations. Countries should also ensure there are more effective financial media channels. In general, regulators should improve countries' information environment.

\textbf{2, Dark pools in emerging markets?} Based on current evaluations of the information environment in several emerging markets, a great proportion of emerging markets are governed by poor legal systems and have limited implemental power against insider trading and poor quality of information disclosure. These countries should be extremely cautious in dark pool trading.  For example, \citet*{bhattacharya_world_2002} found that the enforcement of insider trading laws in 81 emerging markets is significantly low compared with developed countries. \citet*{wang_quality_2011} and \citet*{yu_literature_2009} document poor quality of financial information in mainland China, and they show that up to a quarter of listed firms in mainland China explicitly admitted to the poor quality of their financial information by restating their previous financial reports. \citet*{tang_insider_2013} finds that a poor corporate governance system interacts with abnormal insider trading to aggravate the information environment in Taiwan. \citet*{budsaratragoon_applying_2012} tests insider trading regulations in Thailand and find that severe informational asymmetry, lax enforcement and poor pricing efficiency are endemic. As we point out, dark pools have an amplification effect on price discovery, so introducing dark pool trading in those countries may aggravate the situation.


\section{Conclusion}
This paper studies the impact of dark pools on price discovery in a noisy information framework. We find that the addition of a dark pool to the traditional exchange has an \textit{amplification effect} on price discovery, i.e., it enhances price discovery when the information has high precision and impairs price discovery when the information has low precision. The results reconcile the conflicting empirical findings in current literature and suggest new channels of research to disentangle the relationship between dark pool trading and market quality.

We highlight the dark pool's function as an informational risk mitigator. In equilibrium, information is sorted by market fragmentation.  That is, traders with strong signals trade in the exchange, traders with modest signals trade in the dark pool, and traders with weak signals do not trade. When information precision is low, a large proportion of informed traders with modest signals  crowd in the dark pool to reduce their information risk. Adding a dark pool, thus, shifts a higher fraction of informed traders from the exchange, compared with liquidity traders, leaving a lower informed-to-uninformed ratio in the exchange and thus decreasing price discovery. In contrast, when information precision is high, a large proportion of informed traders with strong signals crowd in the exchange. Adding a dark pool shifts only a small fraction of informed traders from the exchange, compared with liquidity traders, increasing the informed-to-uninformed ratio in the exchange and increasing price discovery.

There are several observations that complement the overall effects on market quality. First, when information precision is low, the market can experience a deterioration of price discovery along with a widened exchange spread. In this case, dark pools are strictly detrimental to the exchange. Second, dark pools always attract informed traders and liquidity traders in a clustered fashion. We should observe both informed and uninformed traders in all trading venues. Third, the ability of dark trades to predict price movement has an inverted U-shape with exchange spread. Therefore, assets with modest exchange liquidity have a high information content in their dark pool trades.

There are aspects regulators should be aware of. First, dark pools and their impacts have significant variance cross-sectionally. The information structure of different assets, industries, and countries differs in nature. The use of dark pools is thus case sensitive. Second, in a deeply fragmented market, policies that help improve the information environment are needed to enhance price discovery. These measures include, among others, enhancing public disclosure by improving accounting and reporting regulations, strengthening legal systems, and implementing laws against insider trading.

\bibliography{dp_literature}{}
\bibliographystyle{apalike}

\section{Appendix}

\subsection{Proof of Lemma \ref{non-split informed}}

Since each trader is infinitesimal and orders are limited by the amount, his or her action has no impact on the market parameters (i.e., the exchange spread $A$ and the dark pool execution probabilities ($\bar{R},~\underline{R}$)). Therefore, splitting the order cannot affect the (per unit) profit in each venue. Without loss of generality, we focus on the case of a positive signal (the case for a negative signal is similar). Suppose that the informed traders have signal $s>0$. Then he has a belief $B(s)> \frac{1}{2}$. Because the profit of a ``Buy'' order in each venue is strictly higher than the profit of a ``Sell'' order, thus it is optimal to choose the ``Buy'' direction. From his or her perspective, given the exchange spread $A$ and the dark pool execution probabilities ($\bar{R},~\underline{R}$), the expected (per unit) profit for trading in the lit market, dark pool, and not trade depends on his or her confidence level $B(|s|)$ and is determined by (\ref{payoff_informed_lit}), (\ref{payoff_informed_dp}), and (\ref{payoff_informed_no}), respectively.

Because these payoffs are linear in $B(|s|)$, given any belief $B(|s|)$, there is always one venue that is no worse than any of other venues. This relationship is shown in Figure~\ref{fig_payoffs_informed}. When $s\neq  \pm s_0 \mbox{ or }  \pm s_1$, the payoff of trading in one venue is strictly better than others, and it is optimal to send the entire order to that venue. When $s =  \pm s_0 \mbox { or }  \pm s_1$,  there are two venues that yield the same payoff, and the trader can choose to split the order or not between these two venues. However, since the realization of the signal among the informed traders are continuously distributed, the measure of informed traders who receive a particular signal is zero. That is, such traders who are indifferent to these two venues has a mass of zero in the market. Therefore, in probability one, all informed traders send entire order to either the exchange or the dark pool, or not trader at all. \qed

\subsection{Proof of Lemma \ref{non-split uninformed}}
From a type $d$ liquidity trader's perspective, the expected per unit payoff from trading in the lit market, dark pool, and completely deference are determined by (\ref{payoff_uninformed_lit}), (\ref{payoff_uninformed_dp}), and (\ref{payoff_uninformed_no}), respectively. Since each individual has no impact to the market, given $A,~\bar{R},~\underline{R}$, the per unit payoff in each venue is fixed. There is always one venue that is no worse than others. In addition, the payoff is linear in the number of units transacted. Hence there is no need to split among different venues or among different periods. \qed

\subsection{Proof of Theorem \ref{thm existence0}}

Hereafter we normalize some variables via dividing by $\sigma_e$, i.e., let $\bm{\mathrm{s}} = \frac{s}{\sigma_e}$, $\bm{\mathrm{s_0}} = \frac{s_0}{\sigma_e}$, $\bm{\mathrm{s_1}} = \frac{s_1}{\sigma_e}$, $\bm{\widehat{\mathrm{s}}} = \frac{\widehat{s}}{\sigma_e}$, $\bm{\sigma} = \frac{\sigma_v}{\sigma_e}$. Then it is equivalent to prove that, given $\bm{\sigma}\geq 0$, there is a unique cut-off $\mathbf{\widehat{s}}$ such that $h^{\mathbb{S}}_I(\mathbf{s}),~h^{\mathbb{S}}_{U,\iota}(d),~A^{\mathbb{S}},~\overline{\gamma_e}^{\mathbb{S}},~\underline{\gamma_e}^{\mathbb{S}},~\alpha_e^{\mathbb{S}}$ consist a equilibrium, in which
\begin{align}\label{eqn:append_34}
h^{\mathbb{S}}_I(\mathbf{s})&=\left\{\begin{array}{ll}
                                                                         (\mbox{``Buy''},\mbox{ Exchange(Lit)}) &\mbox{ if } \mathbf{s} \geq \widehat{\mathbf{s}},\\
                                                                          (\mbox{``Sell''},\mbox{ Exchange(Lit)}) &\mbox{ if } \mathbf{s} < -\widehat{\mathbf{s}}, \\                                                                                                                                                  \mbox{Not trade}&\mbox{ otherwise, }                                                                    \end{array}\right. \\
h^{\mathbb{S}}_{U,\iota}(d)&=\left\{\begin{array}{ll}
            (\mbox{``Buy'' if $\iota$=Buyer, or ``Sell'' if $\iota$=Seller},\mbox{ Exchange(Lit)}) &\mbox{ if } d \geq 2B(\widehat{\mathbf{s}})-1,\\
            \mbox{Delay trade} &\mbox{ otherwise, }
            \end{array}\right.\\
 \overline{\gamma_e}^{\mathbb{S}}&=1-\Phi(\mathbf{\widehat{s}}-\bm{\sigma}),\\
\underline{\gamma_e}^{\mathbb{S}} &= 1-\Phi(\mathbf{\widehat{s}}+\bm{\sigma}),\\
\alpha_e^{\mathbb{S}} &= 1-G(2B(\mathbf{\widehat{s}})-1),\\\label{eqn:append_39}
\frac{A^{\mathbb{S}}}{\sigma_v} &= \frac{\overline{\gamma_e}^{\mathbb{S}}-\underline{\gamma_e}^{\mathbb{S}}}
{\overline{\gamma_e}^{\mathbb{S}}+\underline{\gamma_e}^{\mathbb{S}}+\alpha_e^{\mathbb{S}}\frac{\mu_z}{\mu}},
\end{align}
where $\mathbf{\mathbf{\widehat{s}}}$ is determined by
\begin{align}
2B(\mathbf{\widehat{s}})-1 = \frac{A^{\mathbb{S}}}{\sigma_v}.
\label{equation3_scaled}
\end{align}

We prove the theorem in two steps. First, we show that if $\mathbf{\widehat{s}}$ is given, the other variables $h^{\mathbb{S}}_I(\mathbf{s}),~h^{\mathbb{S}}_{U,\iota}(d),~A^{\mathbb{S}},~\overline{\gamma_e}^{\mathbb{S}},~\underline{\gamma_e}^{\mathbb{S}},~\alpha_e^{\mathbb{S}}$ solved from~\eqref{eqn:append_34}-\eqref{eqn:append_39} form an equilibrium. Then we show that such $\mathbf{\widehat{s}}$ exists and is unique.

Suppose that $\mathbf{\widehat{s}}$ exists. By (\ref{equation3_scaled}), an informed trader with signal $\mathbf{\widehat{s}}$ is indifferent between trading in the exchange and not trade. Since $B(\mathbf{s})$ is increasing in $\mathbf{s}$, $h^{\mathbb{S}}_I(\mathbf{s})$ is an optimal strategy for informed traders. Similarly, since a type $\widehat{d}=2B(\mathbf{\widehat{s}})-1$ uninformed liquidity trader is indifferent between trading on the exchange and deferring trade, $h^{\mathbb{S}}_{U,\iota}(d)$ is an optimal strategy for uninformed traders. By the law of large numbers, given $h^{\mathbb{S}}_I(\mathbf{s})$ and $h^{\mathbb{S}}_{U,\iota}(d)$, the fraction of uninformed traders who trade in the exchange would be $\alpha_e^{\mathbb{S}} = \Pr(d\geq \widehat{d}) = 1-G(2B(\bm{\widehat{s}})-1)$. Thus, the fraction of informed traders who trade in the ``right direction'' would be $\overline{\gamma_e}^{\mathbb{S}} = \Pr (\mathbf{s}\geq \bm{\widehat{s}})=1-\Phi(\bm{\widehat{s}}-\bm{\sigma})$, and the fraction of informed traders who trade in the ``wrong direction'' would be $\underline{\gamma_e}^{\mathbb{S}} = \Pr (\mathbf{s}< \bm{\widehat{s}})=1-\Phi(\bm{\widehat{s}}+\bm{\sigma})$. In addition, for given $\overline{\gamma_e}^{\mathbb{S}},~\underline{\gamma_e}^{\mathbb{S}},~\alpha_e^{\mathbb{S}}$, we can find $A^{\mathbb{S}}$ from~\eqref{eqn:append_39} and it would make the market maker on the exchange breaks even on average. Thus, $h^{\mathbb{S}}_I(\mathbf{\widehat{s}}),~h^{\mathbb{S}}_{U,\iota}(d),~A^{\mathbb{S}},~\overline{\gamma_e}^{\mathbb{S}},~\underline{\gamma_e}^{\mathbb{S}},~\alpha_e^{\mathbb{S}}$ indeed form an equilibrium.

Then we will prove that such $\mathbf{\widehat{s}}$ exists and is unique. After substituting the expressions of $A^{\mathbb{S}},~\overline{\gamma_e}^{\mathbb{S}},~\underline{\gamma_e}^{\mathbb{S}},~\alpha_e^{\mathbb{S}}$ into~\eqref{equation3_scaled}, we obtain the following equation for $\mathbf{\widehat{s}}$: \begin{align}
&\frac{\Phi(\mathbf{\widehat{s}}+\bm{\sigma})-\Phi(\mathbf{\widehat{s}}-\bm{\sigma})}
{2-\Phi(\mathbf{\widehat{s}}+\bm{\sigma})-\Phi(\mathbf{\widehat{s}}-\bm{\sigma})+(1-G(2B(\mathbf{\widehat{s}})-1))\frac{\mu_z}{\mu}}=
2B(\mathbf{\widehat{s}})-1.
\label{equation_hats}
\end{align}
Define
\begin{align*}
f(s) = &(2B(s)-1)\left[2-\Phi(s+\bm{\sigma})-\Phi(s-\bm{\sigma})+(1-G(2B(s)-1))\frac{\mu_z}{\mu}\right] \\
&-\left[\Phi(s+\bm{\sigma})-\Phi(s-\bm{\sigma})\right],
\end{align*}
 and the derivative of $f(s)$ is
\begin{align*}f'(s) = &2B'(s)\left[2-\Phi(s+\bm{\sigma})
-\Phi(s-\bm{\sigma})+(1-G(2B(s)-1))\frac{\mu_z}{\mu}\right] \\
&-2(2B(s)-1)G'(2B(s)-1) B'(s).\end{align*}
We can easily find that $f(\frac{1}{2})<0,~f(+\infty)>0$, $f'(0)>0,~f'(+\infty)=0$. Because $G'(x)+xG''(x)\geq 0,~\forall x \in [0,1]$, we have $f''(s)<0$.
Thus there exists a unique $\mathbf{\widehat{s}}$ such that $f(\mathbf{\widehat{s}})=0$. \qed

\subsection{Proof of Theorem \ref{thm existence}} \label{App:existence}
Hereafter we normalize some variables via dividing by $\sigma_e$, i.e., $\bm{\mathrm{s}} = \frac{s}{\sigma_e}$, $\bm{\mathrm{s_0}} = \frac{s_0}{\sigma_e}$,$\bm{\mathrm{s_1}} = \frac{s_1}{\sigma_e}$, $\bm{\widehat{\mathrm{s}}} = \frac{\widehat{s}}{\sigma_e}$, $\bm{\sigma} = \frac{\sigma_v}{\sigma_e}$.
Then finding the equilibrium is equivalent to solving the following system of equations:
\begin{align}
&B(\mathbf{s}_0)(\bar{R}+\underline{R})=\bar{R},
\label{equation0_scaled}\\
&B(\mathbf{s}_1)\left[(1-\bar{R})+(1-\underline{R})\right]=\frac{A}{\sigma_v}+(1-\bar{R}),
\label{equation1_scaled}\\
\bar{R}&=\mathbb{E}\left[\min\left\{1, \frac{\overline{\gamma_d}\mu + \alpha_d Z^+}{\underline{\gamma_d}\mu + \alpha_d Z^-} \right\} \right],
\label{R_upper_scaled}\\
\underline{R}&= \mathbb{E}\left[\min\left\{1, \frac{\underline{\gamma_d}\mu + \alpha_d Z^-}{\overline{\gamma_d}\mu + \alpha_d Z^+} \right\}\right],
\label{R_lower_scaled}\\
\frac{A}{\sigma_v}&=\frac{\overline{\gamma_e}-\underline{\gamma_e}}{(\overline{\gamma_e}+\underline{\gamma_e})+\alpha_e {\mu_z \over \mu}},
\label{spread_scaled}\\
\overline{\gamma_e}&=1-\Phi(\mathbf{s_1}-\bm{\sigma}),
\label{gamma_upper_scaled}\\
\underline{\gamma_e}&=1-\Phi(\mathbf{s_1}+\bm{\sigma}),
\label{gamma_lower_scaled}\\
\overline{\gamma_d}&=\Phi(\mathbf{s_1}-\bm{\sigma})-\Phi(\mathbf{s_0}-\bm{\sigma}),
\label{beta_upper_scaled}\\
\underline{\gamma_d}&=\Phi(\mathbf{s_1}+\bm{\sigma})-\Phi(\mathbf{s_0}+\bm{\sigma}),
\label{beta_lower_scaled}\\
\alpha_e &=1-G(2B(\mathbf{s_1})-1),
\label{alpha_e_scaled}\\
\alpha_d &= G(2B(\mathbf{s_1})-1)-G(2B(\mathbf{s_0})-1),
\label{alpha_d_scaled}
\end{align}
where
\begin{align}
B(\mathbf{s})=\frac{\phi(\mathbf{s}-\bm{\sigma})}{\phi(\mathbf{s}-\bm{\sigma})+\phi(\mathbf{s}+\bm{\sigma})}.
\end{align}

Before proving the existence of solutions to the system of equations, we introduce the following lemma.
\begin{lem}
Let $\mathbf{s_0} \geq 0$ and $\mathbf{s_1}=\mathbf{s_0}+\epsilon$, we have
\begin{align*}
\lim_{\epsilon \rightarrow 0^+}\underline{R}&=\mathrm{E \left[min\left\{1, \frac{\phi(\mathbf{s_0}+\bm{\sigma})\mu + 2G'(2B(\mathbf{s_0})-1)B'(\mathbf{s_0})Z^-}{\phi(\mathbf{s_0}-\bm{\sigma})\mu + 2G'(2B(\mathbf{s_0})-1)B'(\mathbf{s_0})Z^+} \right\} \right]},\\
\lim_{\epsilon \rightarrow 0^+}\bar{R}&=\mathrm{E \left[min\left\{1, \frac{\phi(\mathbf{s_0}-\bm{\sigma})\mu + 2G'(2B(\mathbf{s_0})-1)B'(\mathbf{s_0})Z^+}{\phi(\mathbf{s_0}+\bm{\sigma})\mu + 2G'(2B(\mathbf{s_0})-1)B'(\mathbf{s_0})Z^-} \right\} \right]},\\
\lim_{\epsilon\rightarrow 0^+} \frac{A}{\sigma_v} &= \frac{\Phi(\mathbf{s_0}+\bm{\sigma})-\Phi(\mathbf{s_0}-\bm{\sigma})}{2-\Phi(\mathbf{s_0}+\bm{\sigma})-\Phi(\mathbf{s_0}-\bm{\sigma})+\left[1-G(2B(\mathbf{s_0})-1)\right]\frac{\mu_z}{\mu}}.
\end{align*}
Moreover, if $\mathbf{s_0}=0$ or $\bm{\sigma}=0$, then $\lim\limits_{\epsilon \rightarrow 0^+}\underline{R}=\lim\limits_{\epsilon \rightarrow 0^+}\bar{R}=1$. Therefore, we define $\underline{R}$, $\bar{R}$, and $\frac{A}{\sigma_v}$ use these limits when $\mathbf{s_0}=\mathbf{s_1}$. \label{lem_Rlimit}
\end{lem}

\begin{proof}
We can prove this by the Taylor expansion. Suppose that $\epsilon$ is sufficiently small. Because $\mathbf{s_0} \geq 0$ and $\mathbf{s_1}=\mathbf{s_0}+\epsilon$, we have, by the Taylor expansion, that $\underline{\gamma_d}= \phi(\mathbf{s_0}+\bm{\sigma}) \epsilon + o(\epsilon) $,  $\overline{\gamma_d}=\phi(\mathbf{s_0}-\bm{\sigma}) \epsilon + o(\epsilon)$, and $\alpha_d= 2G'(2B(\mathbf{s_0})-1)B'(\mathbf{s_0}) \epsilon+ o(\epsilon) $.  Therefore we have
\begin{align*}
\underline{R}&=\mathrm{E \left[min \left\{1, \frac{\phi(\mathbf{s_0}+\bm{\sigma})\mu \epsilon+ 2G'(2B(\mathbf{s_0})-1)B'(\mathbf{s_0})Z^- \epsilon }{\phi(\mathbf{s_0}-\bm{\sigma})\mu \epsilon+ 2G'(2B(\mathbf{s_0})-1)B'(\mathbf{s_0})Z^+ \epsilon } +  o(\epsilon)\right\}  \right]},\\
\bar{R}&=\mathrm{E \left[min\left\{1, \frac{\phi(\mathbf{s_0}-\bm{\sigma})\mu \epsilon+ 2G'(2B(\mathbf{s_0})-1)B'(\mathbf{s_0})Z^+ \epsilon}{\phi(\mathbf{s_0}+\bm{\sigma})\mu \epsilon+ 2G'(2B(\mathbf{s_0})-1)B'(\mathbf{s_0})Z^-\epsilon}  + o(\epsilon) \right\} \right]}.
\end{align*}
Similarly, by the Taylor expansion, we have  $\overline{\gamma_e}=1-\Phi(\mathbf{s_0}-\bm{\sigma})-\phi(\mathbf{s_0}-\bm{\sigma}) \epsilon + o(\epsilon)$, $\underline{\gamma_e}=1-\Phi(\mathbf{s_0}+\bm{\sigma})-\phi(\mathbf{s_0}+\bm{\sigma}) \epsilon + o(\epsilon)$, and $\alpha_e= \left[1-G(2B(\mathbf{s_0})-1)\right] - 2 G'(\cdot)B'(\mathbf{s_0}) \epsilon+ o(\epsilon) $. Therefore
\begin{align*}\small
&\frac{A}{\sigma_v} = \\
&\small\frac{\Phi(\mathbf{s_0}+\bm{\sigma})-\Phi(\mathbf{s_0}-\bm{\sigma})-\left[\phi(\mathbf{s_0}+\bm{\sigma})-\phi(\mathbf{s_0}-\bm{\sigma})\right]\epsilon }{2-\Phi(\mathbf{s_0}+\bm{\sigma})-\Phi(\mathbf{s_0}-\bm{\sigma})+\left[1-G(2B(\mathbf{s_0})-1)\right]\frac{\mu_z}{\mu}
-\left[\phi(\mathbf{s_0}+\bm{\sigma})+\phi(\mathbf{s_0}-\bm{\sigma})+2 G'(\cdot)B'(\mathbf{s_0})\right]\epsilon}\\
&+ o(\epsilon).
\end{align*}
Let $\epsilon \rightarrow 0^+$, and we prove the lemma.
\end{proof}

We prove the theorem in a similar way as in the proof of Theorem~\ref{thm existence0}. First, we show that if $\mathbf{s_0}$ and $\mathbf{s_1}$ are given, the other variables $h_I(\cdot)$, $h_{U,\iota}(\cdot)$, $A,~\bar{R},~\underline{R},~\overline{\gamma_e},~\underline{\gamma_e},~\overline{\gamma_d},~\underline{\gamma_d},~\alpha_d,~\alpha_e$ solved from~\eqref{R_upper_scaled}-\eqref{alpha_d_scaled} form an equilibrium. Then we show that ($\mathbf{s_0},~\mathbf{s_1}$) exists and is unique.

Given $A,~\bar{R},~\underline{R},~\overline{\gamma_e},~\underline{\gamma_e},~\overline{\gamma_d},~\underline{\gamma_d},~\alpha_d,~\alpha_e$ and that $\mathbf{s_0}, \mathbf{s_1}$ determined by (\ref{equation0_scaled}), (\ref{equation1_scaled}), $0< \mathbf{s_0} < \mathbf{s_1} $, we show that it is optimal for informed speculators and uninformed liquidity buyers (and sellers) to following the strategy described respectively by $h_I(\cdot)$ and  $h_{U,\iota}(\cdot), \iota \in \left\{Buyer, Seller\right\}$.

Consider an informed speculator who receives a signal $\bm{\mathrm{s}}\geq 0$ (the case when $\bm{\mathrm{s}}\leq 0$ is symmetric with respect to the vertical axis, and hence the analysis is similar and skipped here). Suppose that $0< \mathbf{s_0} < \mathbf{s_1}$. From his or her perspective, the expected payoffs in the lit market, the dark pool, and no-trade are, respectively, $\left[B(\bm{\mathrm{s}})\sigma_v-(1-B(\bm{\mathrm{s}}))\sigma_v\right] -A$, $B(\mathbf{s})\underline{R}\sigma_v-(1-B(\bm{\mathrm{s}}))\bar{R}\sigma_v$, and 0. Figure \ref{fig_payoffs_informed} captures the payoff as a function of $B(\bm{\mathrm{s}})$. As one can see in the graph, since the payoffs are linear with respect to $B(\bm{\mathrm{s}})$, and $B(\bm{\mathrm{s}})$ is strictly increasing with respect to $\bm{\mathrm{s}}$, the optimal strategy for an informed speculator with signal $\bm{\mathrm{s}}$ should use the exchange (the lit market) to trade when his or her signal $\bm{\mathrm{s}}\geq \mathbf{s_1}$, and the dark pool when $\mathbf{s_0} \leq \bm{\mathrm{s}} < \mathbf{s_1}$, and stay outside when $\bm{\mathrm{s}}< \mathbf{s_0}$. This is marked as the red line in Figure \ref{fig_payoffs_informed}.

The fractions of each type of traders in each venue $\overline{\gamma_e},~\underline{\gamma_e},~\overline{\gamma_d},~\underline{\gamma_d},~\alpha_e,~\alpha_d$ are determined by  (\ref{gamma_upper_scaled}), (\ref{gamma_lower_scaled}), (\ref{beta_upper_scaled}), (\ref{beta_lower_scaled}), (\ref{alpha_e_scaled}), (\ref{alpha_d_scaled}), respectively, and $A,~\bar{R},~\underline{R}$ are given by (\ref{spread_scaled}), (\ref{R_upper_scaled}), (\ref{R_lower_scaled}). Thus properties (ii), (iii) and (iv) in Definition~\ref{eqn def} are satisfied.

Then we need to show that such pair of cut-off ($\mathbf{s_0},~\mathbf{s_1}$) exists and satisfies $0< \mathbf{s_0} < \mathbf{s_1}$. In order to show this, we consider equations (\ref{equation0_scaled}) and (\ref{equation1_scaled}) and show that there is a intersection for the two lines represented by these two equations.
\begin{figure}[h]
\centering
\includegraphics[width=0.7\textwidth]{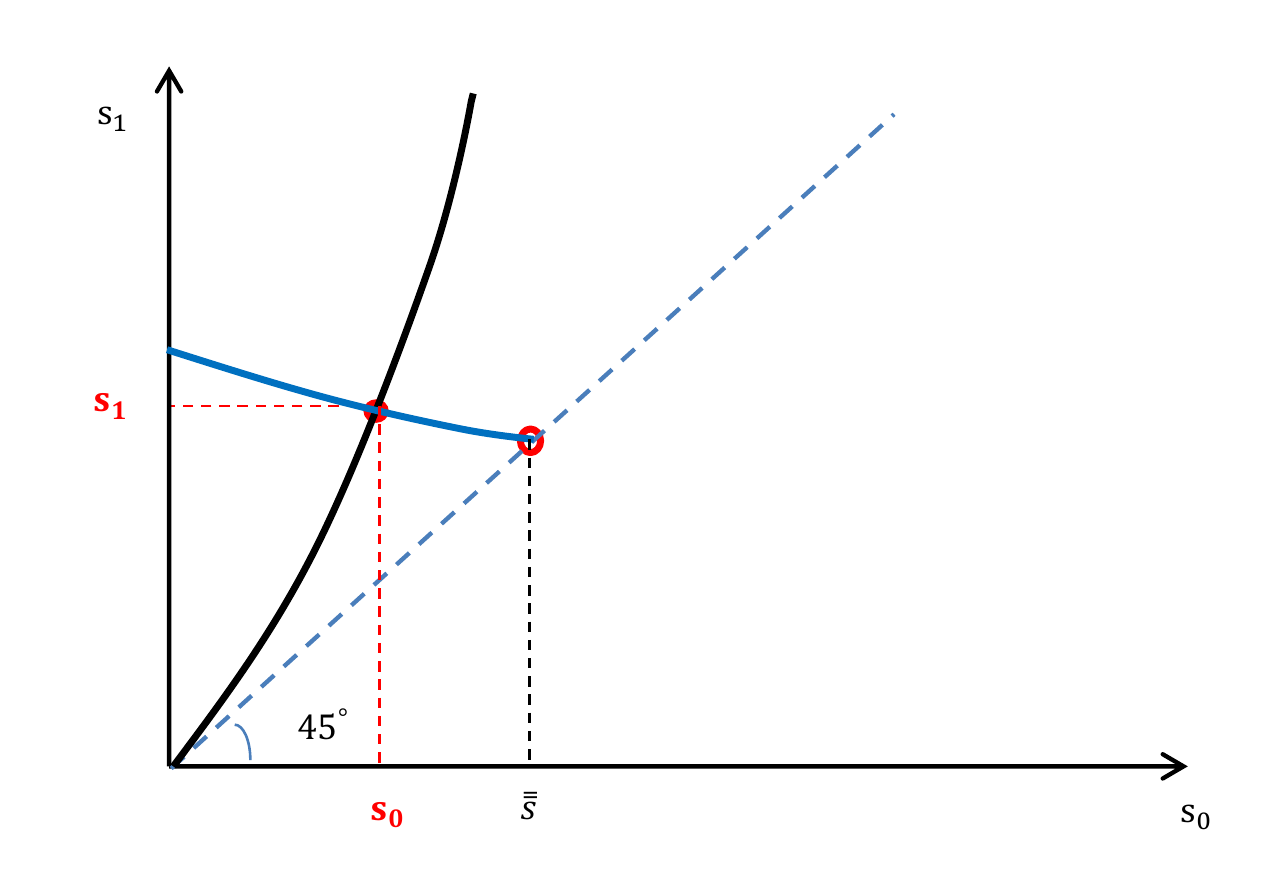}
\caption{Equilibrium Existence }
\label{fig_existence}
\end{figure}

For equation~\eqref{equation0_scaled}, we show that $(\mathbf{s_0},~\mathbf{s_1}) = (0,0)$ satisfies equation (\ref{equation0_scaled}) and behaves as the black line in Figure~\ref{fig_existence}.
\begin{itemize}
\item (i) Suppose $\mathbf{s_0}=0,~\mathbf{s_1}=0$, then $B(\mathbf{s_0})=\frac{1}{2}$, and by Lemma~\ref{lem_Rlimit},  $\Bar{R}=\underline{R}=1$. Therefore equation (\ref{equation0_scaled}) is satisfied.
\item (ii) Now suppose that $\mathbf{s_0}>0$, then $\frac{1}{2}<B(\mathbf{s_0})<1$. To satisfy (\ref{equation0_scaled}),  we need that $\underline{R}<\bar{R}\leq 1$, thus $|\underline{\gamma_d}|<|\overline{\gamma_d}|$. To obtain this, it must be true that $\mathbf{s_1}>\mathbf{s_0}$ if such $\mathbf{s_1}$ exists. By continuity such $\mathbf{s_1}$ must exist for a small enough $\mathbf{s_0}$. (Note that if $\mathbf{s_0}$ is too large, such $\mathbf{s_1}$ may not exist.)
\item (iii) We also show that there exist some $\overline{\mathbf{s}}$ such that $\mathbf{s_1} \rightarrow +\infty$ when $\mathbf{s_0} \rightarrow \overline{\mathbf{s}}$. We rewrite equation (\ref{equation0_scaled}) as $B(\mathbf{s_0}) = \frac{\bar{R}}{\bar{R}+\underline{R}}$. As $\mathbf{s_1} \rightarrow +\infty$, $\overline{\gamma_d} \rightarrow 1- \Phi(\mathbf{s_0}-\bm{\sigma})$, $\underline{\gamma_d} \rightarrow 1- \Phi(\mathbf{s_0}+\bm{\sigma})$, $\alpha_d \rightarrow 1-B(\mathbf{s_0})$. Hence $\bar{R} \rightarrow \mathbb{E} \left[\min \left\{ 1,  \frac{1- \Phi(\mathbf{s_0}-\bm{\sigma})+ [1-B(\mathbf{s_0})] Z^+} {1- \Phi(\mathbf{s_0}+\bm{\sigma})+ [1-B(\mathbf{s_0})] Z^-} \right\}\right] $, and $\underline{R} \rightarrow \mathbb{E} \left[\min \left\{ 1,  \frac{1- \Phi(\mathbf{s_0}+\bm{\sigma})+ [1-B(\mathbf{s_0})] Z^-}{1- \Phi(\mathbf{s_0}-\bm{\sigma})+ [1-B(\mathbf{s_0})] Z^+} \right\}\right] $. Therefore, for any $\mathbf{s_0} \in [0, \infty)$, there must exist $\overline{\gamma_d}>\underline{\gamma_d}$, thus $\bar{R}>\underline{R}$. Then let $\mathbf{s_1} \rightarrow +\infty$, the left hand side of the equation, $B(\mathbf{s_0})$, is equal to $\frac{1}{2}$ if $\mathbf{s_0} = 0$, and is equal to $1$ if  $\mathbf{s_0} \rightarrow +\infty$. However, the right hand side of the equation, $\frac{\bar{R}}{\bar{R}+\underline{R}}$, is greater than $\frac{1}{2}$ if $\mathbf{s_0} =0$, and equal to $\frac{1}{2}$ if $\mathbf{s_0} \rightarrow +\infty$. This is because $\underset{s\rightarrow +\infty}{\lim} \frac{1- \Phi(\mathbf{s_0}-\bm{\sigma})}{1-B(\mathbf{s_0})}=\underset{s\rightarrow +\infty}{\lim} \frac{1- \Phi(\mathbf{s_0}+\bm{\sigma})}{1-B(\mathbf{s_0})}=0 $, so $\underset{s\rightarrow +\infty}{\lim} \bar{R} = \underset{s\rightarrow +\infty}{\lim} \underline{R} = \mathbb{E} \left[ \min \left\{1, \frac{Z^+}{Z^-} \right\} \right]$. By continuity, there must exist an $\overline{\mathbf{s}} \in (0, +\infty)$ such that, as $\mathbf{s_0} \rightarrow \overline{\mathbf{s}},~\mathbf{s_1} \rightarrow +\infty$, LHS = RHS. That is, equation (\ref{equation0_scaled}) is satisfied.
\end{itemize}

For equation (\ref{equation1_scaled}). We rewrite it as
\begin{align}\label{equation1_scaled_b}
B(\mathbf{s_1})=\frac{\frac{A}{\sigma_v}}{(1-\bar{R})+(1-\underline{R})}+\frac{(1-\bar{R})}{(1-\bar{R})+(1-\underline{R})}.
\end{align}
\begin{itemize}
\item (i) Suppose that $\mathbf{s_0}=0$, we prove that there must exist a $\mathbf{s_1}>0$ satisfy (\ref{equation1_scaled}). Note that for any given $\bm{\sigma} \in (0, +\infty)$, $A>0$ is satisfied. If $\mathbf{s_1}=0$, we have $B(\mathbf{s_1})=\frac{1}{2}$ and $\bar{R}=\underline{R}=1$. Plugging into (\ref{equation1_scaled}) gives us $A=0$, which contradicts the fact that $A>0$. If $\mathbf{s_1}<0$, then $B(\mathbf{s_1})<\frac{1}{2}$, $\overline{\gamma_d}<\underline{\gamma_d}$, and $0<\bar{R}<\underline{R}<1$ (we don't consider any $\bar{R},\underline{R}<0$). Hence $\frac{(1-\bar{R})}{(1-\bar{R})+(1-\underline{R})}>\frac{1}{2}>B(\mathbf{s_1})$. In order for~\eqref{equation1_scaled} to be satisfied, we have  $A<0$, which contradicts with that fact that $A>0$. Then we show the existence of $\mathbf{s_1}$ using the continuity of equation~\eqref{equation1_scaled_b}. Its left hand side $B(\mathbf{s_1})$ is increasing in $\mathbf{s_1}$ and $B(0)=\frac{1}{2},~\underset{\mathbf{s_1}\rightarrow \infty}{\lim} B(\mathbf{s_1})=1$. If $\mathbf{s_1}=0$, the right hand side equals $\frac{\frac{A}{\sigma_v}}{(1-\bar{R})+(1-\underline{R})}+\frac{1}{2}>\frac{1}{2}$. However, when $\mathbf{s_1}\rightarrow \infty$, we have $A\rightarrow 0$ and $1>\bar{R}>\underline{R}$, hence the right hand side equals $0+\frac{1-\bar{R}}{(1-\bar{R})+(1-\underline{R})}<\frac{1}{2}$. By the continuity of equation~\eqref{equation1_scaled_b}, there must exist a $\mathbf{s_1}\in (0,+\infty)$ such that the equation is satisfied.
\item (ii)
Next we prove that there exist an $\overline{\overline{\mathbf{s}}}>0$ and small enough $\epsilon>0$ such that for $\mathbf{s_0} = \overline{\overline{\mathbf{s}}},~\mathbf{s_1}=\overline{\overline{\mathbf{s}}}+\epsilon$, equation (\ref{equation1_scaled}) is satisfied as $\epsilon \rightarrow 0^+$. Consider any $\mathbf{s_0}=\mathbf{s}, \mathbf{s_1} = \mathbf{s} + \epsilon$, when $\epsilon>0$ is sufficiently small. By Lemma \ref{lem_Rlimit}, equation (\ref{equation1_scaled}) is equivalent to
\begin{align}\label{eqn:app_55}
B(\mathbf{s})=\frac{\frac{A}{\sigma_v}}{(1-\bar{R})+(1-\underline{R})}+\frac{(1-\bar{R})}{(1-\bar{R})+(1-\underline{R})},
\end{align}
where $\frac{A}{\sigma_v}= \frac{\Phi(\mathbf{s}+\bm{\sigma})-\Phi(\mathbf{s}-\bm{\sigma})}{2-\Phi(\mathbf{s}
+\bm{\sigma})-\Phi(\mathbf{s}-\bm{\sigma})+\left[1-G(2B(\mathbf{s})-1)\right]\frac{\mu_z}{\mu}}$,
$\underline{R}=\mathrm{E \left[min\left\{1, \frac{\phi(\mathbf{s}-\bm{\sigma})\mu + 2G'(2B(\mathbf{s})-1)B'(\mathbf{s})Z^+}{\phi(\mathbf{s}+\bm{\sigma})\mu + 2G'(2B(\mathbf{s})-1)B'(\mathbf{s})Z^-} \right\} \right]}$, and
$\bar{R} =\mathrm{E \left[min\left\{1, \frac{\phi(\mathbf{s}+\bm{\sigma})\mu + 2G'(2B(\mathbf{s})-1)B'(\mathbf{s})Z^-}{\phi(\mathbf{s}-\bm{\sigma})\mu + 2G'(2B(\mathbf{s})-1)B'(\mathbf{s})Z^+} \right\} \right]}$.
Consider $\mathbf{s}$ on $[0,\infty)$. The left hand side of equation~\eqref{eqn:app_55} increases with respect to $\mathbf{s}$. We have $B(0)=\frac{1}{2}$, and $\underset{\mathbf{s}\rightarrow \infty}{\lim} B(\mathbf{s})=1$.  Now consider the right hand side of equation~\eqref{eqn:app_55}. By Lemma~\ref{lem_Rlimit}, we know that if $\mathbf{s}\rightarrow 0^+$, the limit of the right hand side is $\frac{\frac{A}{\sigma_v}}{(1-\bar{R})+(1-\underline{R})}+\frac{1}{2}>\frac{1}{2}$. If $\mathbf{s}\rightarrow \infty$, we have $A\rightarrow 0$ and $1>\bar{R}>\underline{R}$, hence the limit of the right hand side is $0+\frac{1-\bar{R}}{(1-\bar{R})+(1-\underline{R})}<\frac{1}{2}$. By continuity there must exist a $\overline{\overline{\mathbf{s}}}\in (0,\infty)$ such that equation~\eqref{eqn:app_55} is satisfied at $(\overline{\overline{\mathbf{s}}},~\overline{\overline{\mathbf{s}}})$ (i.e., $\mathbf{s_0}=\mathbf{s_1}=\overline{\overline{\mathbf{s}}}$). 
\end{itemize}

The above argument can be summarized by Figure \ref{fig_existence}.
Given $\bm{\sigma}>0$ fixed, the black curve represents the $(\mathbf{s_0}$, $\mathbf{s_1})$ pairs that satisfy equation (\ref{equation0_scaled}). It goes through the point $(0,0)$, is always above the line $\mathbf{s_1}=\mathbf{s_0}$, and $\mathbf{s_1}\rightarrow+\infty$ when $\mathbf{s_0}\rightarrow \overline{\mathbf{s}}$.
The red curve represents the ($\mathbf{s_0}$, $\mathbf{s_1}$) pairs that satisfy equation (\ref{equation1_scaled}). When $\mathbf{s_0}=0$, $\mathbf{s_1} \in (0,\infty)$. And there exists some $\overline{\overline{\mathbf{s}}} \in (0,+\infty)$ such that $\mathbf{s_0}=\mathbf{s_1}=\overline{\overline{\mathbf{s}}}$, satisfies equation (\ref{equation1_scaled}). Then because all functions are continuous, there must exist a pair $(\mathbf{s_0},~\mathbf{s_1}),~0<\mathbf{s_0}<\mathbf{s_1}<+\infty$, such that both equations (\ref{equation0_scaled}) and (\ref{equation1_scaled}) are satisfied. It is the intersection of the black curve and the red curve in Figure \ref{fig_existence}. The existence is then established.\qed

\subsection{Proof of Proposition \ref{prop_monotonicity_lowsigma} and Proposition \ref{prop_monotonicity_lowsigmaparticipation}}

To prove Propositions \ref{prop_monotonicity_lowsigma} and \ref{prop_monotonicity_lowsigmaparticipation}, we need the following two lemmas.

\begin{lem}
Suppose $s(\bm{\sigma})$ is  continuously differentiable over $(0, +\infty)$, and $\underset{\bm{\sigma} \rightarrow 0^+}{\lim} s(\bm\sigma)\bm{\sigma} =0$, then
\begin{align*}
\underset{\bm{\sigma} \rightarrow 0^+}{\lim} \left(\phi(s(\bm\sigma)+\bm{\sigma}) - \phi(s(\bm\sigma)-\bm{\sigma})\right) s'(\bm{\sigma}) &=0\\
\underset{\bm{\sigma} \rightarrow 0^+}{\lim} \left(\Phi(s(\bm\sigma)+\bm{\sigma}) - \Phi(s(\bm\sigma)-\bm{\sigma})\right) s'(\bm{\sigma}) &=0
\end{align*}
In addition,
\begin{itemize}
\item (i) If $\underset{\bm{\sigma} \rightarrow 0^+}{\lim} s(\bm\sigma) =\pm\infty$, $\left|\bm{\sigma}s'(\bm{\sigma}) \right|\leq s(\bm\sigma)$ for sufficiently small $\bm\sigma$.
\item (ii) If $-\infty <\underset{\bm{\sigma} \rightarrow 0^+}{\lim} s(\bm\sigma) <+\infty$, $\underset{\bm{\sigma} \rightarrow 0^+}{\lim} \bm{\sigma} s'(\bm{\sigma})=0$.
\end{itemize}
\label{lem_limit_zero}
\end{lem}

\begin{proof}
(i) Suppose that $\underset{\bm{\sigma} \rightarrow 0^+}{\lim} s(\bm\sigma)=+\infty$. There exists $\epsilon>0$ such that $\forall \bm{\sigma} \in (0, \epsilon),~s(\bm\sigma)\bm{\sigma} >0$ and $\frac{d (s(\bm\sigma)\bm\sigma)}{d \bm\sigma}>0$. Thus $\frac{d (s(\bm\sigma)\bm\sigma)}{d \bm\sigma}=\bm{\sigma}s'(\bm{\sigma}) + s(\bm\sigma) \geq 0$, and
$$\left|\bm{\sigma}s'(\bm{\sigma}) \right|\leq |s(\bm\sigma)|,$$
for $\bm\sigma\in(0,\epsilon)$.
Similarly, if $\underset{\bm{\sigma} \rightarrow 0^+}{\lim} s(\bm\sigma)=-\infty$, we have that
$$\left|\bm{\sigma}s'(\bm{\sigma}) \right|\leq |s(\bm\sigma)|,$$
for sufficiently small $\bm\sigma$.

Therefore, by mean value theorem, we have
\begin{align*}
\underset{\bm{\sigma} \rightarrow 0^+}{\lim} \left(\phi(s(\bm\sigma)+\bm{\sigma}) - \phi(s(\bm\sigma)-\bm{\sigma})\right) s'(\bm{\sigma})&= \underset{\bm{\sigma} \rightarrow 0^+}{\lim}  \int^{s(\bm\sigma)+\bm{\sigma}}_{s(\bm\sigma)-\bm{\sigma}}-xe^{-\frac{x^2}{2}}dx s'( \bm{\sigma})\\
& = \underset{\bm{\sigma} \rightarrow 0^+}{\lim} -2\bm{\sigma} s(\bm\sigma)e^{-\frac{s(\bm\sigma)^2}{2}}s'(\bm{\sigma}).
 \end{align*}
Because $\left|\bm{\sigma}s'(\bm{\sigma}) \right|\leq |s(\bm\sigma)|$ and $\underset{\bm{\sigma} \rightarrow 0^+}{\lim} \left|-2 s(\bm\sigma)^2e^{-\frac{s(\bm\sigma)^2}{2}}\right|=0$, we obtain
\begin{align*}
\underset{\bm{\sigma} \rightarrow 0^+}{\lim} \left(\phi(s(\bm\sigma)+\bm{\sigma}) - \phi(s(\bm\sigma)-\bm{\sigma})\right) s'(\bm{\sigma})=0.
 \end{align*}

Similarly, we have
\begin{align*}
\underset{\bm{\sigma} \rightarrow 0^+}{\lim} \left(\Phi(s(\bm\sigma)+\bm{\sigma}) - \Phi(s(\bm\sigma)-\bm{\sigma})\right) s'(\bm{\sigma})&= \underset{\bm{\sigma} \rightarrow 0^+}{\lim}  \int^{s(\bm\sigma)+\bm{\sigma}}_{s(\bm\sigma)-\bm{\sigma}}e^{-\frac{x^2}{2}}dx s'(\bm{\sigma})\\
& = \underset{\bm{\sigma} \rightarrow 0^+}{\lim} 2   \bm{\sigma} e^{-\frac{s(\bm\sigma)^2}{2}}s'(\bm{\sigma}).
 \end{align*}
Additionally, $\underset{\bm{\sigma} \rightarrow 0^+}{\lim} \left|-2 s(\bm\sigma)e^{-\frac{s(\bm\sigma)^2}{2}}\right|=0$ gives us that
\begin{align*}
\underset{\bm{\sigma} \rightarrow 0^+}{\lim} \left(\Phi(s(\bm\sigma)+\bm{\sigma}) - \Phi(s(\bm\sigma)-\bm{\sigma})\right) s'(\bm{\sigma})=0.\end{align*}

(ii) Suppose that $\underset{\bm{\sigma} \rightarrow 0^+}{\lim} s<+\infty$. On one hand, we have that $\lim\limits_{\bm\sigma\rightarrow 0^+}\frac{d (s(\bm\sigma)\bm\sigma)}{d \bm\sigma}=\lim\limits_{\bm\sigma\rightarrow 0^+}\bm{\sigma}s'(\bm{\sigma}) + \lim\limits_{\bm\sigma\rightarrow 0^+}s(\bm\sigma)=\lim\limits_{\bm\sigma\rightarrow 0^+}\bm{\sigma}s'(\bm{\sigma}) +s(0)$. On the other hand, we have
\begin{align}
\frac{d (s(\bm\sigma)\bm\sigma)}{d \bm\sigma}\big|_{\bm\sigma=0}=\lim\limits_{\bm\sigma\rightarrow 0^+}\frac{s(\bm\sigma)\bm\sigma-0}{\bm\sigma-0}=s(0).
\end{align}
Thus we have
$$\lim\limits_{\bm\sigma\rightarrow 0^+}\bm{\sigma}s'(\bm{\sigma})=0,$$
%
and
\begin{align*}
\underset{\bm{\sigma} \rightarrow 0^+}{\lim} \left(\phi(s(\bm\sigma)+\bm{\sigma}) - \phi(s(\bm\sigma)-\bm{\sigma})\right) s'(\bm{\sigma})&= \underset{\bm{\sigma} \rightarrow 0^+}{\lim}  \int^{s(\bm\sigma)+\bm{\sigma}}_{s(\bm\sigma)-\bm{\sigma}}-xe^{-\frac{x^2}{2}}dx s'(\bm{\sigma})\\
& = \underset{\bm{\sigma} \rightarrow 0^+}{\lim}\left(-2 \bm{\sigma} s'(\bm{\sigma})\right) \cdot  \underset{\bm{\sigma} \rightarrow 0^+}{\lim} s(\bm\sigma)e^{-\frac{s(\bm\sigma)^2}{2}}\\
& =0,\\
\underset{\bm{\sigma} \rightarrow 0^+}{\lim} \left(\Phi(s(\bm\sigma)+\bm{\sigma}) - \Phi(s(\bm\sigma)-\bm{\sigma})\right) s'(\bm{\sigma})&= \underset{\bm{\sigma} \rightarrow 0^+}{\lim}  \int^{s(\bm\sigma)+\bm{\sigma}}_{s(\bm\sigma)-\bm{\sigma}}-xe^{-\frac{x^2}{2}}dx s'(\bm{\sigma})\\
& = \underset{\bm{\sigma} \rightarrow 0^+}{\lim}\left(\bm{\sigma} s'(\bm{\sigma})\right) \cdot  \underset{\bm{\sigma} \rightarrow 0^+}{\lim} e^{-\frac{s(\bm\sigma)^2}{2}}\\
& =0.
 \end{align*}
\end{proof}

\begin{lem}
$\underset{\bm{\sigma} \rightarrow 0^+}{\lim} \mathbf{\widehat{s}} = \mathbf{s}^*$, where $\mathbf{s}^* \in (0, +\infty)$  is determined by the following equation
$$s = \frac{2 \phi (s)}{2- 2\Phi(s)+\frac{\mu_z}{\mu}}.$$
\label{lem_limitzero_hats}
\end{lem}

\begin{proof}

Because $G(\cdot),~\Phi(\cdot)\in C^2$. The implicit function theorem and the uniqueness of $\mathbf{\widehat{s}}$ show that $\mathbf{\widehat{s}}(\bm{\sigma})$ is a continuously differentiable function over $(0, +\infty)$.

When $\bm{\sigma}=0$, we have $\overline{\gamma_e}^{\mathbb{S}}-\underline{\gamma_e}^{\mathbb{S}}=0$ and $\frac{A^{\mathbb{S}}}{\sigma_v}=0$. Equation~\eqref{equation3_scaled} gives us that $B(\mathbf{\widehat{s}})=\frac{1}{2}$ and $\mathbf{\widehat{s}}(\bm\sigma)\bm{\sigma}=0$.

Recall that
\begin{align}
\frac{A^{\mathbb{S}}}{\sigma_v}= \frac{\Phi(\mathbf{\widehat{s}}+\bm{\sigma})-\Phi(\mathbf{\widehat{s}}-\bm{\sigma})}
{2-\Phi(\mathbf{\widehat{s}}+\bm{\sigma})-\Phi(\mathbf{\widehat{s}}
-\bm{\sigma})+(1-G(2B(\mathbf{\widehat{s}})-1))\frac{\mu_z}{\mu}},
\end{align}
$G(\cdot),~\Phi(\cdot)\in C^2$, and $\frac{A^{\mathbb{S}}}{\sigma_v}$ is differentiable of $\bm{\sigma}$ over $(0, +\infty)$.

Taking the derivative, we get
\begin{align*}
\frac{d\left(\frac{A^{\mathbb{S}}}{\sigma_v} \right)}{d \bm{\sigma}}  = & \frac{ \left(\phi(\mathbf{\widehat{s}}+\bm{\sigma}) - \phi(\mathbf{\widehat{s}}-\bm{\sigma})\right) \frac{d \mathbf{\widehat{s}}}{d \bm{\sigma}} + \left(\phi(\mathbf{\widehat{s}}+\bm{\sigma}) + \phi(\mathbf{\widehat{s}}-\bm{\sigma})\right) } {\overline{\gamma_e}+\underline{\gamma_e}+\alpha_e \frac{\mu_z}{\mu}}     \\
& + \frac{\left[\Phi(\mathbf{\widehat{s}}+\bm{\sigma})-\Phi(\mathbf{\widehat{s}}-\bm{\sigma}) \right] \left[ \left(\phi(\mathbf{\widehat{s}}+\bm{\sigma}) + \phi(\mathbf{\widehat{s}}-\bm{\sigma})\right) \frac{d \mathbf{\widehat{s}}}{d \bm{\sigma}} + \left(\phi(\mathbf{\widehat{s}}+\bm{\sigma}) - \phi(\mathbf{\widehat{s}}-\bm{\sigma})\right)\right]}{\left[\overline{\gamma_e}+\underline{\gamma_e}+\alpha_e \frac{\mu_z}{\mu}\right]^2} \\
& + \frac{2 G'(2B(\mathbf{\widehat{s}})-1)\frac{\mu_z}{\mu}\left[\Phi(\mathbf{\widehat{s}}+\bm{\sigma})-\Phi(\mathbf{\widehat{s}}-\bm{\sigma}) \right]\left( \frac{\partial B(\mathbf{\widehat{s}})}{\partial \mathbf{\widehat{s}}}\frac{d \mathbf{\widehat{s}}}{d \bm{\sigma}}  + \frac{\partial B(\mathbf{\widehat{s}})}{\partial \bm{\sigma}}\right)}{\left[\overline{\gamma_e}+\underline{\gamma_e}+\alpha_e \frac{\mu_z}{\mu}\right]^2}.
\end{align*}
Lemma~\ref{lem_limit_zero} gives us
\begin{align}
\underset{\bm{\sigma} \rightarrow 0^+}{\lim}\frac{d\left( \frac{A^{\mathbb{S}}}{\sigma_v} \right)}{d \bm{\sigma}} & =\underset{\bm{\sigma} \rightarrow 0^+}{\lim}\frac{2\phi(\mathbf{\widehat{s}})}{2-2 \Phi(\mathbf{\widehat{s}})+\frac{\mu_z}{\mu}}.
\label{slimit4}
\end{align}

On the other hand, from equation~(\ref{equation3_scaled}), we have
$$\frac{A^{\mathbb{S}}}{\sigma_v}=2B(\mathbf{\widehat{s}})-1.$$
Taking derivative with respect to $\bm{\sigma}$, we get
\begin{align*}
\frac{d\left( \frac{A^{\mathbb{S}}}{\sigma_v} \right)}{d \bm{\sigma}} & = 2B(\mathbf{\widehat{s}}) \left[1-B(\mathbf{\widehat{s}})\right]\left( 2\bm{\sigma} \frac{d \mathbf{\widehat{s}}}{d \bm{\sigma}} + 2\mathbf{\widehat{s}}  \right).
\end{align*}
Using Lemma~\ref{lem_limit_zero} and $\underset{\bm{\sigma} \rightarrow 0^+}{\lim} \mathbf{\widehat{s}}(\bm\sigma)\bm{\sigma} =0$, we obtain
\begin{align}
\underset{\bm{\sigma} \rightarrow 0^+}{\lim}\frac{d\left( \frac{A^{\mathbb{S}}}{\sigma_v} \right)}{d \bm{\sigma}} & = \underset{\bm{\sigma} \rightarrow 0^+}{\lim} \left( \bm{\sigma} \frac{d \mathbf{\widehat{s}}}{d \bm{\sigma}} + \mathbf{\widehat{s}}  \right).
\label{slimit2}
\end{align}

Combing equations (\ref{slimit2}) and (\ref{slimit4}), we have that
$$\underset{\bm{\sigma} \rightarrow 0^+}{\lim} \left( \bm{\sigma} \frac{d \mathbf{\widehat{s}}}{d \bm{\sigma}} + \mathbf{\widehat{s}}  \right)=\underset{\bm{\sigma} \rightarrow 0^+}{\lim}\frac{2\phi(\mathbf{\widehat{s}})}{2-2 \Phi(\mathbf{\widehat{s}})+\frac{\mu_z}{\mu}}.
$$

Suppose that $\underset{\bm{\sigma} \rightarrow 0^+}{\lim} \mathbf{\widehat{s}} = +\infty$, then we have, as we do in the proof of Lemma~\ref{lem_limit_zero}, $\underset{\bm{\sigma} \rightarrow 0^+}{\lim}\bm{\sigma}\frac{d\mathbf{\widehat{s}}}{d\bm{\sigma}} + \mathbf{\widehat{s}} >0,$ which contradicts with $\underset{\bm{\sigma} \rightarrow 0^+}{\lim}\frac{2\phi(\mathbf{\widehat{s}})}{2-2 \Phi(\mathbf{\widehat{s}})+\frac{\mu_z}{\mu}} =0$.

Then we have to show that the limit can not be zero. Because the limit can not be infinity, we have $\underset{\bm{\sigma} \rightarrow 0^+}{\lim}\bm\sigma s'(\bm\sigma)=0$ from Lemma~\ref{lem_limit_zero}. Let $f(s)=\frac{2\phi(s)}{2-2\Phi(s)+\frac{\mu_z}{\mu}}-s$. We can check that there is a unique $\mathbf{s^*}\in (0, +\infty)$ such that $f(\mathbf{s^*})=0$. Therefore,
$$\underset{\bm{\sigma} \rightarrow 0^+}{\lim}\mathbf{\widehat{s}}=\mathbf{s^*} \in (0,+\infty).$$
\end{proof}

We then proceed to prove the propositions.

\textbf{Case I: Without a dark pool}

By Lemma \ref{lem_limit_zero} and Lemma \ref{lem_limitzero_hats}, $ \frac{A^{\mathbb{S}}}{\sigma_v},~\widehat{\alpha_e},~\overline{\gamma_e}^{\mathbb{S}}, ~\underline{\gamma_e}^{\mathbb{S}}$ are differentiable functions of $\bm{\sigma}$, and
$$\underset{\bm{\sigma} \rightarrow 0^+}{\lim}\frac{d\left( \frac{A^{\mathbb{S}}}{\sigma_v} \right)}{d \bm{\sigma}}=\mathbf{s^*} \in (0,+\infty).$$
Also, taking derivative of $B(\mathbf{\widehat{s}})$ with respect to $\bm{\sigma}$, we get
\begin{align}\label{eqn:add_1}
\frac{d B(\mathbf{\widehat{s}})}{d \bm{\sigma}}& = \frac{\partial B(\mathbf{\widehat{s}})}{\partial \mathbf{\widehat{s}}} \frac{d \mathbf{\widehat{s}}}{d \bm{\sigma}} + \frac{\partial B(\mathbf{\widehat{s}})}{\partial \bm{\sigma}}= B(\mathbf{\widehat{s}}) \left(1- B(\mathbf{\widehat{s}}) \right)\left(2 \bm{\sigma}\frac{d \mathbf{\widehat{s}}}{d \bm{\sigma}} + 2 \mathbf{\widehat{s}}  \right).
\end{align}
and the derivative of $\widehat{\alpha_e}$ is
\begin{align*}
\frac{d \widehat{\alpha_e}}{d \bm{\sigma}} = -G'(2 B(\widehat{\mathbf{s}})-1)B(\mathbf{\widehat{s}}) \left(1- B(\mathbf{\widehat{s}}) \right)\left(2 \bm{\sigma}\frac{d \mathbf{\widehat{s}}}{d \bm{\sigma}} + 2 \mathbf{\widehat{s}}  \right).
\end{align*}
When $\bm{\sigma}$ is sufficiently small, we get
\begin{align*}
\underset{\bm{\sigma} \rightarrow 0^+}{\lim} \frac{d \widehat{\alpha_e}}{d\bm{\sigma}} = -\frac{G'(0)\mathbf{s^*}}{2} \in (-\infty, 0).
\end{align*}
Similarly, we take derivative of $\overline{\gamma_e}^{\mathbb{S}}-\underline{\gamma_e}^{\mathbb{S}}$ with respect to $\bm{\sigma}$ and get
\begin{align*}
\frac{d \left(\overline{\gamma_e}^{\mathbb{S}}-\underline{\gamma_e}^{\mathbb{S}} \right)}{d \bm{\sigma}} = \left[\phi(\mathbf{\widehat{s}}+\bm{\sigma}) - \phi(\mathbf{\widehat{s}}-\bm{\sigma}) \right]\frac{d \mathbf{\widehat{s}}}{d \bm{\sigma}} + \left[\phi(\mathbf{\widehat{s}}+\bm{\sigma}) + \phi(\mathbf{\widehat{s}}-\bm{\sigma}) \right],
\end{align*}
and leting $\bm{\sigma} \rightarrow 0^+$, we have
\begin{align*}
\underset{\bm{\sigma} \rightarrow 0^+}{\lim} \frac{d \left(\overline{\gamma_e}^{\mathbb{S}}-\underline{\gamma_e}^{\mathbb{S}} \right)}{d \bm{\sigma}} = 2\phi(\mathbf{\widehat{s}}) \in (0, +\infty).
\end{align*}
Note that $\bm{\sigma} = \frac{\sigma_v}{\sigma_e}$, we conclude the following:

Given $\bm{\sigma}$ sufficiently small, as $\sigma_v$ increases (or $\sigma_e$ decreases),
\begin{itemize}
\item (i) $\frac{A^{\mathbb{S}}}{\sigma_v}$ strictly increases.
\item (ii) $\overline{\gamma_e}^{\mathbb{S}} - \underline{\gamma_e}^{\mathbb{S}}$ strictly increases, and $\alpha_e^{\mathbb{S}}$ strictly decreases.
\end{itemize}

\textbf{Case II, With a dark pool}

Note that when $\bm{\sigma}=0$, we have $\overline{\gamma_e}=\underline{\gamma_e}$ and $\overline{\gamma_d}=\underline{\gamma_d}$. Therefore $\frac{A}{\sigma_v}=0$ and $\bar{R} = \underline{R}$. Equations (\ref{equation0_scaled}) and (\ref{equation1_scaled}) show that $B(\mathbf{s_0})=\frac{1}{2}$ and $B(\mathbf{s_1})=\frac{1}{2}$. If $0<\bm{\sigma}<+\infty$, we have, by Theorem \ref{thm existence}, that $0<\mathbf{s_0}<\mathbf{s_1}<\infty$. Therefore, we have
$\overline{\gamma_e}>\underline{\gamma_e}$, $\overline{\gamma_d}>\underline{\gamma_d}$, $\frac{A}{\sigma_v}>0$, $\bar{R}> \underline{R}$, and $\frac{1}{2}<B(\mathbf{s_0})<B(\mathbf{s_1})<1$. Then we are ready to conclude the following:

Given $\bm{\sigma}$ sufficiently small, as $\sigma_v$ increases (or as $\sigma_e$ decreases),
\begin{itemize}
\item (i) $\frac{A}{\sigma_v}$ increases, and ${\bar{R}}-{\underline{R}}$ increases.
\item (ii) $\overline{\gamma_e} - \underline{\gamma_e}$, $\overline{\gamma_d} - \underline{\gamma_d}$ increases, $\alpha_e$ decreases, and $\alpha_d$ increases.
\end{itemize}

Let $(\mathbf{s_0},~\mathbf{s_1})$ be any equilibrium. Since $G(\cdot)$, and $\Phi(\cdot)$ are twice differentiable, by the implicit function theorem, there  exist continuously differentiable functions $\mathbf{s_0}(\bm{\sigma}),~\mathbf{s_1}(\bm{\sigma})$ defined on $(0,+\infty)$. 

When $\bm{\sigma}\in (0, +\infty)$. By equation (\ref{equation0_scaled}), we have $B(\mathbf{s_0}) = \frac{\bar{R}}{\underline{R}+\bar{R}}\in (0,1)$. Thus rewrite it as
$$\frac{\bar{R}}{\underline{R}} = \frac{1}{\frac{1}{B(\mathbf{s_0}})-1},$$
and the derivative can be found as following:
\begin{align*}
\frac{d\left( \frac{\bar{R}}{\underline{R}} \right)}{d \bm{\sigma}} & = \frac{1}{\underline{R}^2}\left[\frac{d \bar{R}}{d \bm{\sigma}} \underline{R} - \frac{d \underline{R}}{d \bm{\sigma}} \bar{R} \right]\\
& = \frac{1}{\left[ 1- B(\mathbf{s_0})\right]^2}\left( \frac{\partial B(\mathbf{s_0})}{\partial \mathbf{s_0}} \frac{d \mathbf{s_0}}{d \bm{\sigma}} + \frac{\partial B(\mathbf{s_0})}{\partial \bm{\sigma}} \right)\\
& = \frac{B(\mathbf{s_0})}{ 1- B(\mathbf{s_0})}
\left( 2 \bm{\sigma}  \frac{d \mathbf{s_0}}{d \bm{\sigma}} + 2 \mathbf{s_0} \right).
\end{align*}
Also, we know $\underset{\bm{\sigma} \rightarrow 0^+}{\lim} B(\mathbf{s_0}) = \frac{1}{2}$ and $\underset{\bm{\sigma} \rightarrow 0^+}{\lim} \bar{R}= \underset{\bm{\sigma} \rightarrow 0^+}{\lim} \underline{R}=1$, thus
\begin{align*}
\underset{\bm{\sigma}_n \rightarrow 0^+}{\lim}\frac{d \bar{R}}{d \bm{\sigma}} - \underset{\bm{\sigma}_n \rightarrow 0^+}{\lim} \frac{d \underline{R}}{d \bm{\sigma}} & = \underset{\bm{\sigma}_n \rightarrow 0^+}{\lim}\left( 2 \bm{\sigma}  \frac{d \mathbf{s_0}}{d \bm{\sigma}} + 2 \mathbf{s_0} \right).
\end{align*}

Equation (\ref{equation1_scaled}) shows that $$2B(\mathbf{s_1})-1-\left[B(\mathbf{s_1})\underline{R}-(1-B(\mathbf{s_1}))\bar{R}\right]=\frac{A}{\sigma_v}.$$
Taking derivative on both sides, we get
\begin{align*}
\frac{d\left( \frac{A}{\sigma_v} \right)}{d \bm{\sigma}} & = (2-\underline{R}-\bar{R})B(\mathbf{s_0})\left[1-B(\mathbf{s_0})\right]\left( \bm{\sigma} \frac{d \mathbf{s_1}}{d \bm{\sigma}} + \mathbf{s_1}  \right)+\left[1-B(\mathbf{s_0})\right] \frac{d \bar{R}}{d \bm{\sigma}} - B(\mathbf{s_0})  \frac{d \underline{R}}{d \bm{\sigma}},
\end{align*}
and because $\frac{A}{\sigma_v}= \frac{\Phi(\mathbf{s_1}+\bm{\sigma})-\Phi(\mathbf{s_1}-\bm{\sigma})}{2-\Phi(\mathbf{s_1}+\bm{\sigma})-\Phi(\mathbf{s_1}
-\bm{\sigma})+(1-G(2B(\mathbf{s_1})-1))\frac{\mu_z}{\mu}}$,
we have
\begin{align*}
&\frac{d\left( \frac{A}{\sigma_v} \right)}{d \bm{\sigma}} \\
 = & \frac{ \left(\phi(\mathbf{s_1}+\bm{\sigma}) - \phi(\mathbf{s_1}-\bm{\sigma})\right) \frac{d \mathbf{s_1}}{d \bm{\sigma}} + \left(\phi(\mathbf{s_1}+\bm{\sigma}) + \phi(\mathbf{s_1}-\bm{\sigma})\right) } {\overline{\gamma_e}+\underline{\gamma_e}+\alpha_e \frac{\mu_z}{\mu}}     \\
& + \frac{\left[\Phi(\mathbf{s_1}+\bm{\sigma})-\Phi(\mathbf{s_1}-\bm{\sigma}) \right] \left[ \left(\phi(\mathbf{s_1}+\bm{\sigma}) + \phi(\mathbf{s_1}-\bm{\sigma})\right) \frac{d \mathbf{s_1}}{d \bm{\sigma}} + \left(\phi(\mathbf{s_1}+\bm{\sigma}) - \phi(\mathbf{s_1}-\bm{\sigma})\right)\right]}{\left[\overline{\gamma_e}+\underline{\gamma_e}+\alpha_e \frac{\mu_z}{\mu}\right]^2} \\
& + \frac{2 G'(2B(\mathbf{s_1})-1)\frac{\mu_z}{\mu}\left[\Phi(\mathbf{s_1}+\bm{\sigma})-\Phi(\mathbf{s_1}-\bm{\sigma}) \right]\left( \frac{\partial B(\mathbf{s_1})}{\partial \mathbf{s_1}}\frac{d \mathbf{s_1}}{d \bm{\sigma}}  + \frac{\partial B(\mathbf{s_1})}{\partial \bm{\sigma}}\right)}{\left[\overline{\gamma_e}+\underline{\gamma_e}+\alpha_e \frac{\mu_z}{\mu}\right]^2}.
\end{align*}

Similarly to what we shown in the proof of Lemma \ref{lem_limitzero_hats}, we obtain
\begin{align}
\underset{\bm{\sigma} \rightarrow 0^+}{\lim}\frac{d\left( \frac{A}{\sigma_v} \right)}{d \bm{\sigma}} & = \frac{1}{2}\left(\underset{\bm{\sigma} \rightarrow 0^+}{\lim}\frac{d \bar{R}}{d \bm{\sigma}} - \underset{\bm{\sigma} \rightarrow 0^+}{\lim} \frac{d \underline{R}}{d \bm{\sigma}}\right) = \underset{\bm{\sigma}_n \rightarrow 0^+}{\lim}\left( \bm{\sigma}  \frac{d \mathbf{s_0}}{d \bm{\sigma}} +  \mathbf{s_0} \right),
\label{slimit1}
\end{align}
and
\begin{align}
\underset{\bm{\sigma} \rightarrow 0^+}{\lim}\frac{d\left( \frac{A}{\sigma_v} \right)}{d \bm{\sigma}} & = \underset{\bm{\sigma} \rightarrow 0^+}{\lim}\frac{2\phi(\mathbf{s_1})}{2-2 \Phi(\mathbf{s_1})+\frac{\mu_z}{\mu}}
\label{slimit3}.
\end{align}
Combing equations (\ref{slimit1}) and (\ref{slimit3}) gives us
$$\underset{\bm{\sigma} \rightarrow 0^+}{\lim}\left( \bm{\sigma}  \frac{d \mathbf{s_0}}{d \bm{\sigma}} + \mathbf{s_0} \right)=\underset{\bm{\sigma} \rightarrow 0^+}{\lim}\frac{2\phi(\mathbf{s_1})}{2-2 \Phi(\mathbf{s_1})+\frac{\mu_z}{\mu}}$$

Suppose $\underset{\bm{\sigma} \rightarrow 0^+}{\lim} \mathbf{s_0} =+\infty$. Using the similar argument as in the proof of Lemma \ref{lem_limitzero_hats}, we obtain $\underset{\bm{\sigma} \rightarrow 0^+}{\lim}\left( \bm{\sigma}  \frac{d \mathbf{s_0}}{d \bm{\sigma}} + \mathbf{s_0} \right)>0$. However, as $\mathbf{s_0} \rightarrow +\infty$, we have $\mathbf{s_1} \rightarrow +\infty$ and $\frac{2\phi(\mathbf{s_1})}{2-2 \Phi(\mathbf{s_1})+\frac{\mu_z}{\mu}} \rightarrow 0$. This is a contradiction. Therefore, it must be that $\underset{\bm{\sigma} \rightarrow 0^+}{\lim} \mathbf{s_0} <+\infty$.

By Lemma \ref{lem_limit_zero}, $\underset{\bm{\sigma} \rightarrow 0^+}{\lim} \bm{\sigma}  \frac{d \mathbf{s_0}}{d \bm{\sigma}}=0$. So we have
$$\underset{\bm{\sigma} \rightarrow 0^+}{\lim} \mathbf{s_0}=\underset{\bm{\sigma} \rightarrow 0^+}{\lim}\frac{2\phi(\mathbf{s_1})}{2-2 \Phi(\mathbf{s_1})+\frac{\mu_z}{\mu}}.$$
Define $\underset{\bm{\sigma} \rightarrow 0^+}{\lim} \mathbf{s_0} \overset{\triangle}{=}\mathbf{s_0}(0^+)$, $\underset{\bm{\sigma} \rightarrow 0^+}{\lim} \mathbf{s_1} \overset{\triangle}{=}\mathbf{s_1}(0^+)$, and we have
\begin{align*}
&\underset{\bm{\sigma} \rightarrow 0^+}{\lim}\frac{d\left( \frac{A}{\sigma_v} \right)}{d \bm{\sigma}}=\underset{\bm{\sigma} \rightarrow 0^+}{\lim} \mathbf{s_0}=\mathbf{s_0}(0^+)\geq 0,\\
&\underset{\bm{\sigma} \rightarrow 0^+}{\lim} \frac{d \left(\overline{\gamma_e}-\underline{\gamma_e} \right)}{d \bm{\sigma}} = 2\phi(\mathbf{s_1(0^+)})\geq 0,\\
&\underset{\bm{\sigma} \rightarrow 0^+}{\lim}\frac{d \alpha_e}{d \bm{\sigma}} = -\frac{G'(0) \mathbf{s_1}(0^+)}{2}\leq 0, \\
&\underset{\bm{\sigma} \rightarrow 0^+}{\lim}\frac{d \alpha_d}{d \bm{\sigma}} = \frac{G'(0) (\mathbf{s_1(0^+)}-\mathbf{s_0}(0^+))}{2}\geq 0,\\
&\underset{\bm{\sigma} \rightarrow 0^+}{\lim}\frac{d \left(\alpha_e+\alpha_d\right)}{d \bm{\sigma}} = -\frac{G'(0) \mathbf{s_0}(0^+)}{2} \leq 0,
\end{align*}
which conclude the proof.
\qed
\subsection{Proof of Proposition \ref{prop_compare}}
\label{proof_prop_compare}
To prove Proposition \ref{prop_compare}, we need the following lemmas.
\begin{lem}
For any given $\bm{\sigma} \in (0, +\infty)$, $\widehat{\mathbf{s}}(\bm\sigma)<\mathbf{s}_1(\bm\sigma)$.
\label{lem_HATSlessthanS}
\end{lem}
\begin{proof}
Substitute the expressions of $\frac{A}{\sigma_v}$ into equation~\eqref{equation3_scaled} and~\eqref{equation1_scaled}, then $\mathbf{\widehat{s}},~\mathbf{s_1}$ are respectively determined by the following two equations
\begin{align*}
\frac{\Phi(\mathbf{\widehat{s}}+\bm{\sigma})-\Phi(\mathbf{\widehat{s}}-\bm{\sigma})}
{2-\Phi(\mathbf{\widehat{s}}+\bm{\sigma})-\Phi(\mathbf{\widehat{s}}-\bm{\sigma})+(1-G(2B(\mathbf{\widehat{s}})-1))\frac{\mu_z}{\mu}}&=
2B(\mathbf{\widehat{s}})-1,\\
\frac{\Phi(\mathbf{s_1}+\bm{\sigma})-\Phi(\mathbf{s_1}-\bm{\sigma})}{2-\Phi(\mathbf{s_1}+\bm{\sigma})
-\Phi(\mathbf{s_1}-\bm{\sigma})+(1-G(2B(\mathbf{s_1})-1))\frac{\mu_z}{\mu}}&	=
2B(\mathbf{\mathbf{s_1}})-1\\
& \quad-\left[B(\mathbf{\mathbf{s_1}})\underline{R}-(1-B(\mathbf{s_1}))\bar{R}\right].
\end{align*}
Let $f(s) = \frac{\Phi(s+\bm{\sigma})-\Phi(s-\bm{\sigma})}{2-\Phi(s+\bm{\sigma})
-\Phi(s-\bm{\sigma})+(1-G(2B(s)-1))\frac{\mu_z}{\mu}}$, and its derivative is
\begin{align*}
f'(s)&=\frac{D_1(s)+D_2(s)}{\left[{2-\Phi(s+\bm{\sigma})
-\Phi(s-\bm{\sigma})+(1-G(2B(s)-1))\frac{\mu_z}{\mu}}\right]^2},
\end{align*}
where
\begin{align*}
D_1(s)&=\left(\phi(s+\bm{\sigma})-\phi(s-\bm{\sigma})\right)\left( 2-\Phi(s+\bm{\sigma})
-\Phi(s-\bm{\sigma})+(1-G(2B(s)-1))\frac{\mu_z}{\mu}\right)<0,\\
D_2(s)& = -\left(\Phi(s+\bm{\sigma})-\Phi(s-\bm{\sigma})\right)\left( -\phi(s+\bm{\sigma})
-\phi(s-\bm{\sigma})-2G'(2B(s)-1)B'(s)\frac{\mu_z}{\mu}\right)>0.
\end{align*}
Since $G'(s)+sG''(s)\geq 0$, one can represent $f(s)$ as the blue curve in Figure~\ref{fig_higher}.

Let $\widehat{h}(s) =2B(s)-1 $ and $h(s) = 2B(s)-1-\left[B(s)\underline{R}-(1-B(s))\bar{R}\right]$. By equation~(\ref{equation0_scaled}), for any $s>\mathbf{s_0} $, we have $B(s)>B(\mathbf{s_0}) = \frac{\bar{R}}{\underline{R}+\bar{R}}$. That is, $\left[B(s)\underline{R}-(1-B(s))\bar{R}\right]>0$. Therefore $\widehat{h}(s)>h(s)$. In Figure \ref{fig_higher}, $\widehat{h}(s)$ is represented by the red curve, while $h(s)$ is represented by the green curve which is below $\widehat{h}(s)$. Obviously, the intersection point $\mathbf{s_1}$ is larger than $\mathbf{\widehat{s}}$. The Lemma is proved.
\begin{figure}[h]
\centering
\includegraphics[width=0.7\textwidth]{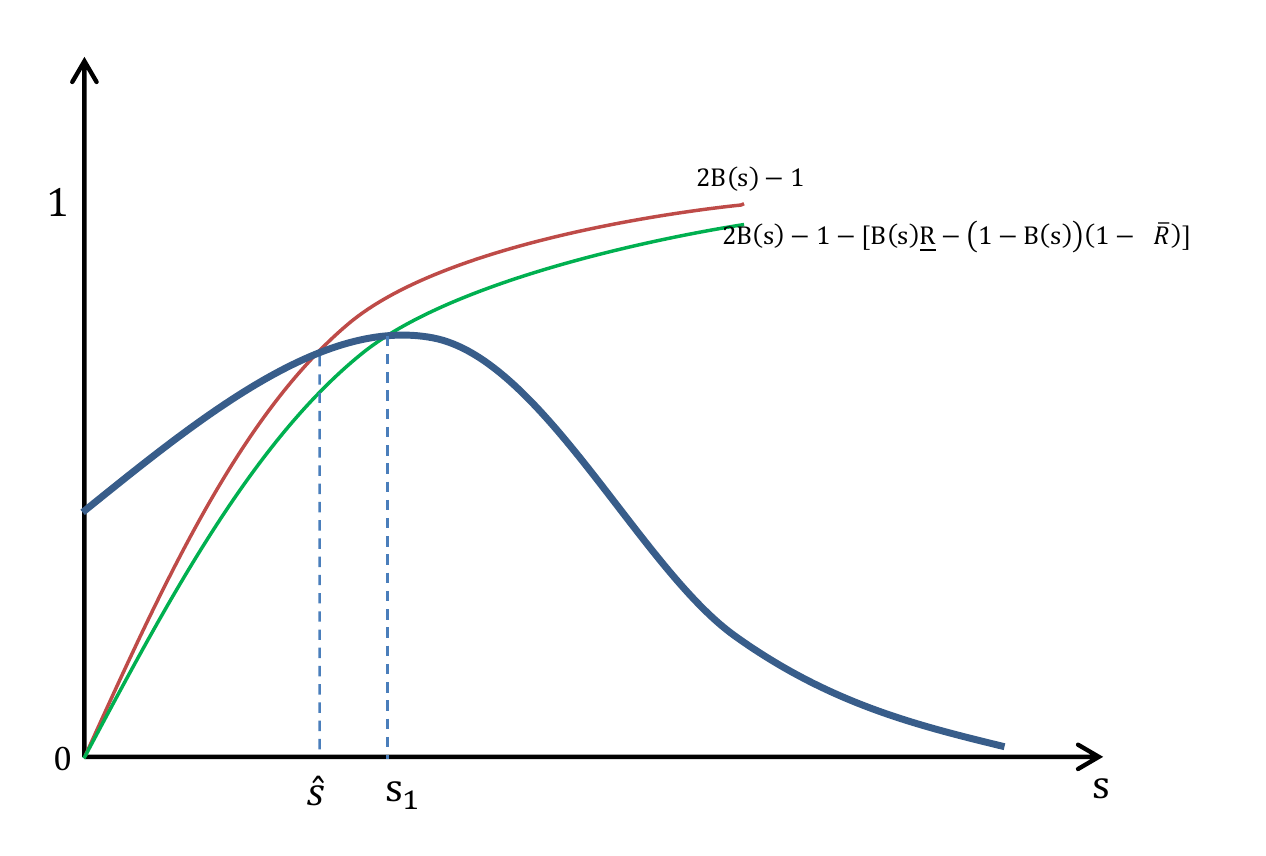}
\caption{$h(s)$ and $\widehat{h}(s)$}
\label{fig_higher}
\end{figure}
\end{proof}

\begin{lem}
If $\bm{\sigma} \rightarrow +\infty$, there exists a unique $\widehat{k}\in(\frac{1}{2},1)$ such that
$\underset{\bm{\sigma} \rightarrow +\infty}{\lim}\overline{\gamma_e}^{\mathbb{S}}=1$,
$\underset{\bm{\sigma} \rightarrow +\infty}{\lim}\underline{\gamma_e}^{\mathbb{S}} =0 $,
$\underset{\bm{\sigma} \rightarrow +\infty}{\lim} \alpha_e^{\mathbb{S}} = 1-G(\widehat{k})$,
and $\underset{\bm{\sigma} \rightarrow +\infty}{\lim} \frac{A^{\mathbb{S}}}{\sigma_v}=\hat k$,
where $\widehat{k}$ is determined by
\begin{align}
\widehat{k}=\frac{1}{1+\left[1-G(\widehat{k})\right]\frac{\mu_z}{\mu}}.
\label{equation_limit3}
\end{align}
In addition, such $\widehat{k}$ is smaller if $\frac{\mu_z}{\mu}$ is larger.
\label{lem_limitinfty_hats}
\end{lem}
\begin{proof}
Suppose $\underset{\bm{\sigma} \rightarrow +\infty}{\lim} \mathbf{\widehat{s}}\bm{\sigma} = +\infty$. Then, when $\bm\sigma\rightarrow +\infty$, we have $ 2B(\mathbf{\widehat{s}})-1 \rightarrow 1$. Thus equation~(\ref{equation3_scaled}) gives us that $ \frac{A^{\mathbb{S}}}{\sigma_v}=1$. However $\widehat{\alpha_e} = 1-G(1)>0$, which implies $\frac{A^{\mathbb{S}}}{\sigma_v}<1$. Therefore, we have
$$\underset{\bm{\sigma} \rightarrow +\infty}{\lim} \mathbf{\widehat{s}}\bm{\sigma} < +\infty.$$

Let $\underset{\bm{\sigma} \rightarrow +\infty}{\lim} \mathbf{\widehat{s}}\bm{\sigma} = \widehat{C}\in [0,+\infty)$, where $\widehat C$ will be determined later. Then we have $\underset{\bm{\sigma} \rightarrow +\infty}{\lim} \mathbf{\widehat{s}}=0$, thus $\underset{\bm{\sigma} \rightarrow +\infty}{\lim} (\mathbf{\widehat{s}}-\bm{\sigma})=-\infty$. Therefore, $\underset{\bm{\sigma} \rightarrow +\infty}{\lim}\overline{\gamma_e}^{\mathbb{S}}=1$, $\underset{\bm{\sigma} \rightarrow +\infty}{\lim}\underline{\gamma_e}^{\mathbb{S}} =0$. Let $\hat k= \lim\limits_{\bm\sigma\rightarrow+\infty}2B({\mathbf{\widehat s}})-1=\frac{1-e^{-2\widehat{C}}}{1+e^{-2\widehat{C}}}$, and we have $\underset{\bm{\sigma} \rightarrow +\infty}{\lim} \alpha_e^{\mathbb{S}} = 1-G(\widehat{k})$ and
$\underset{\bm{\sigma} \rightarrow +\infty}{\lim} \frac{A^{\mathbb{S}}}{\sigma_v}=\frac{1}{1+\left[1-G(\widehat{k})\right]\frac{\mu_z}{\mu}}$. However, $\hat k$ has to satisfy equation~\eqref{equation_limit3} such that equation~\eqref{equation3_scaled} is satisfied.

Let $f(k) = k-\frac{1}{1+\left[1-G(k)\right]\frac{\mu_z}{\mu}}$, and we can easily verify that $f(0)<0$, and $f(1)>0 $. Therefore, there exists a $\widehat{k}\in (0,1)$ such that $f(\widehat{k})=0$, and $\widehat{C}=\frac{1}{2}\ln\frac{1+\widehat{k}}{1-\widehat{k}}$.
\end{proof}

\begin{lem} Let $R = \mathbb{E} \left[ \min \left\{1, \frac{Z^+}{Z^-} \right\} \right]$. Consider any equilibrium $\mathbf{s_0}(\bm{\sigma}), \mathbf{s_1}(\bm{\sigma})$ for $\bm{\sigma} \rightarrow +\infty$. We have $\underset{\bm{\sigma} \rightarrow +\infty}{\lim} \mathbf{s_0}\bm{\sigma} < +\infty$. In addition, the limits of variables can be determined in the following two statements.
\begin{itemize}
\item (i) If $\underset{\bm{\sigma} \rightarrow +\infty}{\lim} \mathbf{s_1}\bm{\sigma} < +\infty$, we have $\underset{\bm{\sigma} \rightarrow +\infty}{\lim}\overline{\gamma_e}=1$, $\underset{\bm{\sigma} \rightarrow +\infty}{\lim}\underline{\gamma_e} =0 $,
    $\underset{\bm{\sigma} \rightarrow +\infty}{\lim} \overline{\gamma_d} =0$,
    $\underset{\bm{\sigma} \rightarrow +\infty}{\lim}\underline{\gamma_d} =0$,
    $\underset{\bm{\sigma} \rightarrow +\infty}{\lim} \alpha_e = 1-G(k_1)$,
    $\underset{\bm{\sigma} \rightarrow +\infty}{\lim} \alpha_d=G(k_1)$,
    $\underset{\bm{\sigma} \rightarrow +\infty}{\lim} \frac{A}{\sigma_v}=\frac{1}{1+\left[1-G(k_1)\right]\frac{\mu_z}{\mu}}$,
    $\underset{\bm{\sigma} \rightarrow +\infty}{\lim}\bar{R} =R$, and
    $ \underset{\bm{\sigma} \rightarrow +\infty}{\lim}\underline{R} =R$, where $k_1\in(\frac{1}{2},1)$ is determined by \begin{align}
    (1-R)k_1=\frac{1}{1+\left[1-G(k_1)\right]\frac{\mu_z}{\mu}}.
    \label{equation_limit00}
    \end{align}
\item (ii) If $\underset{\bm{\sigma} \rightarrow +\infty}{\lim} \mathbf{s_0}\bm{\sigma} = +\infty$, we have $\underset{\bm{\sigma} \rightarrow +\infty}{\lim}\overline{\gamma_e}=1-k_3$,
$\underset{\bm{\sigma} \rightarrow +\infty}{\lim}\underline{\gamma_e} =0$,
$\underset{\bm{\sigma} \rightarrow +\infty}{\lim} \overline{\gamma_d} = k_3$,
$\underset{\bm{\sigma} \rightarrow +\infty}{\lim}\underline{\gamma_d} =0$,
$\underset{\bm{\sigma} \rightarrow +\infty}{\lim} \alpha_e = 1-G(1)$,
$\underset{\bm{\sigma} \rightarrow +\infty}{\lim} \alpha_d=G(1)-G(2k_2-1)$,
$\underset{\bm{\sigma} \rightarrow +\infty}{\lim} \frac{A}{\sigma_v}=\frac{1-k_3}{1-k_3+\left[1-G(1) \right]\frac{\mu_z}{\mu}}$,
$\underset{\bm{\sigma} \rightarrow +\infty}{\lim}\bar{R} =\frac{k_2}{1-k_2}\frac{\left[1-G(1) \right]\frac{\mu_z}{\mu}}{1-k_3 + \left[1-G(1) \right]\frac{\mu_z}{\mu}}$, and
$\underset{\bm{\sigma} \rightarrow +\infty}{\lim}\underline{R}=\frac{\left[1-G(1) \right]\frac{\mu_z}{\mu}}{1-k_3 + \left[1-G(1) \right]\frac{\mu_z}{\mu}}$, where $k_2 \in [\frac{1}{2},1)$ and $k_3 \in [0,1)$ are determined by
\begin{align}
&\frac{\left[1-G(1) \right]\frac{\mu_z}{\mu}}{1-k_3 + \left[1-G(1) \right]\frac{\mu_z}{\mu}}   =  \mathbb{E}\left[\min\left\{1, \frac{Z^-}{\frac{k_3}{G(1)-G(2k_2-1)}+Z^+} \right\} \right],
\label{equation_limit0}\\
&k_2  = \frac{\mathbb{E}\left[\min\left\{1, \frac{\frac{k_3}{G(1)-G(2k_2-1)}+Z^+} {Z^-}\right\} \right]}{\mathbb{E}\left[\min\left\{1, \frac{\frac{k_3}{G(1)-G(2k_2-1)}+Z^+} {Z^-}\right\} \right]+ \mathbb{E}\left[\min\left\{1, \frac{Z^-}{\frac{k_3}{G(1)-G(2k_2-1)}+Z^+} \right\} \right]}.
\label{equation_limit1}
\end{align}
\end{itemize}
\label{lem_limitinfty_s}

\end{lem}
\begin{proof}
Consider any continuously differentiable functions $\mathbf{s_0}(\bm{\sigma}),\mathbf{s_1}(\bm{\sigma})$.

First we show $\underset{\bm{\sigma} \rightarrow +\infty}{\lim} \mathbf{s_0}\bm{\sigma} < +\infty$ by contradiction. Suppose that $\underset{\bm{\sigma} \rightarrow +\infty}{\lim} \mathbf{s_0}\bm{\sigma} = +\infty$, we have $B(\mathbf{s_0})=\frac{1}{1+e^{-2\mathbf{s_0}\bm{\sigma}}}\rightarrow 1$. Since $\mathbf{s_1}>\mathbf{s_0}$, we have $\underset{\bm{\sigma} \rightarrow +\infty}{\lim} \mathbf{s_1}\bm{\sigma} = +\infty$, $B(\mathbf{s_1})=\frac{1}{1+e^{-2\mathbf{s_1}\bm{\sigma}}}\rightarrow 1$. In addition, Equation~\eqref{equation0_scaled} gives us that $\frac{\bar{R}}{\underline{R}}=1 $, i.e., $\bar{R}=\underline{R}$.

If $\underset{\bm{\sigma} \rightarrow +\infty}{\lim}\left(\mathbf{s_0}-\bm{\sigma}\right)<+\infty$, then $\overline{\gamma_d}>0=\underline{\gamma_d}$, which is a contradiction to $\bar{R}=\underline{R}$. If $\underset{\bm{\sigma} \rightarrow +\infty}{\lim}\left(\mathbf{s_0}-\bm{\sigma}\right)=+\infty$, then $\underset{\bm{\sigma} \rightarrow +\infty}{\lim}\mathbf{s_1}-\bm{\sigma}=+\infty$. Therefore we have $\underset{\bm{\sigma} \rightarrow +\infty}{\lim}\overline{\gamma_e}=\underset{\bm{\sigma} \rightarrow +\infty}{\lim} \underline{\gamma_e}=0$, $\underset{\bm{\sigma} \rightarrow +\infty}{\lim} \frac{A}{\sigma_v}=0$, and by equation (\ref{equation1_scaled}), we have $\underset{\bm{\sigma} \rightarrow +\infty}{\lim} B(\mathbf{s_1})=\underset{\bm{\sigma} \rightarrow +\infty}{\lim}\frac{1-\bar{R}}{1-\bar{R}+1-\underline{R}}=\frac{1}{2}$, whche is a contradiction to $\underset{\bm{\sigma} \rightarrow +\infty}{\lim} B(\mathbf{s_1})=1$. Therefore, we have $$\underset{\bm{\sigma} \rightarrow +\infty}{\lim} \mathbf{s_0}\bm{\sigma}=C\in [0, +\infty).$$


Then we show the two statements.
\begin{itemize}
\item  (i) Suppose that $\underset{\bm{\sigma} \rightarrow +\infty}{\lim} \mathbf{s_0}\bm{\sigma}=C_0 \in [0,+\infty)$ and $\underset{\bm{\sigma} \rightarrow +\infty}{\lim} \mathbf{s_1}\bm{\sigma}=C_1 \in[0, +\infty)$, then we have $\underset{\bm{\sigma} \rightarrow +\infty}{\lim}\left(\mathbf{s_0}-\bm{\sigma}\right)\rightarrow -\infty$ and $\underset{\bm{\sigma} \rightarrow +\infty}{\lim}\left(\mathbf{s_1}-\bm{\sigma}\right)\rightarrow -\infty$. Therefore, $\underset{\bm{\sigma} \rightarrow +\infty}{\lim}\overline{\gamma_e}=1$, $\underset{\bm{\sigma} \rightarrow +\infty}{\lim}\underline{\gamma_e} =\underset{\bm{\sigma} \rightarrow +\infty}{\lim} \overline{\gamma_d} =\underset{\bm{\sigma} \rightarrow +\infty}{\lim}\underline{\gamma_d} =0$.

\qquad We show that $\underset{\bm{\sigma} \rightarrow +\infty}{\lim}\bar{R}=\underset{\bm{\sigma} \rightarrow +\infty}{\lim}\underline{R}=R$.
If $C_0=C_1$, Lemma~\ref{lem_Rlimit} and $\underset{\bm{\sigma} \rightarrow +\infty}{\lim} \frac{\phi(\mathbf{s_0}-\bm{\sigma})}{B'(\mathbf{s_0})}=\underset{\bm{\sigma} \rightarrow +\infty}{\lim} \frac{\phi(\mathbf{s_0}+\bm{\sigma})}{B'(\mathbf{s_0})}=0$ give us that $\underset{\bm{\sigma} \rightarrow +\infty}{\lim}\bar{R}=\underset{\bm{\sigma} \rightarrow +\infty}{\lim}\underline{R}=\mathbb{E} \left[ \min \left\{1, \frac{Z^+}{Z^-} \right\} \right]=R$. If $C_0 <C_1$, because $\underset{\bm{\sigma} \rightarrow +\infty}{\lim} \overline{\gamma_d} = \underset{\bm{\sigma} \rightarrow +\infty}{\lim} \underline{\gamma_d}=0$ and $\underset{\bm{\sigma} \rightarrow +\infty}{\lim}\alpha_d>0$, we have $\underset{\bm{\sigma} \rightarrow +\infty}{\lim}\bar{R}=\underset{\bm{\sigma} \rightarrow +\infty}{\lim}\underline{R}=\mathbb{E} \left[ \min \left\{1, \frac{Z^+}{Z^-} \right\} \right]=R$.

\qquad Then equation~\eqref{equation0_scaled} gives us that $\lim\limits_{\bm{\sigma} \rightarrow +\infty}B(\mathbf{s_0})=1/2$ and $\lim\limits_{\bm{\sigma} \rightarrow +\infty}\mathbf{s_0}\bm\sigma=0$. Let $k_1=\lim\limits_{\bm{\sigma} \rightarrow +\infty} 2B(\mathbf{s_1})-1$, and we have  $\underset{\bm{\sigma} \rightarrow +\infty}{\lim} \alpha_e = 1-G(k_1)$, $\underset{\bm{\sigma} \rightarrow +\infty}{\lim} \alpha_d=G(k_1)$ and $\underset{\bm{\sigma} \rightarrow +\infty}{\lim} \frac{A}{\sigma_v}=\frac{1}{1+\left[1-G(k_1) \right]\frac{\mu_z}{\mu}}$. Rewrite equation~\eqref{equation1_scaled} in the following form
$$\left(2B(\mathbf{s_1})-1\right)(1-R)=\frac{1}{1+\left[1-G(2B(\mathbf{s_1})-1)\right]\frac{\mu_z}{\mu}},$$
and $k_1$ has to satisfy equation~\eqref{equation_limit00}. 

\qquad Let $f(k) = (1-R)k-\frac{1}{1+\left[1-G(k)\right]\frac{\mu_z}{\mu}}$. We can verify that $f(0)<0$ and $f(1)>0$ if $1+\left[1-G(1)\right]\frac{\mu_z}{\mu} > \frac{1}{1-R}$. There is a $k_1 \in (0, 1)$ such that $f(k_1)=0$, and $C_1 =\frac{1}{2}\ln \frac{1+k_1}{1-k_1}$.

\item (ii) Suppose that $\underset{\bm{\sigma} \rightarrow +\infty}{\lim} \mathbf{s_0}\bm{\sigma}=C_2 \in [0,+\infty)$ and $\underset{\bm{\sigma} \rightarrow +\infty}{\lim} \mathbf{s_1}\bm{\sigma} = +\infty$. We have $\underset{\bm{\sigma} \rightarrow +\infty}{\lim}\underline{\gamma_e} =0$, $\underset{\bm{\sigma} \rightarrow +\infty}{\lim}\underline{\gamma_d} =0$, and $\underset{\bm{\sigma} \rightarrow +\infty}{\lim} \alpha_e = 1-G(1)$.

\qquad Suppose taht $\underset{\bm{\sigma} \rightarrow +\infty}{\lim} \left(\mathbf{s_1}-\bm{\sigma}\right) = C_3 \in [-\infty, +\infty]$. Let $k_2=\underset{\bm{\sigma} \rightarrow +\infty}{\lim}B(\mathbf{s_0})=\frac{1}{1+e^{-2C_2}}\in [\frac{1}{2},1)$ and $k_3=\underset{\bm{\sigma} \rightarrow +\infty}{\lim}\overline{\gamma_d}=\Phi(C_3)\in [0,1]$. Then we have $\underset{\bm{\sigma} \rightarrow +\infty}{\lim}\overline{\gamma_e}=1-k_3$, $\underset{\bm{\sigma} \rightarrow +\infty}{\lim} \overline{\gamma_d} = k_3$, $\underset{\bm{\sigma} \rightarrow +\infty}{\lim} \alpha_d=G(1)-G(2k_2-1)$, $\underset{\bm{\sigma} \rightarrow +\infty}{\lim} \frac{A}{\sigma_v}=\frac{1-k_3}{1-k_3+\left[1-G(1) \right]\frac{\mu_z}{\mu}}$. Combining equations~\eqref{equation0_scaled} and~\eqref{equation1_scaled}, we have $\underset{\bm{\sigma} \rightarrow +\infty}{\lim}\bar{R} =\frac{k_2}{1-k_2}\frac{\left[1-G(1) \right]\frac{\mu_z}{\mu}}{1-k_3 + \left[1-G(1) \right]\frac{\mu_z}{\mu}}$,
$\underset{\bm{\sigma} \rightarrow +\infty}{\lim}\underline{R}=\frac{\left[1-G(1) \right]\frac{\mu_z}{\mu}}{1-k_3 + \left[1-G(1) \right]\frac{\mu_z}{\mu}}$. In addition, by equations~\eqref{R_upper_scaled} and~\eqref{R_lower_scaled}, $k_2$ and $k_3$ have to satisfy equations~\eqref{equation_limit0} and~\eqref{equation_limit1}.

\qquad Suppose that $1+\left[1-G(1)\right]\frac{\mu_z}{\mu} \leq\frac{1}{1-R}$. For equation (\ref{equation_limit0}), the left hand side is increasing with respect to $k_3$, while the right hand side is decreasing with respect to $k_3$. In addition, when $k_3=0$, $LHS-RHS = \frac{\left[1-G(1) \right]\frac{\mu_z}{\mu}}{1+ \left[1-G(1) \right]\frac{\mu_z}{\mu}} - \mathbb{E}\left[\min\left\{1, \frac{Z^-}{Z^+} \right\} \right] =1-R -\frac{1}{1+\left[1-G(1) \right]\frac{\mu_z}{\mu}}\leq 0 $, and when $k_3 =1$, $LHS-RHS = 1-\mathbb{E}\left[\min\left\{1, \frac{Z^-}{\frac{1}{G(1)-G(2k_2-1)}+Z^+} \right\} \right]>0$. Thus, given any $k_2\in [\frac{1}{2},1)$, there exists a unique $k_3(k_2)\in (0,1)$ that solves equation (\ref{equation_limit0}). Furthermore, as $k_2$ increases, the right hand side of equation~\eqref{equation_limit0} decreases, thus $k_3(k_2)$ is decreasing with respect to  $k_2$. When $k_2\rightarrow 1$, we have $k_3(k_2) \rightarrow 0$. Thus $\underline{R}\rightarrow  \frac{\left[1-G(1) \right]\frac{\mu_z}{\mu}}{1+ \left[1-G(1) \right]\frac{\mu_z}{\mu}}\geq 0$.


\qquad For equation (\ref{equation_limit1}), we substitute $k_3$ with the expression solved from~\eqref{equation_limit0}, and it becomes a function of $k_2$ only. When $k_2=1/2$, we have $k_3\in[0,1)$. Then $LHS-RHS\leq 0$. While when $k_2\rightarrow 1$, we have $LHS-RHS \geq0$. Therefore, there exist $k_2\in[1/2,1)$ and $k_3\in [0,1)$ such that equations~\eqref{equation_limit0} and~\eqref{equation_limit1} are satisfied. Additionally, we have $C_2 = \frac{1}{2}\ln \frac{k_2}{1-k_2}$, $C_3 = \Phi^{-1}(k_3)$.
\end{itemize}
\end{proof}

We now proceed to prove the proposition.
From Lemma~\ref{lem_HATSlessthanS}, we have $\hvs < \vs_1$ for all $\bm\sigma\in (0,+\infty)$. Thus, $\gammaesu-\gammaesl=\Phi(\hvs+\sigma)-\Phi(\hvs-\sigma) >\Phi(\vs_1+\sigma)-\Phi(\vs_1-\sigma) = \gammaeu-\gammael$ and $\alphaes=1-G(2B(\hvs)-1)>1-G(2B(\vs_1)-1)=\alpha_e$.

Let $\widehat{k},~k_1,~k_2,~k_3$ as in   (\ref{equation_limit3}), (\ref{equation_limit00}), (\ref{equation_limit0}), (\ref{equation_limit1}). Suppose  $1-R> \frac{1}{1+[1-G(\widehat{k})]\frac{\mu_z}{\mu}}$, then as $\bm{\sigma}\rightarrow +\infty$, by Lemma~\ref{lem_limitinfty_hats}, we have
\begin{align*}
&\underset{\bm{\sigma} \rightarrow +\infty}{\lim} \frac{A^{\mathbb{S}}}{\sigma_v}=\frac{1}{1+\left[1-G(\widehat{k}) \right]\frac{\mu_z}{\mu}},\\
&\underset{\bm{\sigma} \rightarrow +\infty}{\lim} \alpha_e^{\mathbb{S}} = 1-G(\widehat{k}).
\end{align*}

By Lemma \ref{lem_limitinfty_s}(i), $\underset{\bm{\sigma} \rightarrow +\infty}{\lim} \frac{A}{\sigma_v}=\frac{1}{1+\left[1-G(k_1)\right]\frac{\mu_z}{\mu}}$ and $\underset{\bm{\sigma} \rightarrow +\infty}{\lim} \alpha_e =1-G(k_1)$. We can verify that $\widehat{k}< k_1$ from equations~\eqref{equation_limit3} and~\eqref{equation_limit00}. Therefore $\frac{1}{1+\left[1-G(\widehat{k}) \right]\frac{\mu_z}{\mu}} < \frac{1}{1+\left[1-G(k_1)\right]\frac{\mu_z}{\mu}}$ and $1-G(\widehat{k}) > 1-G(k_1)$. That is, $\underset{\bm{\sigma} \rightarrow +\infty}{\lim} \frac{A^{\mathbb{S}}}{\sigma_v}< \underset{\bm{\sigma} \rightarrow +\infty}{\lim} \frac{A}{\sigma_v}$. We can easily verify that $\underset{\bm{\sigma} \rightarrow +\infty}{\lim} \alpha_e^{\mathbb{S}}< \underset{\bm{\sigma} \rightarrow +\infty}{\lim} \alpha_e+\alpha_d$.

By Lemma \ref{lem_limitinfty_s}(ii), $\underset{\bm{\sigma} \rightarrow +\infty}{\lim} \frac{A}{\sigma_v}=\frac{1-k_3}{1-k_3+\left[1-G(1) \right]\frac{\mu_z}{\mu}}$ and $\underset{\bm{\sigma} \rightarrow +\infty}{\lim} \alpha_e = 1-G(1)$. Then by equation~(\ref{equation_limit0}), $\frac{1-k_3}{1-k_3+\left[1-G(1) \right]\frac{\mu_z}{\mu}}=1-\mathbb{E}\left[\min\left\{1, \frac{Z^-}{\frac{k_3}{G(1)-G(2k_2-1)}+Z^+} \right\} \right]>1-\mathbb{E}\left[\min\left\{1, \frac{Z^-}{Z^+} \right\} \right]=1-R$. Since we suppose that $1-R> \frac{1}{1+[1-G(\widehat{k})]\frac{\mu_z}{\mu}}$, we have that $\frac{1}{1+[1-G(\widehat{k})]\frac{\mu_z}{\mu}}<\frac{1-k_3}{1-k_3+\left[1-G(1) \right]\frac{\mu_z}{\mu}}$, that is, $\underset{\bm{\sigma} \rightarrow +\infty}{\lim} \frac{A^{\mathbb{S}}}{\sigma_v}< \underset{\bm{\sigma} \rightarrow +\infty}{\lim} \frac{A}{\sigma_v}$.

Since $\widehat{k}<1, k_3>0$, we proved that $\underset{\bm{\sigma} \rightarrow +\infty}{\lim} \overline{\gamma_e}^{\mathbb{S}}-\underline{\gamma_e}^{\mathbb{S}}\leq \underset{\bm{\sigma} \rightarrow +\infty}{\lim} \overline{\gamma_e}-\underline{\gamma_e}, \underset{\bm{\sigma} \rightarrow +\infty}{\lim} \alpha_e^{\mathbb{S}}\geq \underset{\bm{\sigma} \rightarrow +\infty}{\lim} \alpha_e $.

Next we consider the case when $\bm\sigma\rightarrow 0^+$. Recall that when $\bm\sigma=0$, we have $\frac{A^{\mathbb{S}}}{\sigma_v}=  \frac{A}{\sigma_v}=0$. So we have to compare their derivatives at $0$.
From the proof of Lemma \ref{lem_limitzero_hats}, we have that $\underset{\bm{\sigma} \rightarrow 0^+}{\lim}\widehat{\mathbf{s}} = \underset{\bm{\sigma} \rightarrow 0^+}{\lim}\frac{2 \phi (\widehat{\mathbf{s}})}{2- 2\Phi(\widehat{\mathbf{s}})+\frac{\mu_z}{\mu}}$, and $\underset{\bm{\sigma} \rightarrow 0^+}{\lim} \mathbf{s_0}=\underset{\bm{\sigma} \rightarrow 0^+}{\lim}\frac{2\phi(\mathbf{s_1})}{2-2 \Phi(\mathbf{s_1})+\frac{\mu_z}{\mu}}.$ Since $\frac{2\phi(s)}{2-2 \Phi(s)+\frac{\mu_z}{\mu}}$ decreases in $s$, we show that either $ \underset{\bm{\sigma} \rightarrow 0^+}{\lim} \mathbf{s_0} < \underset{\bm{\sigma} \rightarrow 0^+}{\lim} \widehat{\mathbf{s}} < \underset{\bm{\sigma} \rightarrow 0^+}{\lim} \mathbf{s_1}$, or $ \underset{\bm{\sigma} \rightarrow 0^+}{\lim} \mathbf{s_0} = \underset{\bm{\sigma} \rightarrow 0^+}{\lim} \widehat{\mathbf{s}} = \underset{\bm{\sigma} \rightarrow 0^+}{\lim} \mathbf{s_1}$. Therefore we have two cases to consider.
(i) $ \underset{\bm{\sigma} \rightarrow 0^+}{\lim} \mathbf{s_0} < \underset{\bm{\sigma} \rightarrow 0^+}{\lim} \widehat{\mathbf{s}} < \underset{\bm{\sigma} \rightarrow 0^+}{\lim} \mathbf{s_1}$. Since $\frac{d\frac{A}{\sigma_v}}{d\bm{\sigma}}$ increases in $s$ when $\bm\sigma\rightarrow 0^+$, we have that $\frac{A^{\mathbb{S}}}{\sigma_v} < \frac{A}{\sigma_v}$, as $\bm\sigma\rightarrow 0^+$. (ii) $ \underset{\bm{\sigma} \rightarrow 0^+}{\lim} \mathbf{s_0} = \underset{\bm{\sigma} \rightarrow 0^+}{\lim} \widehat{\mathbf{s}} = \underset{\bm{\sigma} \rightarrow 0^+}{\lim} \mathbf{s_1}$. In this case $\frac{d\frac{A^{\mathbb{S}}}{\sigma_v}}{d\bm{\sigma}}=\frac{d\frac{A}{\sigma_v}}{d\bm{\sigma}}$, it is undetermined whether $\frac{A^{\mathbb{S}}}{\sigma_v} < \frac{A}{\sigma_v}$ or $\frac{A^\mathbb{S}}{\sigma_v} > \frac{A}{\sigma_v}$, as $\bm\sigma\rightarrow 0^+$. However, we cannot distinguish between case (i) and case (ii).

\subsection{Proof of Proposition \ref{prop_informativeness}}
As $\bm{\sigma}\rightarrow +\infty$, we have $\underset{\bm{\sigma} \rightarrow +\infty}{\lim} \overline{\gamma_e}^{\mathbb{S}}-\underline{\gamma_e}^{\mathbb{S}}=1$, and $\underset{\bm{\sigma} \rightarrow +\infty}{\lim} \alpha_e^{\mathbb{S}}=1-G(\hat k)$. We consider the two case in Lemma~\ref{lem_limitinfty_s}: (i) We have $\underset{\bm{\sigma} \rightarrow +\infty}{\lim} \overline{\gamma_e}-\underline{\gamma_e}=1$ and $\underset{\bm{\sigma} \rightarrow +\infty}{\lim} \alpha_e =1-G(k_1)$. Thus $\liminfty\frac{\overline{\gamma_e}^{\mathbb{S}}-\underline{\gamma_e}^{\mathbb{S}}}{\alpha_e^{\mathbb{S}}}\leq  \liminfty\frac{\overline{\gamma_e}-\underline{\gamma_e}}{\alpha_e}$ because $\hat k < k_1$. (ii) We have $\underset{\bm{\sigma} \rightarrow +\infty}{\lim} \overline{\gamma_e}-\underline{\gamma_e}=1-k_3$ and $\underset{\bm{\sigma} \rightarrow +\infty}{\lim} \alpha_e =1-G(1)$. From $\frac{1}{1+[1-G(\widehat{k})]\frac{\mu_z}{\mu}}<\frac{1-k_3}{1-k_3+\left[1-G(1) \right]\frac{\mu_z}{\mu}}$, we have $\frac{1}{1-G(\widehat{k})}<\frac{1-k_3}{1-G(1)}$, i.e., Thus $\liminfty\frac{\overline{\gamma_e}^{\mathbb{S}}-\underline{\gamma_e}^{\mathbb{S}}}{\alpha_e^{\mathbb{S}}}\leq  \liminfty\frac{\overline{\gamma_e}-\underline{\gamma_e}}{\alpha_e}$.



As $\bm{\sigma}\rightarrow 0^+$, by Lemma \ref{lem_HATSlessthanS}, we have $\mathbf{\widehat{s}}<\mathbf{s_1}$, $\forall \bm{\sigma}>0$. Since $\frac{\overline{\gamma_e}^{\mathbb{S}}-\underline{\gamma_e}^{\mathbb{S}}}{\alpha_e^{\mathbb{S}}}=
\frac{\Phi(\mathbf{\widehat{s}}+\bm{\sigma})-\Phi(\mathbf{\widehat{s}}-\bm{\sigma})}{1-G(2B(\mathbf{\widehat{s}})-1)\frac{}{}}$ and
$\frac{\overline{\gamma_e}-\underline{\gamma_e}}{\alpha_e}=
\frac{\Phi(\mathbf{s_1}+\bm{\sigma})-\Phi(\mathbf{s_1}-\bm{\sigma})}{1-G(2B(\mathbf{s_1})-1)}$, to show that $\frac{\overline{\gamma_e}^{\mathbb{S}}-\underline{\gamma_e}^{\mathbb{S}}}{\alpha_e^{\mathbb{S}}}> \frac{\overline{\gamma_e}-\underline{\gamma_e}}{\alpha_e}$ for small $\bm\sigma$, it is sufficient to show that there exists $\bm{\bar{\sigma}}>0$, s.t. $\forall \bm{\sigma}\in (0, \bm{\bar{\sigma}})$, $\frac{\Phi(s+\bm{\sigma})-\Phi(s-\bm{\sigma})}{1-G(2B(s)-1)}$ decreases in $s$. If $\frac{\mu_z}{\mu}<+\infty$, Lemma~\ref{lem_limitzero_hats} gives us $\underset{\bm{\sigma} \rightarrow 0^+}{\lim} \mathbf{s_1}\geq \underset{\bm{\sigma} \rightarrow 0^+}{\lim} \mathbf{\widehat{s}}>0$. Also, recall that $\underset{\bm{\sigma} \rightarrow 0^+}{\lim} \mathbf{s_1} \bm{\sigma}=\underset{\bm{\sigma} \rightarrow 0^+}{\lim} \mathbf{\widehat{s}} \bm{\sigma} =0$.

We now consider the derivative of $\frac{\Phi(s+\bm{\sigma})-\Phi(s-\bm{\sigma})}{1-G(2B(s)-1)}$ with respect to $s$.
\begin{align}
\frac{d\left( \frac{\Phi(s+\bm{\sigma})-\Phi(s-\bm{\sigma})}{1-G(2B(s)-1)}\right)}{d s} = &\frac{\left(\Phi(s+\bm{\sigma})-\Phi(s-\bm{\sigma}) \right)G'(2B(s)-1) \frac{e^{-2s\bm{\sigma}}}{\left(1+e^{-2s\bm{\sigma}}\right)^2} 4\bm{\sigma}}{\left(1-G(2B(s)-1)\right)^2}\nonumber\\
&+\frac{\left(\phi(s+\bm{\sigma})-\phi(s-\bm{\sigma})\right)(1-G(2B(s)-1))}{\left(1-G(2B(s)-1)\right)^2}.
\label{derivative_info}
\end{align}
Let $M=\underset{s\in[0,+\infty]}{\max}\left[\frac{4G'(2B(s)-1)}{1-G(2B(s)-1)}+1\right] <+\infty$. If $\underset{\bm{\sigma} \rightarrow 0^+}{\lim} s>0$ and  $\underset{\bm{\sigma} \rightarrow 0^+}{\lim} s \bm{\sigma} =0$,  there exists $\bm{\bar{\sigma}}>0$, such that $\forall \bm{\sigma}\in (0, \bm{\bar{\sigma}})$, $s>M e^2 \bm{\sigma}$ and $s \bm{\sigma} <1$. Therefore by the mean value theorem,
\begin{align*}
\frac{d\left( \frac{\Phi(s+\bm{\sigma})-\Phi(s-\bm{\sigma})}{1-G(2B(s)-1)}\right)}{d s} 
<&  \frac{2\bm{\sigma}\left[\phi(s-\bm{\sigma}) 4G'(2B(s)-1)\bm{\sigma}+\phi(s+\bm{\sigma})\left(1-G(2B(s)-1)\right)(-(s-\bm{\sigma}))\right]}{\left(1-G(2B(s)-1)\right)^2}\\
< & \frac{2\bm{\sigma}\left\{\phi(s-\bm{\sigma}) \left[\frac{4G'(2B(s)-1)}{1-G(2B(s)-1)}+1\right]\bm{\sigma}+\phi(s+\bm{\sigma})(-s)\right\}}{1-G(2B(s)-1)}\\
=& \frac{2\bm{\sigma}\left\{\phi(s+\bm{\sigma}) M e^{2s\bm{\sigma}}\bm{\sigma}-\phi(s+\bm{\sigma})s\right\}}{1-G(2B(s)-1)}\\
< & 0.
\end{align*}
Thus $\exists \bm{\bar{\sigma}}>0$, such that $\forall \bm{\sigma}\in (0, \bm{\bar{\sigma}}), s\in [\mathbf{\widehat{s}}, \mathbf{s_1}]$, we have $\frac{d\left( \frac{\Phi(s+\bm{\sigma})-\Phi(s-\bm{\sigma})}{1-G(2B(s)-1)}\right)}{d s}<0$.
Since $\mathbf{\widehat{s}}<\mathbf{s_1}$, we have $\frac{\Phi(\mathbf{\widehat{s}}+\bm{\sigma})-\Phi(\mathbf{\widehat{s}}-\bm{\sigma})}{1-G(2B(\mathbf{\widehat{s}})-1)}>
\frac{\Phi(\mathbf{s_1}+\bm{\sigma})-\Phi(\mathbf{s_1}-\bm{\sigma})}{1-G(2B(\mathbf{s_1})-1)}$. We proved the proposition.\qed

\subsection{Proof of Proposition \ref{prop_sigma_tendency}}
We need the follow Lemma to proceed the proof.
\begin{lem}
$\forall \bm{\sigma}>0, \exists C(\bm{\sigma})>0$, such that $\forall 0\leq s\leq C(\bm{\sigma}), \frac{d\left( \frac{\Phi(s+\bm{\sigma})-\Phi(s-\bm{\sigma})}{1-G(2B(s)-1)}\right)}{d s}>0$.
\label{lem_limit_muzmu}
\end{lem}
\begin{proof}
Consider (\ref{derivative_info}), $\forall \bm{\sigma}>0$, $\left.\frac{d\left( \frac{\Phi(s+\bm{\sigma})-\Phi(s-\bm{\sigma})}{1-G(2B(s)-1)}\right)}{d s}\right|_{s=0}>0$. Therefore, $\exists C(\bm{\sigma})>0$ such that $\forall 0\leq s\leq C(\bm{\sigma})$, we have
$$\frac{d\left( \frac{\Phi(s+\bm{\sigma})-\Phi(s-\bm{\sigma})}{1-G(2B(s)-1)}\right)}{d s}>0.$$
\end{proof}

First we show that $\bar{\sigma}_v=\underset{x>0}{\sup} \left\{x |\forall \sigma_v\in(0,x) ,~\frac{\overline{\gamma_e}^{\mathbb{S}}-\underline{\gamma_e}^{\mathbb{S}}}{\alpha_e^{\mathbb{S}}}>   \frac{\overline{\gamma_e}-\underline{\gamma_e}}{\alpha_e}\right\}$ is increasing in $\sigma_e$.
Because $\frac{\mu_z}{\mu}<+\infty$ and $\frac{\mu_z}{\mu}$ sufficiently large, according to Proposition~\ref{prop_informativeness}, there must exist a $\bm{\widehat{\sigma}}$ such that $\bm{\bar{\sigma}}=\underset{x>0}{\sup} \left\{x |\forall \bm{\sigma}\in(0,x) , \frac{\overline{\gamma_e}^{\mathbb{S}}-\underline{\gamma_e}^{\mathbb{S}}}{\alpha_e^{\mathbb{S}}}>   \frac{\overline{\gamma_e}-\underline{\gamma_e}}{\alpha_e}\right\}$. By definition $\bm{\bar{\sigma}} = \frac{\bar{\sigma}_v}{\sigma_e}$, i.e., $\bar{\sigma}_v=\bm{\bar{\sigma}} \sigma_e$, where $\bm{\bar{\sigma}}$ is a constant. $\bar{\sigma}_v$ is increasing in $\sigma_e$.
As $\sigma_e \rightarrow 0^+$, $\bar{\sigma}_v \rightarrow 0$ and as $\sigma_e \rightarrow +\infty$, $\bar{\sigma}_v \rightarrow +\infty$.

Next we prove that if $\frac{\mu_z}{\mu}$ is large enough, there exists a subsequence $\{(\frac{\mu_z}{\mu})_i\}$ such that $\bar{\sigma}_v$ decreases as $(\frac{\mu_z}{\mu})_i$ increases.

Let $C(\bm{\sigma})$ defined as $\sup\limits_x\left\{x\big|\forall s\in(0,x),~\frac{d\left( \frac{\Phi(s+\bm{\sigma})-\Phi(s-\bm{\sigma})}{1-G(2B(s)-1)}\right)}{d s}>0\right\}$. By Lemma \ref{lem_limit_muzmu}, such $C(\bm{\sigma})$ exists for all $\bm{\sigma}>0$. Note that if $\frac{\mu_z}{\mu} \rightarrow +\infty$, we have $\mathbf{\widehat{s}},~\mathbf{s_1} \rightarrow 0$. Therefore, as $\frac{\mu_z}{\mu}$ becomes sufficiently large, there exists $\bm{\sigma}(\frac{\mu_z}{\mu})$, such that
$$\mathbf{\widehat{s}},\mathbf{s_1}<C(\bm{\sigma}(\frac{\mu_z}{\mu})).$$
Thus, $\left.\frac{d\left( \frac{\Phi(s+\bm{\sigma})-\Phi(s-\bm{\sigma})}{1-G(2B(s)-1)}\right)}{d s}\right|_{s=\mathbf{\widehat{s}}}>0$, $\left.\frac{d\left( \frac{\Phi(s+\bm{\sigma})-\Phi(s-\bm{\sigma})}{1-G(2B(s)-1)}\right)}{d s}\right|_{s=\mathbf{s_1}}>0$.
And since $\mathbf{\widehat{s}}<\mathbf{s_1}$, $\frac{\overline{\gamma_e}^{\mathbb{S}}-\underline{\gamma_e}^{\mathbb{S}}}{\alpha_e^{\mathbb{S}}}<   \frac{\overline{\gamma_e}-\underline{\gamma_e}}{\alpha_e}$

This is to say, when $\frac{\mu_z}{\mu}$ is sufficiently large, we find a upper bound of $\bm{\bar{\sigma}}$, i.e., $\bm{\bar{\sigma}}<\bm{\sigma}(\frac{\mu_z}{\mu})$.

Therefore, there exists a subsequence $\{(\frac{\mu_z}{\mu})_i\}$ such that, as $(\frac{\mu_z}{\mu})_i$ increases, $\bm{\sigma}(\frac{\mu_z}{\mu})$ decreases, and $\bm{\bar{\sigma}}$ decreases, and as $(\frac{\mu_z}{\mu})_i \rightarrow +\infty$, $\bm{\sigma}(\frac{\mu_z}{\mu}) \rightarrow 0$, $\bm{\bar{\sigma}}\rightarrow 0$.

\end{document}